\documentclass[%
 amsmath,amssymb,
 aps,
]{revtex4-2}
\usepackage{amsmath}
\usepackage{amssymb,amsfonts,esint}
\usepackage{array}
\usepackage{graphicx}
\usepackage{dcolumn}
\usepackage{bm}
\usepackage{hyperref}
\usepackage{tikz}
\usepackage{quantikz}
\usepackage{svg}
\usepackage{algorithm2e}
\usepackage{bbold}
\usepackage{relsize}
\usepackage{amsthm}
\usepackage{comment}
\usepackage{braket}

\usepackage{cleveref}
\usepackage{xcolor}
\usepackage{hyperref}
\hypersetup{
    colorlinks=true, 
    linkcolor=blue,  
    citecolor=blue, 
    urlcolor=blue,   
    filecolor=magenta 
}

\usetikzlibrary{calc}

\usepackage[colorinlistoftodos, textsize=small, color=orange!50]{todonotes}

\newtheorem{theorem}{Theorem}
\newtheorem{lemma}[theorem]{Lemma}
\newtheorem{cor}[theorem]{Corollary}

\newtheorem{corollary}[theorem]{Corollary}

\newtheorem{prop}[theorem]{Proposition}

\newtheorem{definition}[theorem]{Definition}

\DeclareMathOperator{\real}{\mathbb{R}}
\newcommand{\cmplx}{\mathbb{C}}
\newcommand{\intg}{\mathbb{Z}}

\newcommand{\id}{\openone}

\newcommand{\Z}{\text{Z}}

\newcommand{\ketbra}[2]{| #1 \rangle\!\langle #2|}

\newcommand{\sel}{\text{SELECT}}
\newcommand{\prep}{\text{PREP}}
\newcommand{\hamt}{\text{HAM$-$T}}
\newcommand{\g}{\mathcal{C}}
\newcommand{\nconst}{\mathcal{A}}

\tikzset{
    gridstyle/.style={
        step=0.7cm,
        color=gray!40, 
        line width=0.4pt
    },
    error/.style={
        -stealth,         
        blue!80!black,    
        line width=1.2pt, 
        rounded corners=1.5pt 
    },
    labelstyle/.style={
        anchor=west,
        font=\small,
        fill=white,       
        inner sep=1.5pt   
    }
}

\begin{document}

\preprint{APS/123-QED}

\title{ First-Quantized Quantum Simulation of Non-Relativistic QED with Emergent Topologically Protected Coulomb Interactions}

\author{Torin F. Stetina}
\email{torin.stetina@xanadu.ai}
\affiliation{Simons Institute for the Theory of Computing, Berkeley, California, 94704, USA}
\affiliation{Berkeley Quantum Information and Computation Center, University of California, Berkeley, California, 94720, USA}
\affiliation{Xanadu, Toronto, ON, M5G2C8, Canada}

\author{Nathan Wiebe}
\email{nawiebe@cs.toronto.edu}
\affiliation{Department of Computer Science, University of Toronto, ON, Canada}
\affiliation{Pacific Northwest National Laboratory, Richland, WA, USA}
\affiliation{Canadian Institute for Advanced Research, Toronto, ON, Canada}


\begin{abstract}
We provide a simulation algorithm that properly addresses light matter interaction between non-relativistic first-quantized charged particles and quantum electromagnetic fields.  Unlike previous work, our Hamiltonian does not include an explicit Coulomb interaction between particles.  Rather, the Coulomb interaction emerges from the imposition of Gauss' law as a constraint upon the system in an appropriate non-relativistic limit.  Furthermore, a form of topological protection emerges in our formalism, analogous to that of the Toric code Hamiltonian. This mechanism prevents simulation-induced electric field errors that can be contracted to a point from causing any deviations from Coulomb's law in the non-relativistic limit and any error that forms a non-contractable loop is energetically dissallowed in the limit of large volume.  We find that, under appropriate continuity assumptions, the number of non-Clifford gates required by our algorithm scales in the thermodynamic limit as $\widetilde{O}(N^{2/3}\eta^{4/3} t \log^5(1/\epsilon))$ for $\eta$ particles, $N$ spatial grid points, simulation time $t$ and error tolerance $\epsilon$.  In comparison, the more specific problem of simulating the Coulomb interaction scales as $\widetilde{O}(N^{1/3} \eta^{8/3} t \log^2(1/\epsilon))$.  This suggests that if $N \in \tilde{o}(\eta^4)$ that our non-relativistic electrodynamic simulation method could provide a computational advantage for electronic structure problems in the thermodynamic limit under appropriate continuity assumptions as it obviates the need to compute the $O(\eta^2)$ pairwise interactions in the Coulomb Hamiltonian.

\end{abstract}

\maketitle


\section{Introduction}
First-quantized quantum simulation methods have fundamentally changed our understanding of scalable quantum algorithms by providing the best-known scaling for simulations of many-electron systems on quantum computers~\cite{babbush2019quantum, 2021_SBWetal, 2023_MSW,berry2024quantum,georges2025quantum,huggins2025efficient}. These methods assume that the system of interest consists of a fixed, finite number of distinguishable particles, with the state of each particle represented on a basis of size $N$. In contrast, second quantized methods track the occupation of indistinguishable particles in each of the $N$ basis states. While both representations describe the same physics, they are not computationally equivalent when it comes to quantum simulation. First-quantized approaches can utilize an exponentially large basis set size using only a polynomial number of qubits, in contrast to second quantization methods which typically require at least one qubit per basis state. In addition to the large improvement in space complexity in quantum simulation, first-quantized Hamiltonians also lead to asymptotic scalings that are sub-linear in the basis set size, $N$, for first-quantized approaches for applications in chemistry.

A significant issue that these methods face stems from the fact that the Coulomb Hamiltonian involves pairwise interactions between all charged particles in the simulation~\cite{babbush2019quantum}.  This means that the cost of (exactly) evaluating the Hamiltonian scales quadratically with the number of particles present in the simulation, and can lead to prohibitive costs in the thermodynamic limit.  The second issue is that existing treatments do not properly deal with the electromagnetic fields that the charged particles create as they move through space, and further, it does not address the fact that electromagnetic signals propagate at a finite speed between particles~\cite{2023_MSW}.  Our work directly addresses the latter and suggests that in some cases chemistry simulation may be accelerated by explicitly simulating an electromagnetic field.
\begin{table}[t]
    \centering
    \setlength{\tabcolsep}{7pt} 
    
    \renewcommand{\arraystretch}{1.6} 
    
    \begin{tabular}{|l|l|l|}
        \hline
        \textbf{Method} & \textbf{Cost} & \textbf{Model \& Description} \\
        \hline
        SBWRB21~\cite{su2021fault} & $\widetilde{O}\left(\frac{\eta^3 N^{1/3}t}{\Omega^{1/3}} \log^2(1/\epsilon)\right)$ & First-Quantized Coulomb Hamiltonian (Int. Pic.) \\
        
        MSW23~\cite{2023_MSW} & $\widetilde{O}\left(Nt(\eta N^{2/3}/\Omega^{2/3} + \Lambda^2) (\eta \log(1/\epsilon)+ N\log^2(\Lambda))\right)$ & Pauli-Fierz-Coulomb Simulation (Qubitization) \\
        
        MSW23~\cite{2023_MSW} & $\widetilde{O}\left(N^{2}t^{1+o(1)} (\eta N^{2/3}/\Omega^{2/3} + \Lambda^2)/\epsilon^{o(1)}\right)$ & Pauli-Fierz-Coulomb Simulation (Hybrid Trotter) \\
        
        \textbf{This Work} & $\widetilde{O}\left( \eta t\left(\frac{Mc^2}{\eta} + \frac{N^{2/3} }{\Omega^{2/3}}\right)M\log^2(\Lambda)\log^2(1/\epsilon) \right)$ & N.R. QED (Int. Pic.) \\
        \hline
    \end{tabular}
    \renewcommand{\arraystretch}{1}
    \caption{Summary of related results in first-quantized simulation.  Here $N$ is the number of spatial grid points, $M$ is the number of electromagnetic field modes, $\eta$ is the number of particles, $t$ is the simulation time and $\epsilon$ is the simulation error, $c$ is the speed of light and $\Lambda$ is the cutoff in the field.
    In most cases, the model for the electromagnetic interaction is different.  Thus caution needs to be exercised when explicitly comparing the computational complexity of simulating these models.}
    \label{tab:babbush}
\end{table}

In the low energy regime of many body charged particle simulations typically found in condensed matter and chemistry, fixing particle number (i.e. total number of electrons and nuclei) in the simulation is a reasonable choice. Of course, in quantum field theories, particle number is in fact not fixed, and particles are just excitations of the underlying quantum field. However, a similar challenge arises when introducing a quantum electromagnetic field, even in the non-relativistic case where the number of charged particles is fixed. The dynamic nature of the field and its interactions, means that the number of field quanta is not conserved. This variability is difficult to reconcile with a fixed-particle, first-quantized description of the electromagnetic field itself.
This problem has been addressed in recent work (including some of the authors) ~\cite{2023_MSW}, although the solution is currently restricted to non-relativistic particles with instantaneous Coulomb interactions interacting with an explicit external electromagnetic field. This model is known as the ``Pauli-Fierz-Coulomb Hamiltonian"~\cite{hiroshima2002self,2023_MSW}.  Their work addresses the challenge of particle conservation in electrodynamics by discretizing space into a set of field modes wherein each field mode acts as a distinguishable particle that tracks discrete values of the electromagnetic field within the formalism of first-quantization. 

Ideally, the effective Coulomb force law between charged particles in a system should not be added in by hand but should actually emerge from interaction with the electromagnetic field via Gauss’ law. In short, this means all particle interactions should be mediated through the electromagnetic field. This approach was not considered in the work of~\cite{2023_MSW} due to the use of the explicit Coulomb interaction term in the Hamiltonian. Additionally, in the original Pauli-Fierz model, which is the progenitor of the Pauli-Fierz-Coulomb model, gives a Hamiltonian that
precisely describes non-relativistic quantum electrodynamics (QED) for a single charged particle. The original model only works for a single particle because the Pauli-Fierz model doesn’t permit particles from generating their own electromagnetic field as Maxwell’s equations would require. This is not relevant in the single-particle non-relativistic limit as self interactions do not occur in such cases; however, in the multi-particle case, electric interactions between the particles will not be properly modeled. For this reason, the Pauli-Fierz-Coulomb Hamiltonian includes an explicit Coulomb interaction between all charged particles in the system while also permitting interactions with a full external electromagnetic field.

Our work goes beyond this by removing the pairwise Coulomb interaction between particles and including a term that locally enforces Gauss’ law in the Hamiltonian, causing the motion of the particles to generate an electric field in which the Coulomb interaction can emerge. Then by exploiting the interaction picture simulation method ~\cite{2018_LW}, we can simulate this
constrained Hamiltonian at low cost. We denote this as the ``constrained Pauli-Fierz Hamiltonian.'' As the interactions with the field are local in the model, we do not need to compute each of the pairwise interactions between all charged particles in the simulation. This leads to a many particle quantum simulation that captures non relativistic quantum electrodynamic effects that implicitly includes the Coulomb interaction.  We discuss the emergence of the Coulomb interaction in Section~\ref{sec:coulomb} and argue that a form of topological protection appears that causes the Coulomb interaction to appear as an emergent property in the non-relativistic thermodynamic limit after only requiring Gauss' law and an initial state that satisfies Coulomb's law.  While the focus of this work is to detail the simulation algorithm of the more general non-relativistic QED Hamiltonian in first-quantization, we can still compare to Coulomb Hamiltonian simulation approaches for condensed matter and chemistry. However, when comparing to the simpler Coulomb Hamiltonian with no explicit field, the algorithmic cost tradeoff between the two approaches requires a more nuanced discussion, which is addressed in Section~\ref{sec:compare}.

The paper is laid out as follows.  Section~\ref{sec:prelim} provides the necessary background on the Pauli-Fierz Hamiltonian and the Hilbert space that the non-relativistic QED Hamiltonian acts on.  Section~\ref{sec:algorithm}, provides background on the oracles assumed in our algorithm, as well as the interaction picture simulation method. Specifically, it contains our bounds on the gate and query complexities of the simulation algorithm, and discusses the lattice discretization effects that dictate the error in the emergence of Coulomb's law from Gauss' law. Section~\ref{sec:compare} discusses the difference in the gate complexity between our algorithm and the first-quantized plane wave simulation algorithm of Ref.~\cite{2021_SBWetal} in both the thermodynamic limit as well as the condensed phase.  We then conclude in Section~\ref{sec:conclusion} and discuss future research directions.

\section{Preliminaries}\label{sec:prelim}

The Pauli-Fierz Hamiltonian provides the correct expression for non-relativistic quantum electrodynamics for a single charged particle~\cite{hiroshima2002self}. This single particle restriction arises because the movement of particles naturally generate fields themselves, which other particles can interact with.  These interactions are not modeled in the Pauli-Fierz Hamiltonian which prevents it from being a complete model for non-relativistic quantum electrodynamics.   Our aim here is to derive an extension to the Pauli-Fierz Hamiltonian that allows back action on the field, with proper constraints that prevent self-interaction error.

First let us discuss what we mean by the non-relativistic limit in this case.  We specifically mean quantum electrodynamics in the limit where $c\gg 1$ where $c$ is the speed of light. In this limit, effects like pair production do not appear and the relativistic dispersion relation becomes the familiar expression.  Further, Ampere's law only appears implicitly here from Gauss' law and the Fourier duality of the electric and magnetic field operators.

The electric fields produced by the charged particles in the system appear via Gauss' law.  This law states that the electric flux through a closed surface is proportional to the enclosed charge, which yields the familiar Coulomb potential of a point charge. The key idea in this work is to then enforce Gauss's law via a sufficiently large energy penalty throughout the simulation.  This treats Gauss's law as a constraint, allowing effects like Coulomb's law to emerge naturally from the dynamics.

We define our Hilbert space to be of the following form.  We assume that our system consists of $\eta$ charged particles and $M$ regions specifying the field.  We represent the field using an electric field basis and the particles in a real-space position basis on a cubic lattice.  Specifically, the particle is represented by a tuple of three integers stored as three $n$-qubit registers corresponding to an evenly spaced one-dimensional grid in each Cartesian axis, such that the position operator of the $i^{\rm th}$ particle in position $\mu$ is of the following form

\begin{definition}\label{def:spatialHilbertSpace}
Let $\mathcal{H}_{i,\mu}$ for particle $i = \{ 0,1,\, \cdots, \eta-1 \}$  and direction $\mu= \{0,1,2\}$ be a complex Hilbert space of dimension $2^n$ wherein the standard basis is chosen such that the position operator is a linear operator in $\mathcal{H}_{i,\mu}$  (denoted $\hat{X}_{i,\mu} \in L(\mathcal{H}_{i,\mu})$)  is
$$
\hat{X}_{i,\mu} = \sum_{k=0}^{2^n-1} \frac{k \Omega^{1/3}}{2^n} \ketbra{k}{k}_{i,\mu}
$$
where $\Omega$ is the volume of the spatial grid for the particles.
\end{definition}
Throughout this work, we use the following notation unless otherwise noted, $\ket{x}\!\bra{x}_i = \openone \otimes \openone \otimes \cdots \otimes \ket{x}\!\bra{x} \otimes  \cdots \otimes \openone$, where $i$ labels the register of qubits that this operator acts upon, and identity, $\openone$, on all other registers, and additional subscripts such as $\ket{x}\!\bra{x}_{i,\mu}$ represent the $\mu$th subregister of register $i$. Additionally, we use $L(\cdot)$ to denote the set of linear operators acting on a given domain.

\begin{definition}
    Let $\mathcal{H'}_{q,\mu}$ be a complex Hilbert space for $q\in \mathbb{Z}_{M}$ and $\mu =\{0,1,2\}$ of dimension $2\Lambda$ such that in the standard basis the electric field operator, $E\in L(\mathcal{H'}_{q,\mu})$ acting on site $q$ in direction $\mu$ takes the form 
    $$
        E_{q,\mu} = \sum_{\epsilon =-\Lambda+1}^{\Lambda} \frac{\epsilon E_{\max}}{\Lambda} \ketbra{\epsilon}{\epsilon}_{q,\mu}.
    $$
    Here we assume that each field site $q$ comprises a cube of volume $\Omega/M$, where $M^{1/3} \in \mathbb{Z}$, denoted $D_{\mathbf{q}} \subset \mathbb{R}^3$ for $p\ne q$ $D_{\mathbf{q}} \cap D_{\mathbf{p}} = 0$ and that units are chosen such that the discretized field is dimensionless and has a maximum value of $\Lambda$. 
\end{definition}

\begin{definition}\label{def:HilbSpace}
    Let $\mathcal{H}$ be the overall Hilbert space for $\eta$ interacting particles and $M$ electric field sites distributed over a cubic grid as $$\mathcal{H} = \left(\bigotimes_{i=0}^{\eta-1}\left(\bigotimes_{\mu=0}^2\mathcal{H}_{i,\mu}\right)\right) \otimes \left(\bigotimes_{q=0}^{M-1} \left(\bigotimes_{\mu=0}^{2}\mathcal{H'}_{q,\mu} \right) \right)  $$
\end{definition}

The starting point for the Pauli-Fierz Hamiltonian is the Hamiltonian for the electromagnetic field.  
The continuum representation of the electromagnetic Hamiltonian is defined in the following manner
\begin{equation}
    H_{EM} = \int d^3 x \, \frac{\epsilon_0}{2} \mathbf{E}^2 + \frac{1}{2\mu_0} \mathbf{B}^2
\end{equation}
where $\epsilon_0$ is the permittivity of free space, and $\mu_0$ is the vacuum magnetic permeability, and both $\frac{\epsilon_0}{2} \mathbf{E}^2$ and $\frac{1}{2\mu_0} \mathbf{B}^2$ have units of energy density. Unless otherwise noted, bold variables indicate vectors throughout the rest of the manuscript. 
When we move to the quantized operator form of $\mathbf{E}$ and $\mathbf{B}$ in the quantum free electromagnetic field Hamiltonian on a discretized cubic lattice, we also have to multiply by the volume element, $h^3$, to recover the proper units of energy.  We will find it convenient to express the theory in terms of $\mathbf{A}$, which  we take to be the generator of displacement of displacements for the electric field for our discretization.  Specifically, an electric displacement of $ik E_{\max}$ acting on a given electric field eigenstate with electric field $E_{\max} k_0/\Lambda$ for integer $k,k_0$ can be expressed as
\begin{equation}
    e^{ik E_{\max}A/\Lambda}\ket{k_0 E_{\max}/\Lambda}=\ket{E_{\max}(((k+k_0 +\Lambda-1)\mod 2\Lambda )-\Lambda)/\Lambda},
\end{equation}
In the subsequent text, we will simplify this discussion by only referring to the electric field basis states by their encodings as an integer rather than the dimensionful field used above for exposition purposes.
 This implies that the final form of the free EM Hamiltonian in atomic units is
\begin{align}
    H_{EM} &= \frac{h^3}{2}\left( \sum_{q,\mu} \epsilon_0 E_{q,\mu}^2 + \sum_{q,\mu} \frac{1}{\mu_0} B_{q,\mu}^2 \right) \\
    &=\frac{h^3}{2}\left( \sum_{q,\mu} \epsilon_0 E_{q,\mu}^2 + \sum_{q,\mu} \frac{1}{\mu_0}([\nabla \times A]_{q,\mu})^2 \right) \\
    &=\frac{h^3}{2}\left( \sum_{q,\mu} \frac{1}{4 \pi} E_{q,\mu}^2 + \sum_{q,\mu} \frac{c^2}{4\pi}([\nabla \times A]_{q,\mu})^2 \right) \\
    &=\frac{h^3}{8\pi}\left( \sum_{q,\mu}  E_{q,\mu}^2 + \sum_{q,\mu} c^2 ([\nabla \times A]_{q,\mu})^2 \right)
\end{align}
where above we use the fact that in atomic units $4 \pi \epsilon_0 = 1$ and the identity $\epsilon_0 = \frac{1}{\mu_0 c^2}$ to solve for the magnetic permeability $\mu_0$ in atomic units, and $c$ is the speed of light. More concretely

\begin{equation}
    H_{EM} = \frac{h^3}{8\pi}\left( \openone\otimes\sum_{q,\mu} E_{q,\mu}^2 + \sum_{q,\mu} c^2 \openone\otimes B_{q,\mu}^2 \right)=\frac{h^3}{8\pi}\left( \openone\otimes\sum_{q,\mu} E_{q,\mu}^2 + \sum_{q,\mu} c^2 (\openone\otimes [\nabla \times A]_{q,\mu})^2 \right),
\end{equation}
where $\nabla$ is the discrete approximation to the spatial gradient operator centered at point $\mu$.

Before defining our full Hamiltonian we need to also define several other operators that will be useful in our construction.  In particular, we will need to define the vector potential operator and the displacement operator for electric field states.  Both of these quantities are defined below.
\begin{definition}
    Let $U_{q,\mu}$, $\Pi_{p,i}$ for $q\in \mathbb{Z}_M$, $\mu,\nu \in \{0,1,2\}$ and $p\in \mathbb{Z}_N$ be linear operators acting on $\mathcal{H}'_{q,\mu}$ and $\mathcal{H}$ respectively and let $e_\nu$ correspond to the unit displacement along direction $\nu$ of the cubic lattice.  The action of these  operators on their the corresponding Hilbert spaces are specified below
    \begin{align*}
         U_{q,\mu}\in L(\mathcal{H}'_{q,\mu}) &:= \sum_{\epsilon=-\Lambda}^{\Lambda+1} |\epsilon +1 \rangle \langle \epsilon |_{q,\mu}\nonumber\\
         A_{q,\mu} \in L(\mathcal{H}'_{q,\mu}) &:=\frac{-i\Lambda}{E_{\max}} \log(U_{q,\mu}) \nonumber\\
             \Pi_{q,i}\in L(\mathcal{H}) &= \sum_{\mathbf{x} \in D_q} \left( \ketbra{x_0}{x_0}_{i,0} \otimes \ketbra{x_1}{x_1}_{i,1} \otimes \ketbra{x_2}{x_2}_{i,2} \right) \otimes \openone  .
    \end{align*}
where throughout this manuscript unless otherwise noted, $\mathbf{x} := (x_0,x_1,x_2)$.
\label{dfn:opDfn}
\end{definition}

The following corollary provides a way, for a field cutoff that is a power of $2$, to simply diagonalize the field operator.  This will prove invaluable below for our implementation of the simulation.
\begin{corollary}[Corollary 1 of Ref.~\cite{2023_MSW}]
    Let $\Lambda$ be a power of $2$ then the operator $A_{q,\mu}$ can be written as
    $$
    A_{q,\mu} =\frac{\pi}{E_{\max}} \left(\frac{2\Lambda-1}{2}\openone - \sum_{i=0}^{\log(2\Lambda)-1}\mathcal{F}\left(2^{i - 1} Z_{(i+1,q,\mu)} \right)\mathcal{F}^\dagger \right)
    $$
    where $\mathcal{F}$ is the fourier transform operator, $Z_{i,q,\mu}$ is the Pauli-$Z$ operator acting on qubit $i$ of the field register for the field at site $q$ in direction $\mu$. 
\end{corollary}

Next, in order to define our overall Hamiltonian we need to discuss the role that Gauss' law plays in the overall simulation.  We enforce Gauss' law through an energy penalty term in the Hamiltonian that forbids any field configuration that does not automatically satisfy the appropriate constraint. The simple differential form of Gauss' Law in continuous space is 
\begin{equation}
    \nabla \cdot \bold{E} = \rho / \epsilon_{0}
\end{equation}
where $\bold{E} = (E_x, E_y, E_z)$ is the electric field vector, $\rho$ is the scalar charge density, and $\epsilon_0$ is the permittivity of free space. Since $4 \pi \epsilon_0 = 1$ in atomic units, and expanding the differential operator we get the following
\begin{equation}
    \frac{\partial E_x}{\partial x} + \frac{\partial E_y}{\partial y} + \frac{\partial E_z}{\partial z} = 4 \pi \rho.
\end{equation}
Now, we must discretize over our periodic 3D cubic lattice of $M$ points for the field, with a spacing of 
\begin{equation}
    h = \Omega^{1/3} / M^{1/3}
\end{equation} where $\Omega$ is the cell volume. We can now describe the equation for Gauss' law at a single lattice point, using the unbolded notation, $q=(q_x,q_y,q_z)$ for simplicity. By enumerating the fundamental lattice vectors $e_{1} = \hat{x}, e_2 = \hat{y}, e_3=\hat{z}$, we have that a centered difference approximation to the operator can be formed via:
\begin{equation}
    \frac{1}{2 h}  \sum_{\mu=0}^2 (E_{q+e_\mu,\mu} - E_{q-e_\mu, \mu} ) = \frac{4 \pi}{h^3} \sum_{i=0}^{\eta-1} \zeta_i \Pi_{q,i}
\end{equation}
where we use the central finite difference formula for the differential operators on the left-hand side, and note that $\Pi_{q,i}$ is a projection operator that counts the number of charges present inside the cube centered at $q$, scaled by the charge of the $i$th particle $\zeta_i$. This formula gives a minimal approximation to Gauss' law, but it has a significant problem in that the discrete approximation to the flux implicitly assumes that the field doesn't substantially vary over the cubic region considered.  This is especially problematic because the derivatives of the field diverge for a box that contains a charge.  This makes the finite approximations to the surface integrals unreliable for a small example like this for any region that includes charge.

Thus we can express a constraint Hamiltonian that provides a unit energy penalty that violates Gauss' law can be given by
\begin{equation}
     H_c = \left(\openone - \prod_{q=0}^{M-1}\!{\rm rect}\left({\frac{h^2}{2}\sum_{\mu=0}^2 \openone\otimes(E_{q+e_\mu,\mu} - E_{q-e_\mu, \mu} ) - 4\pi \sum_{i=0}^{\eta-1} \zeta_i \Pi_{q,i}}\right) \right)
\end{equation}
This justifies our statement of Gauss' law and shows why microscopically enforcing the penalty implies that Gauss' law will be held over the entire space.

We will see that this simple discretization alone is sufficient to obtain the Coulomb interaction asymptotically from Gauss' law.  However, a major drawback of this discretization is that the size of the grid spacing needed to ensure a good approximation to the flux integral can be potentially prohibitive.  We can address this by considering a larger cube of length $(2b+1)h$ for integer value $b$ for the computation of the flux.  We can then estimate the flux through the cube by using a quadrature formula such as the $2b+1$-point Newton-Cotes formula on each of the points in question.  
We then need a way to estimate the field inside a fixed region using a discrete formula.  This is problematic in general because the number of electromagnetic field grids that would be required to represent the field would be large.  Formally, we consider the use of higher-order 
\begin{lemma}\label{lem:2DCotes}
    Let $E$ be an electromagnetic field that has a derivative of at most $\Gamma$ on $\mathcal{C}=\{z : z=bh(1+3e^{i\phi}+3-e^{-i\phi})\}$ centered on any cube $S'$ within the three-dimensional lattice consisting of $(2b+1)^2-1$ cubes of volume $h^3$ centered at $x=0$, with normal vector $\hat{n}$.  We then have that for any  $\epsilon>0$ there exist $\beta_{x}$ such that 
    \begin{enumerate}
        \item $|\sum_{x\in S'} \beta_{x} E(x)\cdot\hat{n}(x) - \oiint_S E(x)\cdot dx |\le \epsilon$ for $b\in O(\log(\Gamma/\epsilon))$.
        \item $\sum_{x\in S'} |\beta_{x}| \in \widetilde{O}(bh\log( \Gamma/\epsilon))$.
    \end{enumerate}
\end{lemma}
Proof follows straight forwardly from standard bounds on the remainder error for Newton-Cotes formulas.

If we define this domain surrounding point $q$ to be $\mathcal{D}_{q;b}$ and define $\mathcal{S}_{q,\mu;b}$ to be the set of points on the boundary of $\mathcal{D}_{q;b}$ that has normal vectors pointing in the $\pm \mu$ direction,  we can then show that the expansion coefficients $\beta_{u,v,\mu}$ correspond to the points with coordinates $(u,v)$ on the boundary.  The more general expression for the Gauss' law constraint over a cubic mesh is then given below.

\begin{definition}\label{def:constraint}
Let $S_{q,\mu;b}$ be the set of lattice points in the plane orthogonal to normal vector $e_{\mu}$ for $\mu = \{0,1,2\}$ on surface of a cubic lattice of volume $(2b+1)^3$ centered at lattice site $q\in \mathbb{Z}_M$ and let $\mathcal{D}_{q;b}$  be the set of all $(2b+1)^3$ points within this cube.  We then define the constraint Hamiltonian for weights $\beta: \left(\bigcup_{\mu=0}^2 S_{q,\mu;b}\right) \mapsto \mathbb{R}$ to be

$$
H_c := \left(\openone - \prod_{q=0}^{M-1}\!{\rm rect}\left({\frac{h^2}{2}\sum_{\mu=0}^2 \sum_{u \in \mathcal{S}_{q,\mu;b}}\beta_{u} \openone\otimes(E_{q+be_{\mu}+ u,\mu} - E_{q-be_\mu+ u,\mu} ) - 4\pi \sum_{q'\in \mathcal{D}_{q;b}}\sum_{i=0}^{\eta-1} \zeta_i \Pi_{q',i}}\right) \right)
$$
\end{definition}

In the case of the Newton-Cotes formulas,  the coefficients $\beta_{u}$ obey $\sum_{u} |\beta_{u}|\in O(b^2)$ and the error is on the order of $e^{-O(b)}$ for sufficiently smooth integrands.   This will allow us to make the discretization error in flux calculations for Gauss' law arbitrarily small by increasing $b$ rather than shrinking $h$, which we will see plays an important role in ensuring that the Coulomb potential emerges in an inexpensive way from Gauss' law.

With these definitions in mind we can now state our final Hamiltonian for $\eta$ interacting non-relativistic charged particles as follows. 

\begin{definition}[Constrained Pauli-Fierz Hamiltonian]\label{def:Ham} Let $\lambda>0$ be a penalty strength for the constraint Hamiltonian $H_c$.  The Hamiltonian $H \in L(\mathcal{H})$ in the Coulomb Gauge (which consists of the choice $\nabla \cdot A =0$) is defined to be a Hermitian operator on the $M$ field locations as well as the $\eta$ particles supported on a spatial grid of size $N$, such that each $i$th particle has electric charge $\zeta_i \in \mathbb{Z}\setminus 0$
    $$
H := H_f + \lambda H_c
    $$
where $H_c$ is the constraint Hamiltonian of Definition~\ref{def:constraint} and the ``free'' Hamiltonian $H_f$ is the Pauli-Fierz Hamiltonian:
$$
H_f:= \sum_{i=0}^{\eta-1} \frac{1}{2m_i} \sum_{\mu=0}^2  \left(  -i \nabla_{i,\mu}\otimes \openone  -   \zeta_i \sum_{q=0}^{M-1} \openone\otimes\delta_{x\in D_q} A_{q,\mu} \right)^2 + \frac{h^3}{8\pi}\sum_{q=0}^{M-1}\sum_{\mu=0}^2   \openone\otimes E_{q,\mu}^2  +\frac{c^2 h^3}{8 \pi}\sum_{q=0}^{M-1}  \openone\otimes |\nabla_q \times A_{q}|^2
$$
where $m_i$ is the $i$th particle mass. We further take the components of the gradient operator at site $\mu$, $\nabla_\mu$, to be given by 
 the approximate derivative operator $\nabla_{i,\mu}$ to be the $2a+1$ point finite difference formula
$$
    \nabla_{i,\mu}\psi(x)=\frac{1}{\Delta}\sum_{k=-a}^ad_{2a+1,k}'\psi(x+k\Delta e_{i,\mu})
 $$
 where $\hat{e}_{i,\mu}$ is the unit vector along the $\mu^{th}$ component of particle $i$ of $x$, $(x_{\mu}+k\Delta e_{i,\mu})$ is evaluated modulo the grid length $\Omega^{1/3}$ (and hence $\Delta = \Omega^{1/3} / N^{1/3}$) and
 $$
    d_{2a+1,k}'=\frac{(-1)^{k+1}(a!)^2}{j(a-k)!(a+k)!}\qquad d_{2a+1,0}'=0,
 $$
 and the curl of the discrete vector potential is
 $$
 [\nabla_q \times A_q]_\ell = \sum_{\mu,\nu,\ell} \varepsilon_{\mu\nu\ell} \nabla_{q,\mu} A_{q,\nu} = \sum_{\mu,\nu,\ell}\frac{1}{h} \sum_{k=-a}^a d'_{2a+1} A_{q + k e_\mu, \nu} \varepsilon_{\mu\nu\ell}
 $$
 where $\varepsilon_{\mu\nu\ell}$ is the three dimensional Levi-Civita symbol. Further, for convenience assume that in our units $\max_i |\zeta_i|/m_i \le 1$ and throughout the paper we assume that $a$ is used for the order of the divided difference formula over both the field and particle grids.
\end{definition}

In practice, quantum simulation of the overall Hamiltonian in \Cref{def:Ham} can be prohibitively expensive in the limit of large $\lambda$, which is needed in order to ensure that the quantum dynamics remains within the feasible region for the Gauss' law constraint.  We address this by transforming into an interaction frame of the constraint Hamiltonian and all terms that commute with the constraint.  This interaction Hamiltonian is given in the following definition.
\begin{definition}[Interaction Hamiltonian for Constrained Pauli-Fierz]\label{def:HamInt}
Under the assumptions of~\Cref{def:Ham} we define the interaction picture version of the Hamiltonian in the frame of the constraint Hamiltonian and the Electric term to be $H_{\rm int}: \mathbb{R} \mapsto L(\mathcal{H})$ via
    \begin{align*}
H _{\rm int}(t)&=   \sum_{i=0}^{\eta-1} \frac{1}{2m_i} \sum_{\mu=0}^2  e^{i (\lambda  H_c + \frac{h^3}{8 \pi}\sum_{q,\mu} \openone\otimes E_{q,\mu}^2)t}\left(  -i \nabla_{i,\mu}\otimes \openone  -   \zeta_i \sum_{q=0}^{M-1} \openone\otimes\delta_{x\in D_q} A_{q,\mu} \right)^2e^{-i (\lambda  H_c + \frac{h^3}{8 \pi}\sum_{q,\mu} \openone\otimes E_{q,\mu}^2)t} \nonumber\\
&\quad + \frac{c^2 h^3}{8 \pi} \sum_{q=0}^{M-1}\sum_{\mu=0}^2   e^{i(\lambda  H_c + \frac{h^3}{8 \pi}\sum_{q,\mu} \openone\otimes E_{q,\mu}^2)t}\left(|\nabla_{\mu}\times A_\mu |^2\right)e^{-i (\lambda  H_c + \frac{h^3}{8 \pi}\sum_{q,\mu} \openone\otimes E_{q,\mu}^2)t }
    \end{align*}
\end{definition}

One of the major challenges that arises with our formulation of the Hamiltonian, is that it is given using the Coulomb gauge.  The key issue is that our discrete Hamiltonian does not necessarily keep the dynamics within the desired gauge sector.  While this could be addressed by including a penalty term, similar to our approach for imposing Gauss' law, we see that this is not in fact necessary for a broad class of sufficiently smooth quantum distributions, as shown in lemma \ref{lem:coulombGauge} of Appendix~\ref{app:cgauge}.

\section{Simulation Using Interaction Picture Algorithm}\label{sec:algorithm}
Our algorithm for simulating the quantum dynamics of the constrained Pauli-Fierz Hamiltonian is straight forward.  We simply simulate ${\mathcal{T}} e^{-i \int_0^t H_{\rm int}(s) ds}$ within error $\epsilon$ using the truncated Dyson series simulation algorithm~\cite{2018_LW, kieferova2019simulating}.  We use this approach because it provides near optimal scaling with $t$ and further requires a number of queries to a block encoding of $H_f$ that is independent of $\lambda$~\cite{2018_LW} (although the gate complexity depends logarithmically on $\lambda$).  This means that we can use enormous values of the constraint without substantially affecting the cost. 

To properly execute the simulation algorithm, we must determine the interaction strength, $\lambda$, required to ensure that the error in the quantum dynamics is at most $\epsilon$ for a finite evolution time, $t$. Such a bound has been already been given in Ref.~\cite{2022_RRW}, and we state this result below.
\begin{lemma}[Lemma 7.1 of Ref.~\cite{2022_RRW}]\label{lem:constraint}
Let $H=H_0 + \lambda H_c$ where Hermitian $H_0, H_c \in L(\mathcal{H})$ for $\lambda>0$ be a variable describing the strength of the constraint such that $\|H_f\|_\infty \ll \lambda$ and $H_c \succeq 0$, such that the dimension of the kernel of $H_c$ is at least $1$ and the second smallest eigenvalue is at least $1$.  We then have that for any $\ket{\psi}$ in the kernel of $H_c$, 
$$ \|e^{-i (H_0 + \lambda H_c)t} \ket{\psi} - \lim_{\lambda \rightarrow \infty} e^{-i (H_0 + \lambda H_c)t} \ket{\psi}  \|_2 \le \epsilon$$ for any $\epsilon>0$ can be achieved using a choice of the penalty strength that scales a as $\lambda \in \Theta(\|H_0\|_\infty^2 t/\epsilon)$.
\end{lemma}
Therefore, we see that the rapid growth of the interaction strength as $\epsilon\rightarrow 0$ necessitates the use of techniques, such as the interaction picture simulation method, to achieve a poly-logarithmic scaling in $1/\epsilon$.

Our simulation algorithm for the constrained Pauli-Fierz Hamiltonian utilizes the interaction picture simulation method of Ref.~\cite{2018_LW}.
The idea behind the algorithm of Ref.~\cite{2018_LW} involves applying a truncated Dyson series expansion to simulate a short time step for a quantum dynamical system, and then use oblivious amplitude amplification to boost the probability of success to nearly $1$. The Dyson series can be derived by iterating the integral form of the Schr\"{o}dinger equation, $U(t) = \id -i \int_0^t H_{\rm int}(s) U(s) ds$ which yields the series
\begin{equation}
    U(t) = \id -i \int_0^t H_{\rm int}(s) ds - \int_0^t \int_0^{s_1} H_{\rm int}(s_1) H_{\rm int}(s_2) d^2s +\cdots
\end{equation}
We can then implement this propagator through a linear combinations of unitaries (LCU).  In order to construct this operator we need to provide a method for both block encoding the time-dependent Hamiltonian, and also time ordering the resulting product.  We refer to the oracle that block encodes the time-dependent Hamiltonian as $\hamt$, and is given in the following lemma along with the cost of using it to implement the Dyson series.

For the case of time-independent Hamiltonians, it is clear that the Dyson series reduces to the ordinary Taylor series expansion.  Our goal then is to implement this expansion of the operator and examine the costs of the time-ordering in an LCU that block-encodes this operator.
The fundamental oracle that we need to construct provides a block-encoding of $H_{\rm int}(t)$.  However, as the truncated Dyson series is an expansion in both time as well as order, 
we use the following query complexity result of Hamiltonian simulation in the interaction picture.

\begin{lemma}[Interaction Picture Simulation Method (Lemma 6 of Ref.~\cite{2018_LW})]
Let $A\in L(\mathcal{H})$, $B\in L(\mathcal{H})$, let $\alpha_A$ and $\alpha_B$ be known constants such that $\|A\|\leq\alpha_A$ and $\|B\|\leq\alpha_B$ and let $\mathcal{L}=2^{n_t}$ be the maximum number of discrete times considered in the simulation. Assume the existence of a unitary oracle that implements the Hamiltonian within the interaction picture, denoted $\hamt\in L(\cmplx^{\mathcal{L}}\otimes \mathcal{H}) $ such that
\begin{eqnarray}
    \left(\bra{0}_a\otimes\id_s\right)\hamt\left(\ket{0}_a\otimes\id_s\right) = \sum_{k=0}^{\mathcal{L}-1}\ket{k}\!\bra{k}\otimes\frac{e^{iA\tau k/\mathcal{L}}Be^{-iA\tau k/\mathcal{L}}}{\alpha_B}.
\end{eqnarray}
The time-evolution operator $e^{-i(A+B)t}$ may be approximated to error $\epsilon$ in the operator norm with the following cost.
\begin{enumerate}
    \item Simulations of $e^{-iA\tau}$ : $O(\alpha_Bt)$,

    \item Queries to $\hamt$ : $O\left(\alpha_Bt\frac{\log (\alpha_Bt/\epsilon)}{\log\log (\alpha_Bt/\epsilon) }  \right)$,

    \item Qubits : $n_s+O\left(n_a+\log\left(\frac{t}{\epsilon}(\alpha_A+\alpha_B) \right) \right)$,

    \item Primitive gates : $O\left(\alpha_Bt\left( n_a+\log\left(\frac{t}{\epsilon}(\alpha_A+\alpha_B) \right) \right) \frac{\log(\alpha_Bt/\epsilon)}{\log\log(\alpha_Bt/\epsilon)} \right)$.
\end{enumerate}
\label{lem:2018LW}
\end{lemma}
The key point behind this method is that the cost of simulation depends only logarithmically on $\alpha_A$, provided that the cost of simulating $e^{-iAt}$ is subdominant. This is typically true when $A$ is diagonal or efficiently diagonalizable in the computational basis. However, advantages can also be observed whenever the evolution operator is fast-forwardable, which can occur for other families of Hamiltonians as well~\cite{zlokapa2024hamiltonian}. 
In our case, the constraint Hamiltonian will have a very large coefficient $1$-norm associated with it so for that reason we go into the interaction frame of that term.

 Let us define parts of the Hamiltonian $H$ (Definition \ref{def:Ham}) as follows.  Let $\zeta_i$ be the charge of the $i$th particle in the system and let $A_{q,\mu}$ be the vector potential operator acting on site $q$ in direction $\mu$.  Let $H_{f2}$ be the part of the electromagnetic Hamiltonian that does not commute with the Gauss' law term, and with these definitions we can write
\begin{equation}
    H_{\rm int}(t):= e^{i (\lambda  H_c + H_{f1})t}(H_{\pi} + H_{f2})e^{-i (\lambda  H_c + H_{f1})t},
\end{equation}
where
\begin{eqnarray}
    H_{\pi} &=&   \sum_{j=0}^{\eta-1} \frac{1}{2m_i}  \sum_{\mu=0}^2  \left( -  i\nabla_{j,\mu}\otimes\id  - \zeta_i \sum_{q=0}^{M-1}  \id \otimes A_{q,\mu} \delta_{x_j \in D_{q}} \right)^2   \nonumber \\
    &=&\sum_{j=0}^{\eta-1} \frac{1}{2m_i} \sum_{\mu=0}^2  \left( -  \nabla^2_{j,\mu}\otimes\id  - \zeta_i \sum_{q=0}^{M-1} 2i \nabla_{j,\mu} \otimes A_{q,\mu} \delta_{x_j \in D_{q}} + \zeta_i^2 \sum_{q=0}^{M-1}\id\otimes A_{q,\mu}^2    \delta_{x_j \in D_{q}}\right) \label{eqn:Hpi}    \\
    H_{f2}&=& \frac{c^2 h^3}{8 \pi}\sum_{q=0}^{M-1}  \openone\otimes |\nabla_q \times A_{q}|^2= \frac{c^2 h^3}{8 \pi}\sum_{q=0}^{M-1}  \openone\otimes \left(\sum_{ijk} \varepsilon_{ijk} \nabla_{q,j} A_{q,k}\right)^2\nonumber\\
    &:=&  \frac{c^2 h}{8 \pi}\sum_{q=0}^{M-1}  \openone\otimes\left(\sum_{i,j,k} \sum_{p=-a}^a d'_{2a+1} A_{q + p e_j, k} \varepsilon_{ijk}\right)^2
\end{eqnarray}
In order to implement $\hamt$ we need to construct a block encoding of $H_{\pi}+H_{f2}$ and then conjugate it with $e^{-i(\lambda H_c + \frac{h^3}{8\pi}\sum_{q,\mu} \id \otimes E_{p\mu}^2/2 )t}$.  For notational simplicity, we define the field Hamiltonian to be $H_{f1} + H_{f2}$ where 
\begin{equation}
    H_{f1}:= \frac{h^3}{8\pi}\sum_{q,\mu} \id \otimes E_{q,\mu}^2
\end{equation}
and 
\begin{equation}
    \|H_{f1}\|_\infty \in O(M\Lambda^2) \label{eq:Hf1Bd}
\end{equation}
Next, we discuss the cost of implementing the interaction picture transformation, as discussed in the following lemma.

\begin{lemma}\label{lem:VintImp}
    Let $\lambda \in \mathbb{R}$, an operator $V_{\rm int}\in L(\mathbb{C}^{2^{n_t}} \otimes \mathcal{H})$ where Hilbert space $\mathcal{H}$ is given in \Cref{def:HilbSpace},  and let $n_t=\log(\mathcal{L})$ be the number of qubits used to represent the time for a maximum simulation time of $T_{\max}$.  A unitary operator $V_{\rm int}$ can be constructed such that the conditional time evolution of the constraint term is approximated as
    $$
    \left\|V_{\rm int}- \sum_{k=0}^{\mathcal{L}-1} \ketbra{k}{k}\otimes e^{-i (\lambda H_c + \sum_{q,\mu} \id \otimes E_{q,\mu}^2 )k (T_{\max}/\mathcal{L})}\right\| \le \epsilon,
    $$
    using a number of gate operations drawn from the Clifford and $T$-gate library that scales with the number of particles, $\eta$, the number of electric field cells $M$ and the number of position grid points $N$ and using a cube of sidelength $2b+1$ electric field sites that is in
    $$
    \mathcal{C}(V_{\rm int}) = \widetilde{O}( b^3(M+\sum_i |\zeta_i|)\log^2(N\Lambda) + n_t\log(n_t/\epsilon))
    $$
\end{lemma}

The proof of this lemma is perhaps the most technical aspect of our work.  It involves the use of quantum merge sort to inexpensively re-arrange quantum arrays that represent the charges and the fields so that the Gauss' law checks can be performed without doing a brute force search over all field cells and all charges to see whether a charge is present in a given field cell.  This reduces the cost from one that scales as $O(M\eta)$ to $O(M+\eta)$.  The proof is provided in Appendix~\ref{app:vintImp}.

Now we discuss the cost of block encoding $H_{\pi}+H_{f2}$. For this, we require the decomposition of the operators $A_{q,\mu}$, $A_{q,\mu}^2$, $\nabla$, $\nabla^2$, $E$, $E^2$ into a linear combination of unitaries. This can be obtained from the results of Ref.~\cite{2023_MSW}.   We include a summary of these results in Appendix \ref{app:lcu} for clarity. We use the results from Ref.~\cite{2023_MSW} for efficient recursive block encoding that help in optimizing the resources below. For completeness, we provide a definition of block encoding and the LCU method below.
\begin{definition}[Block encoding \cite{2019_GSLW}]
 Suppose $A$ is an $n$-qubit operator, $\epsilon\in\real_{+}$ and $m\in\mathbb{N}$. We then say that the $(m+n)$-qubit unitary $U_A$ is an $(\alpha,m,\epsilon)$-block-encoding of $A$ if 
 \begin{eqnarray}
  \|A-\alpha\left(\bra{S}\otimes\id_n\right)U_A\left(\ket{S}\otimes\id_n\right)\|_{\infty}\leq\epsilon,
 \end{eqnarray}
where $\ket{S}$ is an $m$-qubit state.
\end{definition}
 We will often drop the second argument and write `$(\lambda,-,\epsilon)$-block-encoding of $A$', because we focus on the gate complexity and the second argument only captures the extra ancilla needed in the block encoding. Often, even for more brevity we write `block-encoding of $\frac{A}{\lambda}$' to express the case where $\epsilon=0$. Suppose we have a Hamiltonian $H_i$ expressed as a LCU, i.e. $H_i=\sum_{j=1}^{M_i}h_{ij}U_{ij}$, such that $\lambda_i=\sum_j|h_{ij}|$. 
Further, without loss of generality, we can assume that $h_{ij} \ge 0$ by absorbing phases into the unitaries.
In this case, we can have a $\left(\lambda_i,\log M_i,0\right)$-block encoding of $H_i$ using an ancilla preparation subroutine and a unitary selection subroutine, which we denote by $\prep_i$ and $\sel_i$ respectively.
\begin{eqnarray}
    \prep_i\ket{0}^{\log M_i}&=&\sum_{j=1}^{M_i}\sqrt{\frac{h_{ij}}{\lambda_i}}\ket{j}   \label{eqn:prepi} \\
 \sel_i&=&\sum_{j=1}^{M_i}\ket{j}\bra{j}\otimes U_{ij}   \label{eqn:seli} 
\end{eqnarray}
The standard method for implementing a block encoding is given by the LCU lemma from Ref.~\cite{2012_CW} which is stated below for reference.
\begin{lemma}[LCU Lemma~\cite{2012_CW}]\label{lem:LCU}
Let $\prep_i$ and $\sel_i$ be unitary operators that conform to Eqs.~\eqref{eqn:prepi} and~\eqref{eqn:seli} and act on $\mathbb{C}^{2^{n+m}}$ we then have that 
\begin{eqnarray}
    (\bra{0}\otimes \id_n)(\prep_i^{\dagger}\otimes \id_n)\cdot \sel_i\cdot(\prep_i\otimes \id_n)(\ket{0}\otimes \id_n)&=&\frac{H_i}{\lambda_i}.    \label{eqn:prepiSeli}
\end{eqnarray}
\end{lemma}
If we then have $M$ Hamiltonians, $H_1,\ldots,H_M$, each of which has its own LCU decomposition, we can define the subroutines as in Eqs.~\eqref{eqn:prepi} and~\eqref{eqn:seli}. Now we use these subroutines to define the following,
\begin{eqnarray}
 \prep\ket{0}^{\log M+\sum_i\log M_i}&=&\left(\sum_{i=1}^M\sqrt{\frac{w_i\lambda_i}{\mathcal{\nconst}}}\ket{i}\right)\otimes\bigotimes_{i=1}^M\prep_i    \label{eqn:divPrep} \\
 \sel&=&\sum_{i=1}^M\left(\ket{i}\bra{i}\otimes\bigotimes_{k=1}^{i-1}\id\otimes\sel_i\otimes\bigotimes_{k=i+1}^M\id\right)     \label{eqn:divSel}
\end{eqnarray}
where $w_i>0$ and $\nconst=\sum_{i=1}^Mw_i\lambda_i$.  It can be shown that the above two subroutines block encode the sum of Hamiltonians, that is the following equation holds \cite{2012_CW, 2023_MSW}. 
\begin{eqnarray}
    &&(\bra{0}\otimes 1)(\prep^{\dagger}\otimes \id)\cdot\sel\cdot(\prep\otimes \id)(\ket{0}\otimes 1)=\frac{1}{\nconst}\sum_{i=1}^Mw_iH_i. \nonumber
\end{eqnarray}
In general we can of course perform such a block encoding in a single step or we could construct our block encodings as a sum of other block encodings and implement that block encoding in a larger space.  This approach is extensively used in chemistry simulations~\cite{2018_BGBetal,2023_MSW} and is stated below as a theorem for convenience.
\begin{theorem} [Recursive Block Encoding \cite{2023_MSW}]
Let $H=\sum_{i=1}^Mw_iH_i$ for $M$ a positive integer power of $2$, $w_i> 0$ and Hermitian $H_i$ expressed as sum of unitaries: $H_i=\sum_{j=1}^{M_i}h_{ij}U_{ij}$ such that $\lambda_i=\sum_j|h_{ij}|$. Each of the summand Hamiltonian is block-encoded using the subroutines defined in Eqs.~\eqref{eqn:prepi} and \eqref{eqn:seli}. Then, we can have an $(\mathcal{A},\log(M),0)$-block encoding of $H$, where $\nconst=\sum_{i=1}^Mw_i\lambda_i$, using the ancilla preparation subroutine ($\prep$) defined in Eq.~\eqref{eqn:divPrep} and the unitary selection subroutine ($\sel$) defined in Eq.~\eqref{eqn:divSel}.
\begin{enumerate}
    \item The PREP subroutine has an implementation cost of $\mathcal{C}_{\prep}=\sum_{i=1}^M\mathcal{C}_{\prep_i}+\mathcal{C}_{w}$, where $\mathcal{C}_{\prep_i}$ is the number of gates to implement $\prep_i$ and $\mathcal{C}_w$ is the cost of preparing the state $\sum_{i=1}^M\sqrt{\frac{w_i\lambda_i}{\nconst}}\ket{i}$.

    \item The $\sel$ subroutine can be implemented with a set of multi-controlled-X gates - \\
    $\{M_i\text{ pairs of }C^{(\log M_i)+1}X\text{ gates }:i=1,\ldots,M\}$, $M$ pairs of $C^{\log M}X$ gates and $\sum_{i=1}^MM_i$ single-controlled unitaries - $\{cU_{ij}: j=1,\ldots,M_i; \,i=1,\ldots,M\}$. 
\end{enumerate}
 \label{thm:blockEncodeDivConq}
\end{theorem}

Suppose, in Theorem \ref{thm:blockEncodeDivConq}, all the $H_i$ are the same but they act on disjoint subspaces. In this case, each $\prep_i$ is the same and so it is sufficient to keep only one copy of $\prep_i$ in the $\prep$ subroutine of Eq.~\eqref{eqn:divPrep}. We can absorb $w_i$ in the weights of the unitaries obtained in the LCU decomposition of $H_i$. Thus, in this case we have
 \begin{eqnarray}
  \prep\ket{0}^{\log M+\log M_i}=\left(\sqrt{\frac{1}{M}}\sum_{i=1}^M\ket{i}\right)\otimes \prep_i \ket{0} \label{eqn:divPrepEq}
 \end{eqnarray}
We also need to make slight modifications to the $\sel$ procedure.
This time, we keep an extra ancilla qubit, initialized to 0, in each subspace. Given a particular state of the first register, we select a subspace by flipping the qubit in the corresponding subspace. The unitaries in each subspace are now additionally controlled on this qubit (of its own subspace). In Ref.~\cite{2023_MSW} we have discussed the more general situation when each $\prep_i$ are same but the Hamiltonians $H_i$ are different.

In our SELECT subroutines we need to implement a number of unitaries controlled on the state of some ancillae. Usually these are decomposed into a number of multi-controlled-X gates and controlled unitary operators. If we implement each multi-controlled-X gate with a Clifford+T gate set \cite{2017_HLZetal}, then we incur a higher T-gate count. Instead, a group of $n$ multi-controlled-X gates can be optimized and implemented with $O(n)$ T-gates using the constructions in Refs.~\cite{2020_dMGM, 2023_RBMetal, 2023_MSW} or with $O(n-\log n)$ T-gates using the construction in Ref.~\cite{2024_Mqram}, the latter having the benefit of exponentially less T-depth. Though T-depth is an important metric from a fault-tolerant perspective \cite{2013_AMMR, 2022_GMM2}, in this paper we focus only on the T-count. Here we mention that a controlled-rotation can be implemented by decomposing into two single-qubit rotation gates, which are further implemented with a Clifford+T gate set using the constructions in Refs.~\cite{2015_KMM, 2016_RS}, both incurring a T-count of $O(\log\frac{1}{\epsilon})$, where $\epsilon$ is the precision or synthesis error. Alternatively, one can use the algorithm in Ref.~\cite{2022_GMM} that works for arbitrary multi-qubit unitaries, and the implementations in this paper indicate that the T-count of controlled rotations is at most the T-count of single-qubit rotations. Various synthesis algorithms chose different distance metrics and the two most popular ones are the operator norm and the global phase invariant distance. In Ref.~\cite{2021_M}, inequalities relating these two metrics have been derived and in applying these results to the 2-qubit case, we find that the T-count of controlled rotations varies by a constant, which we ignore for our asymptotic analysis.

We now briefly describe the recursive block encoding of $H_{\pi}$, using Theorem~\ref{thm:blockEncodeDivConq}. First, we fragment $H_{\pi}$ (Eq. \ref{eqn:Hpi}) as follows.
\begin{eqnarray}
    && H_{1\pi}^{j,\mu} = \frac{1}{m_j}\nabla_{j,\mu}^2\otimes\id, \quad H_{2\pi}^{j,q,\mu} = \frac{\zeta_j}{m_j}\nabla_{j,\mu}\otimes A_{q,\mu},\quad H_{3\pi}^{j,q,\mu} = \frac{\zeta_j^2}{m_j}\id\otimes A_{q,\mu}^2,  \nonumber \\
    && H_{23\pi}^{j,\mu} = \sum_{q=0}^{M-1}\left( -2i H_{2\pi}^{j,q,\mu}+ H_{3\pi}^{j,q,\mu}\right)\delta_{x\in D_q},  \nonumber \\
    && \qquad = \sum_{q=0}^{M-1} \left(-2i H_{2\pi}^{j,q,\mu}+H_{3\pi}^{j,q,\mu}\right)\left(\id - R_{x\in D_q} \right)/2\nonumber\\
    && H_{\pi} = \frac{1}{2}\sum_{j=0}^{\eta-1}\sum_{\mu=0}^2-H_{1\pi}^{j,\mu}+H_{23\pi}^{j,\mu}
    \label{eqn:HpiFragment}
\end{eqnarray}
Above we have introduced a reflection operator of general form $R_{x\in D_q}$ which has the property that for any quantum state $\sum_x a_x \ket{x}$
\begin{equation}
    R_{x\in D_q}\left(\sum_{x\in D_q} a_x \ket{x} + \sum_{y \not \in D_q} a_y \ket{y}\right) := -\sum_{x\in D_q} a_x \ket{x} + \sum_{y \not \in D_q} a_y \ket{y}.
\end{equation}
This is useful because it removes the Kronecker-delta from the definition of the Hamiltonian and replaces it with a sum of unitaries.  This sum of unitaries can then be easily implemented using the LCU approach. Specifically, we denote $R_{j,x\in D_q}$ to denote this is acting on the $j$th particle grid register. Now, the explicit block encoding costs are described in the lemmas below.

\begin{lemma}[Block Encoding of $H_{\pi}$]\label{lem:CostPi}
Unitary operators $\prep_\pi,\sel_\pi$ can be constructed that provide an $(\mathcal{A}_2',\cdot,\epsilon)$ block-encoding of $H_\pi \in L(\mathcal{H})$, meaning that
$$
\left|(\bra{0}\otimes \id)|(\prep_{\pi}^{\dagger}\otimes \id)\cdot\sel_{\pi}\cdot(\prep_{\pi}\otimes \id) (\ket{0}\otimes \id) - \frac{H_{\pi}}{ \nconst_2' }\right| \le \epsilon
$$
where
$$
\nconst_2'\le \frac{2\pi^2 \eta}{\Delta^2} + \frac{6 \pi \eta M \ln(2a^2)}{h\Delta} + \frac{6 \pi \eta M }{h^2}
$$
using a number of $T$-gates needed to implement $\prep_{\pi}^{\dagger}\cdot\sel_{\pi}\cdot\prep_{\pi}$  that scale as
$$
\mathcal{C}_1 =\widetilde{O}\left((a+\log^2(\Lambda))\log 1/\epsilon + (M+\eta)\left( a\log(N\Lambda)\right)\right)
$$
\end{lemma}
The proof of this lemma is technical and is provided in Appendix~\ref{app:ProofCostPi}.

Next, in order to apply the recursive block-encoding Theorem, we need to also construct a block-encoding of the field Hamiltonian.   Recall that in the definition of $H_f$ there is a term that emerges from the electrostatic potential and another term that describes the magnetic interaction $H_{f2}$.
\begin{lemma}\label{lem:HpiBE}
An $(O(Mc^2\log^2(a)/h),\cdot,\epsilon)$ block-encoding of $H_{f_2} \in L(\mathcal{H})$ and the number of $T$-gates needed to implement this block encoding scale as
     $$\mathcal{C}_{2}\in \widetilde{O}(Ma^2\log( \Lambda)\log(1/\epsilon)).$$
\end{lemma}
Proof of this lemma is again provided in Appendix~\ref{app:ProofCostPi}.

\begin{cor}\label{cor:B}
    An $(O(\mathcal{A}'_2 + Mc^2\log^2(a)/h),\cdot, \epsilon)$ block-encoding, where $\mathcal{A}_2'$ is the normalization constant of $H_{\pi}$ defined in Lemma~\ref{lem:CostPi}, of $B:=H_\pi + H_{f_2}\in \mathcal{L}(\mathcal{H})$ can be constructed using a number of $T$ gates that is in
    $$\mathcal{C}_B:=\widetilde{O}\left((a+\log^2(\Lambda))\log 1/\epsilon + (M+\eta)\left( a\log(N\Lambda)\right)+Ma^2\log(\Lambda)\log(1/\epsilon)\right)
    $$
\end{cor}
\begin{proof}
We can block encode a linear combination of $H_{\pi}$ and $H_{f2}$ using the technique described in Theorem \ref{thm:blockEncodeDivConq}.  Then using the additivity of block-encoding constants and the gate bounds that are given in Lemmas~\ref{lem:CostPi} and~\ref{lem:HpiBE}, the total number of gate operations needed in the entire block encoding is
\begin{eqnarray}
     O(\g_1+\g_2) \subseteq \widetilde{O}\left((a+\log^2(\Lambda))\log 1/\epsilon + (M+\eta)\left( a\log(N\Lambda)\right)+Ma^2\log(\Lambda)\log(1/\epsilon)\right)
\end{eqnarray}
\end{proof}

\begin{theorem}\label{thm:gatecount}
Let $H$ be the Gauss' law constrained Pauli-Fierz Hamiltonian as given by Definition~\ref{def:Ham} for $\eta$ particles, $N$ spatial grid points in the simulation with grid spacing $\Delta =\Omega^{1/3}/N^{1/3}$, $M\le N$ grid points for the electric field mesh, an electric field cutoff of $\Lambda$ and a finite difference approximation to the kinetic operator that uses $2a+1$ points and a Gauss' law check that uses a cube consisting of $(2b+1)^3$ points.  The number of $T$ gates needed to perform a simulation of the Hamiltonian for time $t$ within error $\epsilon$ with respect to the Euclidean norm is then for $c\gg 1$ in
    $$
\widetilde{O}\left( \eta t\left(\frac{Mc^2}{\eta} + \frac{1 }{\Delta^2}\right)\left(M(b^3+a^2)+\sum_i|\zeta_i| \right)\log^2(N\Lambda)\log^2(1/\epsilon)   \right)
    $$
\end{theorem}
\begin{proof}
We apply Lemma \ref{lem:2018LW} in order to estimate the cost of implementing $H_{\rm int}(t)$.  We will first compute the number of queries to $\hamt$, the oracle that block-encodes the time-dependent Hamiltonian, and then multiply the result by the number of gates needed to implement the $\hamt$ oracle.

In our case, $A=H_{f1}+H_c$ and $B=H_{f2}+H_{\pi}$ in the interaction picture simulation definition of Lemma \ref{lem:2018LW}. Thus using the triangle inequality we obtain that our block encoding of $B$ has the following scaling for its coefficient $\ell_1$-norm as
\begin{eqnarray}
    \|B\|_{\ell_1} &\in& O\left(\frac{Mc^2 \log^2(a)}{h}+\frac{\eta}{\Delta^2}+\frac{\eta\ln 2a^2}{ch\Delta}+\frac{\eta}{c^2h^2} \right) \nonumber \\
    &\in& O\left(\eta\log^2(a)\left(\frac{Mc^2}{\eta}+\frac{1}{\Delta^2} +\frac{1}{c^2h^2}\right)\right).
    \label{eqn:alphaB}
\end{eqnarray}
under the assumption that $h\ge \Delta$ and $c\gg 1$.
 Thus from the previous lemma, we can approximate $e^{-i\hat{H}t}$ to error $\epsilon$, using the following number of queries to HAM-T.
 \begin{eqnarray}\label{eq:final_cost}
     N_{HT}=\widetilde{O}\left( \eta \log^2(a)t\left(\frac{Mc^2}{\eta} + \frac{1 }{\Delta^2}\right)\log(1/\epsilon)  \right)
 \end{eqnarray}
 The implementation of the $\hamt$ oracle works as follows: 
\begin{enumerate}
    \item For each Hamiltonian $H_j$ in $B$ construct a block encoding $U_j$ of $H_j$ using oracles $\prep_j$ and $\sel_j$.
    \item Modify each $\sel_j$ via
    $$
    \sel_j \rightarrow V_{\rm int}^\dagger(\id \otimes \sel_j) V_{\rm int}
    $$
    \item Use existing oracles $\prep_j$ in conjunction with $\sel_j$ to reweight the Hamiltonian terms by their original weights in the non-interaction frame (denoted here as $\alpha_B(j)$) to construct a block encoding of  $\sum_j \alpha_B(j) H_{{\rm int},j}(t)$.
\end{enumerate}
The validity of this algorithm to implement the $\hamt$ oracle immediately follows from the LCU lemma, and the definition of $\hamt$ given in Lemma~\ref{lem:2018LW}.  From this it is clear that the gate complexity of the simulation satisfies
\begin{align}
    \mathcal{C}_{\hamt} 
    &= O(\mathcal{C}_B+ \mathcal{C}_{V_{\rm int}})\nonumber\\
    &\subseteq \widetilde{O}\Biggr(\left((a+\log^2(\Lambda))\log 1/\epsilon + (M+\eta)\left( a\log(N\Lambda)\right)+Ma^2\log(\Lambda)\log(1/\epsilon)\right)\nonumber\\
    &\qquad+ \left(b^3(M+\sum_i |\zeta_i|)\log^2(N\Lambda) + n_t\log(n_t/\epsilon)\right)\Biggr)\nonumber\\
    &\subseteq\widetilde{O}\biggr( b^3(M+\sum_i|\zeta_i|)\log^2(N\Lambda) +(n_t+Ma^2\log^2(N\Lambda))\log(1/\epsilon)\Biggr)\nonumber\\
    &\subseteq\widetilde{O}\biggr( \left(M(b^3+a^2)+\sum_i|\zeta_i|+n_t\right)\log^2(N\Lambda)\log(1/\epsilon)\Biggr)\label{eq:GHamt}
\end{align}
where we have used Lemma~\ref{lem:VintImp} for $\mathcal{C}_{V_{\rm int}}$.  The remaining issue involves defining the number of time-points needed for the simulation.  The value of $n_t$ needed in order to ensure error at most $\epsilon$ for a simulation of duration $t$ is given from~Lemma~\ref{lem:2018LW} and using Eq.~\eqref{eq:Hf1Bd} as
\begin{equation}
    n_t \in O\left(\log((\lambda+ M\Lambda^2)t/\epsilon) \right)
\end{equation}
Then applying Lemma~\ref{lem:constraint} we see that the value of $\lambda$ needed to achieve an error of $\epsilon$ in the Dyson series simulation gives
\begin{equation}
    n_t \in O\left(\log(((\|H_\pi\|_\infty^2 t/\epsilon)+ M\Lambda^2)t/\epsilon) \right)
\end{equation}and using the bound on the norm of $H_\pi$ given in Lemma~\ref{lem:CostPi} we see that for constant $c$ and $h\ge \Delta$ 

\begin{align}
    n_t &\in O\left(\log\left(\left(\left(\left(\frac{\eta}{\Delta^2}+\frac{\eta \log (a^2)}{ch\Delta}+\frac{\eta}{c^2h^2} \right)^2 t/\epsilon\right)+ M\Lambda^2\right)t/\epsilon\right) \right)\nonumber\\
    &\subseteq O\left(\log\left(\frac{\eta^2t^2\log^2(a)}{\Delta^4\epsilon^2}+ \frac{M\Lambda^2t}{\epsilon}\right) \right).
\end{align}

Substituting this into Eq.~\eqref{eq:GHamt} yields a total number of gates per invocation of $\hamt$
\begin{align}
    \mathcal{C}_{\hamt} &\in 
     \widetilde{O}\biggr( \left(M(b^3+a^2)+\sum_i|\zeta_i|+\log\left(\frac{\eta^2t^2\log^2(a)}{\Delta^4\epsilon^2}+ \frac{M\Lambda^2t}{\epsilon}\right) \right)\log^2(N\Lambda)\log(1/\epsilon)\Biggr)
\end{align}
The total number of gates in the simulation is then under
\begin{align}
   N_{HT} \mathcal{C}_{\hamt}&\in\widetilde{O}\left( \eta t\left(\frac{Mc^2}{\eta} + \frac{1 }{\Delta^2}\right)\left(M(b^3+a^2)+\sum_i|\zeta_i|+\log\left(\frac{\eta^2t^2\log^2(a)}{\Delta^4\epsilon^2}+ \frac{M\Lambda^2t}{\epsilon}\right) \right)\log^2(N\Lambda)\log^2(1/\epsilon)  \right)\nonumber\\
   &\subseteq \widetilde{O}\left(\eta t\left(\frac{Mc^2}{\eta} + \frac{1 }{\Delta^2}\right)\left(M(b^3+a^2)+\sum_i|\zeta_i|+\log\left(\frac{\Lambda}{\epsilon}\right) \right)\log^2(N\Lambda)\log^2(1/\epsilon)  \right)\nonumber\\
   &\subseteq \widetilde{O}\left( \eta t\left(\frac{Mc^2}{\eta} + \frac{1 }{\Delta^2}\right)\left(M(b^3+a^2)+\sum_i|\zeta_i| \right)\log^2(N\Lambda)\log^2(1/\epsilon)   \right)
\end{align}
under the assumption that $\log(\Lambda/\epsilon)\in O(M)$.
\end{proof}

\section{Emergence of Coulomb Interaction}\label{sec:coulomb}
As mentioned above, the Coulomb interaction studied in prior work should be an emergent property of non-relativistic quantum electrodynamics.  This is significant because it implies that, in principle, the methodologies proposed here can be used to simulate both chemistry as well as the Pauli-Fierz-Coulomb Hamiltonian in the correct limit.  Our aim here is to investigate the nature of this limit and argue that we can use our methodologies here to perform such simulations.  Analogous discussions have been performed in the lattice gauge theory literature, see for example Ref.~\cite{kogut1979introduction}.  Our analysis will reveal however, that the stability of the Coulomb interaction can be seen as emerging in the thermodynamic limit out of the topological properties of the underlying electric field states under the dynamics of the Gauss' law constrained Pauli-Fierz Hamiltonian in a manner analogous to the emergence of topological protection for the Toric code Hamiltonian.  Specifically, we will argue that the dynamics of the model cannot lead to contractable or uncontractable loops of charge excitations being created.  These loops can be thought of as logical operations in the Gauss' law ``code'' and in turn opens up the possibility that such errors could be detected and corrected in the appropriate limit for the simulation.

The Coulomb Hamiltonian in particular takes the following form as an operator acting only on the Hilbert space $\mathcal{H}_{C}:=\bigotimes_{i=0}^{\eta-1}\left(\bigotimes_{\mu=0}^2\mathcal{H}_{i,\mu}\right)$
\begin{equation}
    H= \sum_{i,\mu} \frac{p_{i,\mu}^2}{2m_i} + \sum_{i\neq j} \frac{\zeta_i\zeta_j}{|x_i - x_j|}
\end{equation}
The magnetic interaction is completely absent from this Hamiltonian and the two particle interaction term can be thought of as the energy of the electromagnetic field.  This interaction also only appears in the limit of low velocity motion, because otherwise Amp\`ere's law leads to a non-negligible magnetic interaction between the particles.  However, this interaction can be seen to be equivalent to the dynamics given by an electric field of the form
\begin{equation}
    E(x) = \sum_{i} \frac{\zeta_i (x-x_i)}{|x-x_i|^3}
\end{equation}

Our aim is ultimately to place restrictions on the evolution of the field showing that under Gauss' law, in the non-relativistic limit, an initial field distribution that satisfies Coulomb's law must also retain it, as time proceeds.  Our reasoning behind this will be topological in nature.

\subsection{Topological Protection of Simulations}
Topology is the study of the shapes of objects and at its core the concept of homotopy creates a notion of equivalence between different geometrical bodies.  We say that two such bodies are equivalent in this sense if there exists a continuous deformation that maps one body into the other one.  From this perspective a torus is innequivalent to a sphere because the hole in the center of the torus prevents it from being continuously deformed into the sphere.  

In our setting, we consider our space to have periodic boundary conditions for simplicity.  This makes our space a Torus.  Within this space, the electric field can be thought of as forming paths.  Formally, these paths are a sequence of positions and directions within the electric field grid such that one electric field vector within the path points to the next.  For example $(E_{p,x},E_{p+e_x,y}, E_{p+e_x+e_y,z})$ would form a path.  This path is contractable, meaning that we can continuously deform it back to a point.  In contrast, a path that wraps around the boundary and intersects with its self is not contractable.  For example if we had a $3\times 3 \times 3$ cubic grid, then the path $(E_{-1,x},E_{0,x},E_{1,x})$ is not contractable because any attempt to shorten it in the aforementioned way would break a link in the path which is a discontinuous change.

We now wish to think about the dynamics of the electric field using the language of topological codes.  The idea behind this is to consider a code space that describes the subspace of all quantum states that we consider to be valid.  We demand that the input states satisfy Gauss' law on this torus.  Specifically, we require for any input state $\ket{\psi_0}$ that $H_c \ket{\psi_0} =0$.  This corresponds to assuming that the Gauss' law constraint is satisfied for the entire evolution.  Next, we will require that the initial state  obeys  Coulomb's law. We can think of such states in terms of a projector onto the space of initial field configurations that satisfy Coulomb's law.  Specifically the projector onto this space is
\begin{equation}
    \Pi_{coul} = \sum_{x_{i,\mu}}\bigotimes_{i,\mu} \ketbra{x_{i,\mu}}{x_{i,\mu}} \bigotimes_{j,\nu} \ketbra{E(x_j)\cdot e_\nu}{E(x_j)\cdot e_\nu},
\end{equation}
and all states within the space obeying Coulomb's law.  Coulomb's law commutes with the Gauss' law constraint: $[H_{coul},H_c]=0$.  This means that we can think of the states that obey Gauss' law as a subspace of the codespace that is stabilized by $\openone - H_c$.  We therefore define a physical error to be the addition of a path to an electric field configuration that causes it to violate $\openone - H_c$.  A logical error, on the other hand,  corresponds to adding an electric field path that satisfies Gauss' law.  

We also have an additional constraint to account for.  As mentioned, the appearance of magnetic fields can ruin the emergence of the Coulomb Hamiltonian.  This means that initial configurations that lead to substantial magnetic fields, or rates of change of magnetic fields, need to be removed in order for Coulomb's law to emerge.  This places a further restriction, analogous to that imposed by the plaquette operators in the Toric code, that no closed electric field loops are present in the system.  We define these loops through the following electric field plaquette operator, which we can view as a negatively oriented electric field path as given by the following definition and the projector onto all states that obey this property.
\begin{definition}
    Let $\Pi_{\Box}$ be the projector onto all states that are in the kernel of $\sum_{p} \sum_{\nu\in \{x,y,z\}} (E^{\Box}_{p,\nu})^2$ where for example
    \begin{align*}
    E^{\Box}_{p,x} := &\left(E_{p+e_y,z}+E_{p+e_y-e_z,z} - E_{p-e_y,z} -E_{p-e_y+e_z,z}\right. \nonumber\\
&\left.+E_{p+e_z,y}+E_{p+e_x-e_y,y}-E_{p-e_z,y}-E_{p-e_z+e_y,y}\right)
\end{align*}
\end{definition}

These plaquette operators measure the sum of the electric field along a closed loop.  The electric field integral about a closed loop is related to the rate of change of magnetic flux through the loop from the integral form of the Farraday-Maxwell equation.   Intuitively then, we expect the creation of a contractable field loop will lead to a magnetic field after a short time evolution.  Such fields are not expected to arise in the non-relativistic limit and so we anticipate that a  discrete Maxwell-Faraday relation would forbid this.  We state such a relation below and prove in Appendix~\ref{app:Maxwell}.

\begin{prop}[Discrete Maxwell-Faraday Relation]\label{prop:Maxwell}
    Let $\ket{\psi}$ be in a family of states for varying $\Lambda$ that are smooth in the sense that for all $p$ and all directions $\nu$ $\|(U_{p,\nu}-I)\ket{\psi}\| \in O(L/\Lambda)$ for some constant $L>0$.  Further, assume that for example
    \begin{align*}
            (\nabla \times A)_{p,x}= \frac{1}{4h}&\left(A_{p+e_y,z}+A_{p+e_y-e_z,z} - A_{p-e_y,z} -A_{p-e_y+e_z,z}\right. \nonumber\\
    &\left.+A_{p+e_z,y}+A_{p+e_x-e_y,y}-A_{p-e_z,y}-A_{p-e_z+e_y,y}\right)
    \end{align*}
    (and similar expressions are used for the $y,z$ components)
    is used as a discrete approximation to the curl of the vector potential.  We then have that for any contractable simple closed and positively oriented path $C$
    $$
    \lim_{\Lambda\rightarrow \infty} \bra{\psi} \oiint_C \partial B(t) / \partial_t\cdot dA \ket{\psi} = -\lim_{\Lambda \rightarrow \infty} \bra{\psi(t)} \oint_{C}E(q)\cdot dq \ket{\psi(t)}.
    $$
\end{prop}

Let us define $P$ to be a set of pairs of vertices of the electric field grid and their directions.  We take the slightly redundant following notation to describe such a set.  Let $P$ be a set of the form $\{(j_i,\nu_i,s_i):i=1,\ldots,|P|\}$ where $j,\nu$ are the locations and components of the fields at site $j$ and $s\in \{-1,1\}$ is the direction of the path chosen where $s=1$ is a forward edge and $s=-1$ a reverse.  We then define the electric field shift operator to be 

\begin{equation}
    {F}_{P;\vec{m}}:= \openone \bigotimes_{i=1}^{|P|}U_{j_i,\nu_i}^{m_i s_i},
\end{equation}
where $(j_i,m_i,s_i)$ are the indexed elements of the set.

Next, note that without loss of generality we can decompose the set $P$ into a set of simple paths and cycles.  In particular for any set $P$ there exist $P_1,P_2,\ldots,P_K$ for $K< 3M$ such that each set $\{(j,v,s)\}$ assigned to each $P_\ell$ is a path. That is to say, we can express
\begin{equation}\label{eq:PathDef}
    {F}_{P;\vec{m}}:= \prod_{\ell=1}^K F_{P_\ell;\vec{m}_\ell} =\prod_{\ell=1}^K\openone \bigotimes_{i=1}^{|P_{\ell}|}U_{[j_{\ell}]_i,[\nu_\ell]_i}^{[m_\ell]_i [s_\ell]_i},
\end{equation}

Our aim is to show a form of topological protection for the field in that as $c\rightarrow \infty$ that no electric field loops can be present in the output state, with high probability, unless the initial state also contains them.
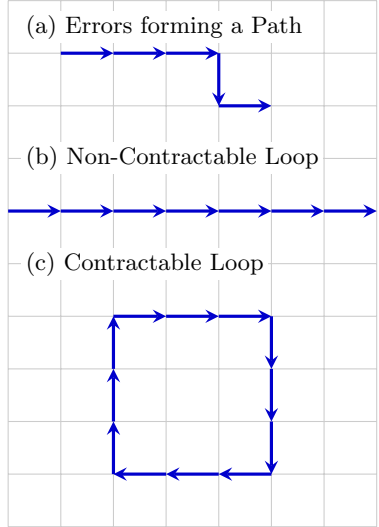
\begin{figure}
    \centering 

   \begin{tikzpicture}[x=0.7cm, y=0.7cm]
        \draw[gridstyle] (0,0) grid (7,10);

        \node[labelstyle] at (0.25, 9.5) {(a) Errors forming a Path};
        \draw[error] (1,9) -- (2,9);
        \draw[error] (2,9) -- (3,9);
        \draw[error] (3,9) -- (4,9);
        \draw[error] (4,9) -- (4,8); 
        \draw[error] (4,8) -- (5,8);

        \node[labelstyle] at (0.25, 7) {(b) Non-Contractable Loop};
        \draw[error] (0,6) -- (1,6);
        \draw[error] (1,6) -- (2,6);
        \draw[error] (2,6) -- (3,6);
        \draw[error] (3,6) -- (4,6);
        \draw[error] (4,6) -- (5,6);
        \draw[error] (5,6) -- (6,6);
        \draw[error] (6,6) -- (7,6);

        \node[labelstyle] at (0.25, 5) {(c) Contractable Loop};
        \draw[error] (2, 4) -- (3, 4); 
        \draw[error] (3, 4) -- (4, 4); 
        \draw[error] (4, 4) -- (5, 4); 
        \draw[error] (5, 4) -- (5, 3); 
        \draw[error] (5, 3) -- (5, 2); 
        \draw[error] (5, 2) -- (5, 1); 
        \draw[error] (5, 1) -- (4, 1); 
        \draw[error] (4, 1) -- (3, 1); 
        \draw[error] (3, 1) -- (2, 1); 
        \draw[error] (2, 1) -- (2, 2); 
        \draw[error] (2, 2) -- (2, 3); 
        \draw[error] (2, 3) -- (2, 4); 
        
    \end{tikzpicture}

    \caption{
        Visualization of three types of error patterns considered on top of a background field that satisfies Gauss' law.
        Case (a) is a chain of errors that forms a simple path. These paths violate Gauss' law and do not emerge in the simulation.
        Case (b) demonstrates a non-contractable loop, which corresponds to a logical error in the code space that satisfies Gauss' law. These errors are suppressed by the volume of the box and vanish in the thermodynamic limit.
        Case (c) describes a contractable loop of errors in the electric field. These errors do not occur in the non-relativistic limit because the plaquette operator corresponding to this error pattern commutes with the Hamiltonian.
    }
    \label{fig:paths}
\end{figure}

Our argument will now look at what happens to the energy for such a state $\ket{\psi_0}$ that is stabilized by the projectors $\Pi_{\Box}$ and $\openone - H_c$ when we modify the electric field.  Let us denote such a modification as an operator $F$ and then the mean energy of the perturbed state that is excited is of the form
$ \bra{\psi_0} F^\dagger H F\ket{\psi_0}$.  We further demand that the perturbation $F$ acts as a logical operation on the codespace defined by these projectors.
We can then consider what happens to the mean energy given such a modification to a system
\begin{equation}\Delta_{\mathcal{U}}:=\bra{\psi_0} F^\dagger H F\ket{\psi_0}-\bra{\psi_0}H \ket{\psi_0}
\end{equation}
Our argument is that in the thermodynamic limit, for any $F$ that acts as a logical operation on this space, meaning that it satisfies~Gauss' law, it will either lead to a large energy shift or that the dynamics of the system in the non-relativistic limit will be protected from errors.  To classify the allowed geometries of logical errors, we analyze the structure of the underlying paths $P$ upon which the field modification operator $F_{P;\vec{m}}$ acts.

Let us assume that ${F}_{P;\vec{m}}$ is a logical operation (meaning that the path does not violate Gauss' law) and that $P$ can be decomposed into $K$ cycles and simple paths such that $K$ is minimal and $m_\ell=1$ is constant for each $\ell\in [1,K]$.   If any two paths, say $P_i$ and $P_j$, could be joined head-to-tail, the decomposition could be expressed with $K-1$ paths, contradicting the premise of minimality.  This disconnectedness, in turn, forbids any simple path from starting or ending in the bulk of the lattice; such a geometry is depicted in Fig.~\ref{fig:paths}(a) for a 2D plane which we will argue leads to a violation of Gauss' law.

More formally,
let us assume that there exist at least two simple paths such that $P_1 \cup P_2$ is a simple path and the constant values of both strings are equal $\max{m_1}=\max{m_2}$.  
This is tantamount to saying that either the initial point of $P_2$ corresponds to the endpoint of $P_1$ or vice-versa.  Let us take $P_1$ and $P_2$ to be the first two such paths in some enumeration of the simple paths in the decomposition of ${F}_{P;\vec{m}}$.  If these first two paths exist then ${F}_{P;\vec{m}}$ can be decomposed into at least $K-1$ paths.  We assumed, however, that $K$ is minimal and so such $P_1$ and $P_2$ cannot exist.  Thus if such a path exists we must have that $\max{m_1} \ne \max{m_2}$, which is not true because we assumed that $\max m_1 =\max m_2 =1$.  
As we assumed that these were the first such paths in the decomposition it follows that for all $i,j$ $P_i\cup P_j$ is not a simple path under our assumption that the paths are constructed such that $\max m_i = \max m_j=1$.

{Let us now consider a path $P_i$ from the set of simple paths in the decomposition.  The endpoints of the $P_i$ must either be on the boundary or in the bulk.  Let us assume that $P_i$ has an isolated end point in the bulk.  If this is true then there exists a closed cubic surface $S_i$ such that (slightly abusing notation to describe the discrete sum as an integral) $\oiint_{S_i} F_{p,\vec{m}} \cdot dA\neq 0$ because flux is destroyed at such an end point unless there exists another path that creates an outgoing electric flux to counteract it.  This is not possible for all such $S_i$ unless the end points of the path corresponds to the starting point of another path.  This is impossible by the above argument.  Thus Gauss' law would be violated by the existence of such a path which shows that no such paths can exist by contradiction.  Repeating the same argument leads us to the conclusion that the starting point cannot be in the bulk either.  
Thus all simple paths in the decomposition must terminate on the boundary and begin on the boundary.}

This implies that any logical error operator must be composed entirely of closed loops. On a lattice with periodic boundary conditions, these loops are classified as either \emph{non-contractable loops} that wrap around the global topology of the simulation volume, as in Fig.~\ref{fig:paths}(b), or \emph{contractable loops} that form localized, closed paths within it, as in Fig.~\ref{fig:paths}(c). This classification is essential for the analysis below that demonstrates that these distinct types of loops are suppressed by energetic penalties and symmetry constraints in the relevant physical limits.
  This language is chosen to reflect terminology used in topology wherein loops that are formed around the boundary of a torus are non-contractable, whereas those that do not are called contractable because they can be continuously deformed to a point.  Here the notion of homotopy can be thought of as emerging from the observation that we can consider a cycle $C_i$ of field creation operators that is contractable and note that for any pair of non-parallel edges $e_{k},e_{k+1}$ that are adjacent to each other in the cycle and are on the boundary of the closure of the cycle, we can construct a cycle $C_j$ with edges $e_{k,k_0},e_{k+1,k}, e_{\ell,k+1}, e_{\ell+1,\ell}$ such that that 
\begin{equation}
    C_j C_i = U_{\ell} U_{\ell+1}.
\end{equation}  
This corresponds to a contraction of the corner of the loop to a smaller volume.  As
\begin{equation}
    [C_i, H_c]=0~\Leftrightarrow [C_iC_j,H_c]=0.
\end{equation}
From this perspective we can view the reducible cycles $C_i,C_j$ as logical operations on our codespace formed by the kernel of $H_c$.  In this sense if we form an equvalence class between all states that are linked by logical operations within this code space, if $\ket{\psi}$ obeys $H_c \ket{\psi}=0$ then $\prod_i C_i \ket{\psi}$ does as well for any contractable cycles.  By repeating the same argument given above about multiplying by cycles it is clear that any cycle $C_j$ can be expressed as the product of cycles $C_i$ that consist of only $4$ vertices.  Thus in this sense, any reducible cycle can be eliminated by adding in primitive cycles that consist of only field creation operators on $4$ vertices.

Now let us consider the shift in mean energy that occurs by introducing a non-contractable path of electric field excitations on any state, $\ket{\psi_0}$, in the kernel of $H_c$ that arises from adding an irreducible cycle of charge to the state:
\begin{equation}
    \Delta_{\mathcal{U}} = \bra{\psi_0} F_p^\dagger H_{f1} F_p - H_{f1}\ket{\psi_0}= \frac{h^3}{8\pi}\left( (2E_{\max}/\Lambda) \bra{\psi_0} \sum_{i=1}^{|P|} \openone \otimes E_{j_i}\cdot e_{\nu_i} \ket{\psi_0} + (E_{\max}/\Lambda)^2|P|\right)
\end{equation}
As the energy scales with $|P|$ we can see that the smallest energy that can be obtained occurs for the smallest value of $|P|$, which occurs for any path that wraps around the torus created by out periodic boundary conditions.  The minimum length of such a loop is $|P| \in \Omega(M^{1/3})=\Omega(\Omega^{1/3}/h)$ and so since the path is by definition oriented in the direction of the electric field, the energy contribution from the sum around the path is at least zero and so the energy shift can be lower bounded by
\begin{equation}
    \min_{F_p} \Delta_{\mathcal{U}} \in \Omega(h^2 E_{\max}^2 \Omega^{1/3} /\Lambda^2).
\end{equation}

Next we claim that if $\ket{\psi_0}$ has mean energy $\langle E\rangle$ and standard deviation of its energy $\sigma$ then $\ket{\psi_0(t)}$ also has the same mean energy for any time $t$.  This can be easily validated by noting that the Hamiltonian and its square both commute with the time evolution operator and thus by substituting into the Heisenberg equations of motion we see that $\langle E(t)\rangle=\langle E\rangle$ and $\sigma(t)=\sigma$.  Then, from Chebyshev's inequality we see that the probability of an energy deviation of this size given that $\langle E\rangle\in O(1)$ (which is appropriate for our discretization as we remove the zero-point energy from the Hamiltonian for each electric field site) is
\begin{equation}
    \delta = O\left(\left(\frac{\Lambda^2\sigma}{h^2E_{\max}^2\Omega^{1/3}-\langle E\rangle\Lambda^2}\right)^2  \right)
\end{equation}
Thus for fixed $h$ and $E_{\max}$, the probability that such a non-contractable loop is observed vanishes as $\Omega \rightarrow \infty$.

We summarize the above argument in the following lemma.

\begin{lemma}\label{lem:nonContract}
    Let $\ket{\psi_0}$ be a state such that $H_c \ket{\psi_0}=0$, has mean energy $\mu \in O(1)$, and a standard deviation of the energy of $\sigma$ that are both independent of system volume $\Omega$. Then assuming periodic boundary conditions,  
    the probability of observing any excitation field along a simple path on the electric field torus $P$, $F_{P;\vec{m}}$ of the form of Eq.~\eqref{eq:PathDef} on the torus of electric field sites for the evolved state $\ket{\psi(t)} = e^{-iHt} \ket{\psi_0}$ for any fixed $t$ obeys the following properties
    \begin{enumerate}
        \item ${\rm Pr}(F_{P;\vec{m}}=0)$ if $\vec{m}$ is non-constant or if $P$ is not a closed loop on the torus.
        \item Let $P$ be a non-contractable loop with $|\vec{m}_j| \ge 1$ for all $j$. Then for fixed, finite $h$ and $E_{\max}/\Lambda$, we have that the probability of observing a non-contractable field loop vanishes as the size of the simulation box increases, meaning that  $\lim_{\Omega \rightarrow \infty}{\rm Pr}(F_{P;\vec{m}})=0$.
    \end{enumerate}
\end{lemma}

We can see that $\nabla \times A$ can be written as an of electric field incrementer operator that acts on a square of field sites in the limit where the field does not change substantially from site to site.  This fact, used above to show the Faraday-Maxwell , is also useful for showing that contractable cycles cannot be created in the non-relativistic (i.e. $c\rightarrow \infty$) limit if appropriate continuity assumptions are made of the field.  This shows that if we begin with a field distribution where every field line begins from a charge source, then the dynamics will always lead to a similar distribution.  This is an important step towards arguing for Coulomb's law emerging and in turn justifying that this formalism can be used for chemistry in the correct asymptotic limit.

\begin{prop}\label{prop:contract}
    Let $\ket{\psi}$ be in a set of states that is in the kernel of $\sum_{p}\sum_{\nu=1}^3(E_{p,\nu}^{\Box})^2$ 
and let the assumptions of Proposition~\ref{prop:Maxwell} hold.
We then have that for any fixed $t,E_{\max},L$ that for any $t\ge 0$ and any closed field loop $C$, the probability that a measurement of $\oint_C E\cdot ds=0$ for $\ket{\psi(t)}$ approaches $1$ in the limit as $c\rightarrow \infty$ and $\Lambda\rightarrow \infty$.

\end{prop}
\begin{proof}
To begin, let us define the effective operator that we get upon stripping away the $U_{pq}$ terms from the commutation relations between $\nabla \times A$, $E^2$ and $E^{\Box}$ wherein the $x-$component of the object is explicitly shown as
\begin{align}
    E^{\Box}_{p,x} := &\left(E_{p+e_y,z}+E_{p+e_y-e_z,z} - E_{p-e_y,z} -E_{p-e_y+e_z,z}\right. \nonumber\\
&\left.+E_{p+e_z,y}+E_{p+e_x-e_y,y}-E_{p-e_z,y}-E_{p-e_z+e_y,y}\right).
\end{align}
From this we see that by definition that if $E^{\Box}_{q,\nu} \ket{\psi} =8kE_{\max}/\Lambda \ket{\psi}$ for integer $k$
\begin{equation}
    e^{-i(4h) (\nabla\times A)_{p,x}E_{\max}/\Lambda} E^{\Box}_{p,x} e^{i(4h) (\nabla\times A)_{p,x}E_{\max}/\Lambda} \ket{\psi} = \frac{8(k+1)E_{\max}}{\Lambda}\ket{\psi},
\end{equation}
which is to say that the evolution increments the electric field around the loop described by the contractable loop in $E^\Box_{p,x}$.  We then have that we can raise the electric field in such a loop by applying the curl of the vector potential as a Hamiltonian.  Thus the two quantities act as conjugate variables. 
Now let us assume that the eigenvalues of $a^\dagger a$ are of the form $a^\dagger a\ket{n} = n\ket{n}$.  
We then have the following, using the above smoothness assumptions on the state
\begin{equation}
    [(\nabla\times A)_{q,x}, hE^\Box_{q,x}] = -i\sum_{p,r \in \Box}U_{p,r} = -iI + O(L/\Lambda).
\end{equation}
Thus we see that these two expressions have the same commutation relations, up to errors that are small for large $\Lambda$ and a small continuity constant $L$, as position and momentum.  This leads us to the following definitions which correspond to the original harmonic oscillator representation, with $P= (\nabla\times A)_{q,x}$ and $Q= h E^{\Box}_{q,x}$.  In these variables the Hamiltonian containing these two terms can be written as
\begin{equation}
    H_{q,x}^{\Box} = \frac{h^3}{4\pi}\left(\frac{c^2 P^2}{2} + \frac{Q^2}{2h^2}\right), 
\end{equation}
 which corresponds to a Harmonic oscillator with $m=c^2$ and $\omega^2 = 1/(c^2 h^2)$.  

To see that dynamics cannot introduce new contractable loops, we simply need to show that a system prepared initially in the zero-eigenstate of $H^{\Box}_{q,x}$ remains in such a state for the entire evolution.  
Let us assume that we begin with a system prepared in a minimum uncertainty coherent state centered at $Q=P=0$ in phase space.  The dynamics of this minimum uncertainty state centered will have $\langle Q(t)\rangle=O(Lt/\Lambda)$, once considering the errors in our dynamics and using the fact that $\|e^{-iHt} - e^{-i\tilde{H}t}\|\le \|H-\tilde{H}\|t$, for all $t$ using standard expressions for the dynamics of a coherent state.  The standard deviation of this distribution is $O(1/(m\omega)) = O(h/c)$.  Thus in the limit where $c\rightarrow \infty$ for fixed $h$ we have that for an initial state with mean position $0$ and constant $\mathcal{K}$, 
\begin{equation}
    P(|Q|\ge \mathcal{K})\in O(h^2/(\mathcal{K}^2 c^2)).
\end{equation}  This implies that for the coherent state $P(|E^{\Box}| \ge \Delta') \in O(1/(\Delta'^2 c^2))$.  If we consider the error in the dynamics, then we have that $P(|\tilde{E}^{\Box}(t)| \ge \Delta') \in O(1/(\Delta'^2 c^2) + \frac{E_{\max} L t}{\Lambda})$.  From this we see that the probability that for the dynamics in the large $c$ limit, an initial state that has no contractable loops will remain without contractable loops.  Thus these errors will not be visible in the limit of large $\Lambda$ and $c$ for fixed $E_{\max}$ for any $q$.  

Finally, let $C$ be a generic contractable loop in the $z-y$ plane.  As shown above, any contractable loop that satisfies Gauss' law can be written as a sum of contractable loops.  Specifically, this means that we can express the field along the loop as $\sum_{q\in C} h E_{q,x}^{\Box}$.  Thus by linearity, the evolution of the sum over the contractable loop cannot change its initial eigenvalue of zero in the limit as $c\rightarrow \infty$ and $\Lambda\rightarrow \infty$ for fixed $E_{\max}$ because the evolutions are independent.  Specifically, from the union bound the probability that there is at least one site that violates a given level of error tolerance in $\Delta'$ is 
\begin{equation}
    {\rm Pr}({\rm contract. loop})\in O(\Omega^{2/3}/(\Delta'^2 c^2 h^2) +\Omega^{2/3}E_{\max} L t/(\Lambda h^2)  ) 
\end{equation} which vanishes for fixed $h$ and $E_{\max}$ in the limit as $c\rightarrow \infty$.
\end{proof}

Above we argue that with high probability, no field loops can be created in the large $c$ limit.  We state this observation as the following corollary.
\begin{corollary}
    Under the assumptions of Proposition~\ref{prop:contract} Lemma~\ref{lem:nonContract} we have that the probability that an excitation in the field appears that leads to a field loop in a configuration that already obeys Gauss' law vanishes for any fixed $t$ as $\Omega, c,\Lambda \rightarrow \infty$.
    
\end{corollary}

Finally, we want to argue how the Coulomb interaction emerges from non-relativistic quantum electrodynamics to illustrate that, in appropriate limits, we can use the full theory to simulate individual charged particles.  Let us again only focus on a single particle source.  We choose this because the Coulomb field does not directly appear from multiple charges, but can be found using the principle of superposition of the underlying fields. First, from the above corollary we see that in the appropriate limit that $\bra{\psi(t)}\oint_C E\cdot ds \ket{\psi(t)} =0$ for any simple closed loop $C$.  This implies that $\bra{\psi}\oiint_C \frac{\partial B(t)}{\partial t}\cdot dA\ket{\psi} = 0 $ for an initial state $\ket{\psi}$ that satisfies Gauss' law.  Thus the fact that the integral of Electric field is zero implies that the derivative of the magnetic field approaches zero.  This implies that if $\ket{\psi}$ has negligible magnetic field in it then in the limit as $c\rightarrow \infty$, $\omega\rightarrow \infty$ and $\Lambda\rightarrow \infty$ that it will also have a negligible magnetic field at each field site.

Next, note that if the magnetic component of the Hamiltonian can be neglected for all input states in the above limit, then the remaining terms in the Hamiltonian are rotationally symmetric in the continuum limit.  This means that we can perform a rotation of the sites and the underlying field, and the Hamiltonian retains the same form.  Thus, if we begin with a state that is rotationally symmetric, it will remain rotationally symmetric over the time evolution.  The Coulomb field is rotationally symmetric and thus if the initial state is chosen to satisfy the Coulomb Hamiltonian, then the rotational symmetry of the Coulomb field cannot be changed under evolution.  The prior arguments showed that no excitations that are contractable loops, uncontractable loops, or non-closed simple paths can be created without violating the symmetries of the initial state or Gauss' law in the limit where $\Lambda, c, \Omega \rightarrow \infty$.  This implies that the field will remain symmetric in this limit.  Then if the field is symmetric, the Gauss' law constraint implies that for a $2$-sphere of radius $R$ about the charge we have that the surface integral over this region will approach
\begin{equation}
    \oiint E\cdot dA = \zeta/4\pi~\Rightarrow E = \frac{\zeta r}{|r|^3}
\end{equation}
which shows that the Coulomb potential emerges out of non-relativistic quantum electrodynamics in the limit of large system size, cutoff and speed of light.  Further, the protection against spurious electromagnetic fields that could break the symmetry results from our topological protection against non-contractable loops that arises naturally from the constrained Pauli-Fierz Hamiltonian.

If an initial state is provided that satisfies Gauss' law then we can see that in the continuum limit that the Hamiltonian remains invariant under rotations of the field and the particle grids.  Thus if we have this property and a rotationally invariant input state $\ket{\psi_0}$, then a rotation operator for the field and particle grid about any axis $\mu$, $R_{\mu}(\theta)$ obeys
\begin{equation}
    R_{\mu}(\theta) e^{-iHt} \ket{\psi_0} = e^{-i Ht} R_{\mu}(\theta)\ket{\psi_0} = e^{-iHt} \ket{\psi_0}
\end{equation}
Thus, a rotationally invariant initial state remains rotationally invariant under this Hamiltonian, and in the continuum and non-relativistic limits Coulomb's law emerges from the assumptions laid out in this section.  We leave a discussion of the precise role that discretization has on the rotational symmetry for future work.

An interesting corollary of this work is that because both contractable and non-contractable loops cannot arise from the dynamics in the non-relativistic and thermodynamic limits, we have that measurement of these quantities can be used to actively detect and reject these errors.  We leave a detailed discussion of the advantages of doing so, and the possibility of applying quantum error correction to undo such simulation errors, for future work as well.

\section{Cost Comparison with Coulomb Potential Simulations}\label{sec:compare}
The majority of quantum chemistry is, in essence, a study of the quantum mechanics of the Coulomb Hamiltonian.  However, we know fundamentally that the Coulomb Hamiltonian is an imperfect description of nature.  It does not properly account for magnetic effects and it does not acknowledge the fact that the speed of light is finite.  This means that the Coulomb interaction is clearly not instantaneous.  Here, we consider using the constrained Pauli-Fierz Hamiltonian discussed above, to simulate a more complete theory of chemistry and find, surprisingly, that it also can lead to improved asymptotic scaling \emph{in certain regimes} within the thermodynamic limit.
We address the discretization error from not implementing the Coulomb interaction directly on the particle lattice, and instead its effect from the local Gauss' law constraints.

While the Gauss' law constrained Pauli-Fierz Hamiltonian describes physics beyond the Coulomb Hamiltonian, it is still insightful to discuss how our algorithm scales in simulating Coulomb interactions compared to state of the art methods. 
Assuming further that $M\in \widetilde{O}(\eta)$, $a\in O(\log(1/\epsilon))$, $c,\Lambda \in \Theta(1)$, $\sum_i|\zeta_i| \in O(\eta)$ and suppressing logarithmic factors, we have that our method can perform the simulation with a number of non-Clifford operations as given by Theorem~\ref{thm:gatecount} to be
 \begin{eqnarray}
   \widetilde{O}\left( \eta t\left(\frac{Mc^2}{\eta} + \frac{1 }{\Delta^2}\right)\left(M(b^3+a^2)+\sum_i|\zeta_i| \right)\log^2(N\Lambda)\log^2(1/\epsilon)   \right) \rightarrow \widetilde{O}\left( \frac{\eta^2}{\Delta^2} \left(b^3 + \log^2(1/\epsilon) \right)t \log^2(1/\epsilon) \right) 
 \end{eqnarray}
 Expanding with the original definition $\Delta = \frac{\Omega^{1/3}}{N^{1/3}}$ this becomes
 \begin{eqnarray}
    \widetilde{O}\left( \frac{N^{2/3}\eta^2}{\Omega^{2/3}} \left(b^3 + \log^2(1/\epsilon) \right)t \log^2(1/\epsilon) \right)
 \end{eqnarray}
 Additionally, if we assume that $b\in O(\log(1/\epsilon))$ as is expected from high-order integrator formulas the cost becomes
\begin{eqnarray}
    \widetilde{O}\left( \frac{N^{2/3} \eta^{2}}{\Omega^{2/3}}t\log^5(1/\epsilon) \right) .
 \end{eqnarray}
 The thermodynamic limit of a large number of electrons, where $\eta \propto \Omega$ is an important limit to consider~\cite{2021_SBWetal}. In this thermodynamic limit, the scaling becomes
\begin{eqnarray}
    \mathcal{C}_{\rm Therm., QED}\in \widetilde{O}\left( N^{2/3} \eta^{4/3}   t \log^5(1/\epsilon) \right).
 \end{eqnarray}
 Interestingly, this shows that if $N\in \widetilde{o}(\eta^4)$ then the cost of simulating full non-relativistic quantum electrodynamics (and $c\in O(1)$ relative to the number of particles) can become lower than simulation of the Coulomb Hamiltonian, described in more detail below.

We find that because the Coulomb interaction is mediated through updating the electric field with the local Gauss' law constraint, we are able to achieve sub-quadratic scaling in the number of particles $\eta$, which is a unique feature of our method. However, the trade-off comes is the scaling of $N$. Compare the gate complexity above to the standard Coulomb Hamiltonian simulation in first-quantization \cite{babbush2019quantum, su2021fault}, using the interaction picture simulation algorithm in the frame of the kinetic energy,  where the simulation scales in the thermodynamic limit as
 \begin{eqnarray}
    \mathcal{C}_{\rm Therm., Coulomb}\in \widetilde{O}\left(  N^{1/3} \eta^{8/3} t \log^2(1/\epsilon)\right).
 \end{eqnarray}
The gate complexity of our approach in the ordinary condensed matter regime where $N \gg \eta$ is better (up to logarithmic factors) than existing approaches if $N\in \tilde{o}(\eta^{4})$ and $\epsilon$ is constant. 
Recall that the number of spatial grid points considered is independent of the number of particles in first-quantized methods. 
This scaling advantage is therefore expected to hold in the regime of weak interactions between subsystems, as the required number of grid points, $N$, scales nearly linearly with the number of particles, $\eta$, in that limit. Therefore, this suggests that if we are working in this limit, then the simulation of quantum electrodynamics with an emergent Coulomb interaction can provide a substantial advantage over the Coulomb approach (ignoring the fact that the non-relativistic electrodynamics approach we propose is actually more general and asymptotically accurate than the conventional quantum chemistry Hamiltonian).

 In summary, the key distinction between these approaches lies in the treatment of the explicit Coulomb interactions. 
 In the standard treatment in this regime, one does not worry about the scaling of the number of field modes present because an effective potential that takes these interactions into account.  In our algorithm, we need to model an enormous number of field modes in order to account for the dense collection of charges that can arise here.  Even with this, it is not clear that the worst case scaling of our algorithm can be saturated due to the Pauli-exclusion principle and a more detailed analysis of the discretization errors may be needed to understand the relative looseness of these bounds in the degenerate regime.  Regardless, at present it is fair to say that these above examples suggest that a QED simulation, such as our method, may be asymptotically superior for condensed matter simulations in the dilute regime but is unlikely to surpass existing methods in the degenerate regime where proliferation of the number of electromagnetic modes is a bottleneck to such algorithms.

\section{Conclusion}\label{sec:conclusion}
We provide an approach for simulating non-relativistic quantum electrodynamics on a quantum computer.  Our approach borrows from prior work on simulating the Pauli-Fierz-Coulomb model~\cite{2023_MSW} in first-quantization, but uses a constraint to impose Gauss' law to introduce interactions between charges rather than requiring an explicit Coulomb interaction be added to the model, referred to as the constrained Pauli-Fierz Hamiltonian.  Our approach uses an interaction picture simulation in the frame of the Gauss' law constraint operator and the electric components of the energy to perform the simulation similar to Ref.~\cite{2022_RRW}.  The implementation of the constraint evolution requires several novel constructions involving coherent quantum sorting algorithms to efficiently implement the operator in cost that scales near linearly in the number of field cells.  

We find that  the dynamics of this theory can be simulated using a number of non-Clifford operations that scales in an appropriate thermodynamic limit where the number of field cells is proportional to the number of charged particles as $\widetilde{O}(N^{2/3} \eta^{4/3} t \log^5(1/\epsilon))$.  In contrast, in the analogous limit  the Coulomb Hamiltonian alone scales as $\widetilde{O}(N^{1/3} \eta^{8/3} t\log^2(1/\epsilon))$.  This provides quadratically better scaling with $\eta$ but quadratically worse scaling with $N$ than can be attained using existing first-quantized simulation schemes.

We further investigate the prospect of using our simulation strategy as a replacement for the Coulomb Hamiltonian in the non-relativistic limit.  We show that, under appropriate smoothness assumptions of the continuous limit of the discrete wave function, that the emergence of the Coulomb potential can be thought of as occurring because of a form of topological protection.  Specifically, we show that in the limit where $c\rightarrow \infty$ that Gauss' law forbids any configuration of fields that either is not a closed loop or does not terminate on the boundary.  We then use this to define a code space defined by the $+1$ eigenspace of the Gauss' law projector and the projector onto the space that has no contractable field loops.  We show that in the limit as the volume of the simulation increases that, $c\rightarrow \infty$ and certain continuity assumptions that neither contractable nor uncontractable field loops can can emerge by symmetry constraints or energy constraints.  We then argue that if the initial distribution satisfies Coulomb's law then so too will the evolved state.  Thus, we conclude that our theory can be used, in an appropriate limit, in place of the Coulomb Hamiltonian to simulate systems such as chemistry.  This suggests that this approach may provide a path forward to improved scaling for chemistry in the thermodynamic limit through the explicit introduction of quantum fields.

The most obvious open question in this work surrounds the nature of the comparison between our approach and chemistry simulation using the Coulomb Hamiltonian.  While we argued the Coulomb Hamiltonian can emerge in the correct limit, a proper comparison between our approach and the Coulomb Hamiltonian would need to address the precise size of the electric field grid and the volume needed to provide an $\epsilon$-approximation to the correct Coulomb dynamics.  These discretization errors are famously difficult to prove and are essential for understanding the complexity tradeoffs involved in using this strategy to simulate chemistry.  If such a discretization bound could be proved to be reasonably small, however, it would suggest that quantum electrodynamics may actually provide a computationally simpler model of chemistry than the Coulomb Hamiltonian while also allowing the inclusion of electromagnetic effects not present in the Coulomb Hamiltonian.

One interesting issue is that the scaling of the two methods differ with $N$ arises because of the different constraints on the interaction frame here and in the work of Refs.~\cite{babbush2019quantum, su2021fault}.  The interaction frame is chosen here to be that of the Gauss' law projector.  As this operator does not commute with the kinetic operator, we cannot perform the transformation to a rotating frame for both operators without losing the fast-forwardability that underpins the practicality of the interaction picture transformation.  Such optimizations may be possible to incorporate into our framework, but improved methods of constraining the dynamics or low-energy assumptions on the kinetic operator would be needed to provide the improved scaling with $N$ observed in other simulation techniques.

A further question that merits investigation involves exploiting the topological protection that we observed.  Since excitations that correspond to contractable paths or closed loops are protected against by energetic or symmetry assumptions, measurements of these properties can be used to protect against simulation or physical errors present in the simulation.  This could be significant, as the extra memory overheads needed to represent the electric field grid could potentially be offset by a lessened need for quantum error correction.

More broadly, this work raises a question about our assumptions about computational simplicity in nature.  From the perspective of classical computing, quantum simulation using the Coulomb Hamiltonian is far more tractable than simulating the vastly larger Hilbert space of non-relativistic quantum electrodynamics.  However, this work reveals that the quantum complexity can follow the opposite pattern because large Hilbert spaces are not a driver of complexity in quantum algorithms.  If we further take the perspective that the universe is computationally equivalent to a large quantum computer, then it would be surprising if the fundamental models that describe the system were not optimized to require minimal computational resources to execute on such a quantum computer.  If this conjecture is true, then this pattern should reveal itself for the simulation of even more fundamental quantum theories than this one, and may perhaps reveal that at a foundational level, physical law should be maximally computationally efficient on any physically realistic model of computing.  Probing such questions may prove to be a fruitful avenue of inquiry in the foundations of physics and computer science.

\begin{acknowledgements}
We thank Priyanka Mukhopadhyay for her extensive feedback on the research as well as her assistance writing an earlier version of this draft and for providing key ideas for optimization of the underlying circuits used here.  NW would also like to thank Kianna Wan and Dominic Berry for significant discussions about the computational overheads of computing pairwise interactions in the Coulomb Hamiltonian. TFS would like to also thank Subhayan Roy Moulik for fruitful discussions in the conceptualization stage of this work.  This work was primarily supported by the U.S. Department of Energy, Office of Science, National Quantum Information
 Science Research Centers, Co-design Center for Quantum Advantage (C2QA) under contract number DE
SC0012704 (PNNL FWP 76274). TFS acknowledges partial support as a Quantum Postdoctoral Fellow at the Simons Institute for the Theory of
Computing, supported by the U.S. Department of Energy, Office of Science, National Quantum Information
Science Research Centers, Quantum Systems Accelerator. NW’s research is also
 supported by PNNL’s Quantum Algorithms and Architecture for Domain Science (QuAADS) Laboratory Directed Research and Development (LDRD) Initiative.
 The Pacific Northwest National Laboratory is operated
 by Battelle for the U.S. Department of Energy under
 Contract DE-AC05-76RL01830.
\end{acknowledgements}
\bibliography{ref}

\appendix

\section{Discrete Couloumb Gauge Term}\label{app:cgauge}
An interesting consequence of the use of a field-theoretic description of physics is the introduction of Gauges, which are unobservable choices of representation in the underlying theory that nonetheless need to be preserved in order to maintain consistency for the problem.  Here we argue that under appropriate continuity assumptions that the Coulomb gauge is preserved under the dynamics of the constraint for the discrete dynamics considered above in the limit where the cutoff $\Lambda$ tends to infinity.
\begin{lemma}\label{lem:coulombGauge}
For the Hamiltonian $H\in L(\mathcal{H})$ given in Definition~\ref{def:Ham}, let $\ket{\phi}\in \mathcal{H}$ be a quantum state within a set of states, $S_{\phi}$ such that
\begin{enumerate}
    \item $\sup_{q,\mu, S_{\phi}}\|A_{q,\mu} \ket{\phi}\| \in O(\Lambda^{\gamma}/E_{\max})$ for universal constant $\gamma \in [0,1)$,
    \item the discrete Coulomb Gauge integral term yields  $\sup_{\ket{\phi} \in S_{\phi}, \mathcal{D}_{q;b}}|\bra{\phi} \oiint_{D_q ;b} A \cdot dS \ket{\phi}|^2 =0$ where the integral refers to an integral  over the cube $\mathcal{D}_{q;b}$ using a $2b+1$-point Newton-Cotes formula,
    \item there exists $\aleph\in \mathbb{R}$ such that for any $q,\mu$ in the cubic lattice $\|\bra{\phi} U_{q,\mu} - \openone \ket{\phi}\| \le \aleph/\Lambda $.
\end{enumerate}  The expectation value of the Hamiltonian and the approximate Coulomb gauge constraint within the set of states $S_\phi$
    $$
    \bra{\phi}[H,H_{Ac}(q)]\ket{\phi}:=\bra{\phi}\left[H,\frac{h^2}{2}\sum_{\mu=0}^2 \sum_{u \in \mathcal{S}_{q,\mu;b}}\beta_{u} \openone\otimes(A_{q+be_{\mu}+ u,\mu} - A_{q-be_\mu+ u,\mu} )\right]\ket{\phi}
    $$
    obeys
    \begin{equation}
        \sum_q^M| \bra{\phi}[H,H_{Ac}(q)]\ket{\phi}| \in O\left(\frac{Mb^2h^2}{\Lambda}\left(\aleph + \Lambda^\gamma\right) \right).
    \end{equation}
\end{lemma}
\begin{proof}
First we note that $H= H_f + H_c$.  Then because the kinetic and the magnetic parts of the Hamiltonian are diagonal in the eigenbasis of the magnetic vector potential
\begin{equation}
    \left[  \sum_{i=0}^{\eta-1}\frac{1}{2m_i} \sum_{\mu=0}^2  \left(  -i \nabla_{i,\mu}\otimes \openone  -  \zeta_i \sum_{q=0}^{M-1}  \openone\otimes\delta_{x\in D_q} A_{q,\mu} \right)^2 + \frac{c^2 h^3}{8\pi}\sum_{q=0}^{M-1}  \openone\otimes |\nabla_q \times A_{q}|^2, \,\,\, H_{Ac}\right]=0.
\end{equation}
However, the commutator with the electric field operator is not necessarily zero.  We have for a given electric field site $\chi$
\begin{align}
    [E_{\chi,\mu}, H_{Ac}] &= \left[E_{\chi,\mu},\sum_{u\in \mathcal{S}_{q,\mu;b}} \beta_u (\openone\otimes (A_{q+be_{\mu}+u,\mu} - A_{q-be_{\mu}+u,\mu})\right]
\end{align}
The exact expression for the commutator is cumbersome, but from the definition of $A$ as the generator of displacements of electric field we have that
\begin{equation}
    e^{iA_{q,\mu} E_{\max}/\Lambda} = U_{q,\mu}
\end{equation}
Then from Taylor's theorem if we assume that we restrict ourselves to input states on which $\| A_{q,\mu} \ket{\psi}\| \in O(\Lambda^\gamma/E_{\max})$ for universal constant $\gamma>0$ then
\begin{equation}
    [E_{q,\mu}, U_{q,\mu}]= \frac{iE_{\max}}{\Lambda} [E_{q,\mu}, A_{q,\mu}] + O\left(\frac{E_{\max}}{\Lambda^{2-\gamma}} \right)
\end{equation}
By the definition of the displacement operator
\begin{equation}
    [E_{q,\mu},U_{q,\mu}] = \frac{E_{\max}}{\Lambda} \sum_{\epsilon} \ketbra{\epsilon+1}{\epsilon}.
\end{equation}
Thus we have that,
\begin{equation}
    [E_{q,\mu},A_{q,\mu}] = -i \sum_{\epsilon} \ketbra{\epsilon+1}{\epsilon} + O(1/\Lambda^{1-\gamma})=-i U_{q,\mu} + O(1/\Lambda^{1-\gamma}).
 \end{equation}
Thus under the above continuity assumptions that for all states in $S_{\phi}$
\begin{align}
    |\bra{\phi} U_{q+b,\mu} -U_{q-b,\mu} \ket{\phi}| \le |\bra{\phi} U_{q+b,\mu} -\openone \ket{\phi}| + |\bra{\phi} \openone -U_{q-b,\mu} \ket{\phi}| \le 2\aleph/\Lambda.
\end{align} 
This in turn shows from Lemma~\ref{lem:2DCotes},
\begin{align}
    \bra{\phi}[E_{\chi,\mu},H_{Ac}] \ket{\phi}&=\bra{\phi}\left[E_{\chi,\mu},\sum_{u\in \mathcal{S}_{q,\mu;b}} \beta_u (\openone\otimes (A_{q+be_{\mu}+u,\mu} - A_{q-be_{\mu}+u,\mu})\right]\ket{\phi} \nonumber\\
    &=\sum_{u\in S_{q,\mu;b}} \bra{\phi}\beta_u \left((\openone \otimes (-iU_{q+be_\mu+u,\mu} +iU_{q-be_\mu+u,\mu})) +O(1/\Lambda^{1-\gamma})\right)\ket{\phi}\nonumber\\
    &=O\left(\frac{b^2}{\Lambda}\left(\aleph + \Lambda^{\gamma} \right)\right).
 \end{align}
 Finally, after summing over all sites $\chi$ and directions $\mu$ we see that
 \begin{equation}
     \|\bra{\phi}[H,H_{Ac}]\ket{\phi}\| = O\left(\frac{Mb^2h^2}{\Lambda}\left(\aleph + \Lambda^\gamma\right) \right).
 \end{equation}
\end{proof}
 This shows that if we assume that $M\in \Theta(\Omega/h^3)$ then the result vanishes as $h\rightarrow 0$ if $\Lambda \in o(1/h)$ implying that we do not need to explicitly use a constraint to impose the Coulomb gauge if we increase to spatial and field discretization scale sufficiently.

\section{LCU decompositions of operators}
\label{app:lcu}
This section provides for completeness, explicit LCU decompositions borrowed from~\cite{2023_MSW} and elsewhere  for the terms that arise in the constrained Pauli-Fierz Hamiltonian.  These are essential for our algorithm as it gives us the coefficient one-norm and the costs of the prepare and select subroutines needed for the truncated dyson series.

The following LCU decomposition of $A$ and $A^2$ immediately follows  from Corollary 26 and 27 in \cite{2023_MSW}
\begin{lemma}[Decomposition of $A$ and $A^2$]
 Let $d=2\Lambda=2^{\xi}$. Then we can write for $c_0 = (d-1)/2$
 \begin{eqnarray}
  A_{q,\mu}=\frac{\pi}{E_{\max} h}\left(c_0\openone-\mathcal{F}_{q,\mu}\left(\sum_{i=0}^{\xi-1}2^{i-1}\Z_{(i+1)}\right)\mathcal{F}^{\dagger}_{q,\mu}\right) \nonumber
 \end{eqnarray}
 and 
 \begin{eqnarray}
    A^2_{q,\mu}=  \frac{\pi^2}{d^2h^2}\mathcal{F}_{q,\mu}\left(\left(c_0^2+\sum_{i=0}^{\xi-1}2^{2i}\right)\openone-2c_0\sum_{i=0}^{\xi-1}2^i\Z_{(i+1)}+2\sum_{i=0}^{\xi-2}\sum_{j=i+1}^{\xi-1}2^{i+j}\Z_{(i+1)}\Z_{(j+1)}\right)\mathcal{F}^{\dagger}_{q,\mu}  \nonumber
\end{eqnarray}
with coefficient $1$-norms that obey \begin{align*}
\|A_{q,\mu}\|_{\ell_1} \le \frac{2\pi}{h}, \qquad \|A_{q,\mu}^2\|_{\ell_1} \le \frac{4\pi^2}{h^2}.
\end{align*}
\label{cor:A_lcu}
\end{lemma}

We approximate the gradient and laplacian  operators with $(2a+1)$-point central difference formula, by which we get decompositions as sum of adders. In general the approximation error is given by the following expression,
\begin{lemma}[Laplacian Decomposition \cite{2005_L, 2017_KWBA}]
 $$
 \nabla_{\mu}^2\psi(x)=\frac{1}{\Delta^2}\sum_{k=-a}^ad_{2a+1,k}\psi(x+k\Delta\hat{e}_{\mu})+\mathcal{R}_{2a+1}
 $$
 where $\hat{e}_{\mu}$ is the unit vector along the $\mu^{th}$ component of $x$, $(x+k\Delta\hat{e}_{\mu})$ is evaluated modulo the grid length $L$, $\mathcal{R}_{2a+1}\in O(h^{2a-1})$ and
 $$
    d_{2a+1,k\neq 0}=\frac{2(-1)^{a+k+1}(a!)^2}{(a+k)!(a-k)!k^2}\qquad d_{2a+1,k=0}=-\sum_{k=-a,k\neq 0}^ad_{2a+1,k}.
 $$
 \label{lem:lcuNabla2}
\end{lemma}
More specific bounds can be made on the truncation error for these central difference formulas through the use of Taylor's remainder theorem.  Also the bounds on the 1-norm of the coefficients for the adder decomposition is  summarized below.
\begin{lemma}[Theorem 7 and Lemma 6 in \cite{2017_KWBA}]
 Let $\psi(x)\in\cmplx^{2a+1}$ on $x\in\real$ for $a\in\intg_{+}$. Then the error in the $(2a+1)$-point centered difference formula for the second derivative of $\psi(x)$ evaluated on a uniform mesh with spacing $h$ is at most
 $$
 \left|\mathcal{R}_{2a+1}\right|\leq \frac{\pi^{3/2}}{9}e^{2a[1-\ln 2]}\Delta^{2a-1}\max_x\left|\psi^{(2a+1)}(x)\right|.
 $$
 Also, the sum of the norms of the coefficients is bounded above as follows.
 $$
    \sum_{k=-a,k\neq 0}^a\left|d_{2a+1,k}\right|\leq \frac{2\pi^2}{3}.
 $$
 \label{lem:d}
\end{lemma}
Thus $\nabla^2$ is approximated by a sum of $2a+1$ adders and the $\ell_1$ norm of the coefficients is at most $\frac{4\pi^2}{3h^2}$.

Next we need a similar expression for the gradient so that we can understand how to block encode the result as a function of the number of points used in the decomposition.  The first result stems from earlier work by Li \cite{2005_L} which gives a high order derivative expression using centered differences.
\begin{lemma}[Gradient Decomposition \cite{2005_L}]
 $$
    \nabla_{\mu}\psi(x)=\frac{1}{\Delta}\sum_{k=-a}^ad_{2a+1,k}'\psi(x+k\Delta\hat{e}_{\mu})+\mathcal{R}_{2a+1}'
 $$
 where $\hat{e}_{\mu}$ is the unit vector along the $\mu^{th}$ component of $x$, $(x_{\mu}+k\Delta\hat{e}_{\mu})$ is evaluated modulo the grid length $L$, $|\mathcal{R}'_{2a+1}| \in O(\Delta^{2a})$ and
 $$
    d_{2a+1,k}'=\frac{(-1)^{k+1}(a!)^2}{j(a-k)!(a+k)!}\qquad d_{2a+1,0}'=0.
 $$
 \label{lem:lcuNabla}
\end{lemma}
Next there we need to bound the one-norm of this formula, the bound for which is given below.
\begin{lemma}[\textbf{Lemma 31 in \cite{2023_MSW}}]
The sum of the norms of the coefficients in the $(2a+1)$-point centered finite difference formula is bounded above as follows.
 \begin{eqnarray}
    \sum_{k=-a,k\neq 0}^a\left|d_{2a+1,k}'\right|&\leq& 2\ln a+\gamma\qquad\text{where } \gamma\approx0.577\text{ is the Euler-Mascheroni constant.}     \nonumber \\
    &\leq&\ln 2a^2 \qquad\text{when }a\geq\sqrt{e}\approx 1.4  \nonumber
 \end{eqnarray}
 \label{lem:d'}
\end{lemma}

Finally, for completeness we prove a specific truncation bound on the finite difference approximation to the gradient.
\begin{lemma}[\textbf{Lemma 32 in \cite{2023_MSW}}]
 Let $\psi(x)\in\cmplx^{2a+1}$ on $x\in\real$ for $a\in\intg_{+}$. Then the error in the $(2a+1)$-point centered difference formula for the first derivative of $\psi(x)$ evaluated on a uniform mesh with spacing $h$ is at most
 $$
   \left|\mathcal{R}_{2a+1}'\right|\leq \frac{(2\ln a+\gamma)}{6\sqrt{\pi}}e^{2a[1-\ln 2]}\Delta^{2a+1}\max_x\left|\psi^{(2a+1)}(x)\right|.
 $$
 \label{lem:R'}
\end{lemma}

Thus $\nabla$ can be written as sum of $2a$ unitaries, which are adders, and the $\ell_1$ norm of the coefficients is at most $2\ln a+\gamma\leq \ln (2a^2)$. 

As a final note, we see here that the accuracy of the discrete derivatives considered increases as we increase the number of points used in the formula provided that the underlying wave function is sufficiently smooth.  For our purposes, we will not discuss in detail the specific value of $a$ that is optimal and will assume that it is a constant.  This is because it is in general difficult to provide bounds on the values of the higher-order derivatives of the wave function as a function of the evolution time.  While high order formulas in principle can be valuable to address accuracy concerns, we need such guarantees in order to understand the optimal order to take for a particular evolution.  This is especially relevant since the initial state is not necessarily in $C^\infty$ and thus the asymptotic advantages may disappear for classes of functions that are not sufficiently smooth.  For these reasons, we leave detailed discussion of the truncation error to subsequent work and focus on the case where $a$ is a constant.

 In the electric link basis, $E_{\ell,\mu}^2=\sum_{\epsilon=-\Lambda}^{\Lambda-1}\epsilon^2\ket{\epsilon}\bra{\epsilon}_{\ell,\mu} $, is a diagonal positive integer matrix and so we can use Lemma 21 in Ref.~\cite{2023_MSW} to express it as a linear combination of at most $1+\lceil\log_2(\Lambda^2+1)\rceil\approx 2\log_2\Lambda$ unitaries and the $\ell_1$ norm of the coefficients is at most $\Lambda^2$. Alternatively, $E^2$ can be expressed as linear combination of slightly more number of Z operators, but with the same $\ell_1$ norm \cite{2020_SLSW}.
\begin{eqnarray}
 E^2=\frac{1}{6}\left(2^{2\xi-1}+1\right)\id+\sum_{j=0}^{\xi -1}2^{j-1}Z_j+\sum_{j=0}^{\xi-2}\sum_{k>j}^{\xi-1}2^{j+k-1}Z_jZ_k, \quad\text{ where }\quad\xi=\log_2\Lambda
 \label{eqn:E2_lcu_0}
\end{eqnarray}

We know that $U=\sum_{\epsilon=-\Lambda}^{\Lambda-1}\ket{\epsilon+1}\bra{\epsilon}=\exp(i\Delta A)=\mathcal{F}C\mathcal{F}^{\dagger}$, where C is the Sylvester's "clock" matrix defined as
\begin{equation}
    C = \begin{pmatrix}
            1 & 0 & 0        &\cdots & 0 \\
            0 & \omega & 0   &\cdots & 0 \\
            0 & 0 & \omega^2 &\cdots & 0 \\
            \vdots & \vdots & \vdots & \ddots & \vdots \\
            0 & 0 & 0 &\cdots & \omega^{d-1} \\
        \end{pmatrix}   \qquad [\omega = e^{2\pi i/d},\quad d=\text{dimension of }C]    \nonumber
\end{equation}

\begin{lemma}[Corollary 34 in Ref.~\cite{2023_MSW}]
Let $R(2^k)=\begin{bmatrix}1 & 0 \\ 0 & \omega^{2^k} \end{bmatrix}$ be a rotation gate. Then,
$ U=\mathcal{F}\left(\bigotimes_{k=0}^{\log_2d-1}R(2^k)\right)\mathcal{F}^{\dagger} $.
    \label{cor:U_lcu}
\end{lemma}

\section{Proof of Lemma~\ref{lem:CostPi}}\label{app:ProofCostPi}
\begin{proof}[Proof of Lemma~\ref{lem:CostPi}]
Let us divide $H_\pi$ into three terms.  First let $H_{1\pi}=\sum_{j,\mu}H_{1\pi}^{j,\mu}$,  be the kinetic term for particle $j$ in direction $\mu$.  Second let $H_{2\pi}^{j,q,\mu}$ be the magnetic term coupling the velocity of particle $j$ in direction $q$ to the magnetic field at site $\mu$, and similarly for  $H_{3\pi}^{j,q,\mu}$.  Note that the first term does not depend on $q$ because it only has a second derivative term on the particle grid, and does not depend on the field.  
We prove this lemma by first block encoding $H_{1\pi}^{j,\mu}$, $H_{2\pi}^{j,q,\mu}$ and $H_{3\pi}^{j,q,\mu}$ separately.  We achieve this first by using the following prepare operations to block encode these Hamiltonians $\prep_{1\pi}^{j,\mu}$, $\prep_{2\pi}^{j,q,\mu}$ and $\prep_{3\pi}^{j,q,\mu}$ respectively. We further define the corresponding unitary selection subroutines to be $\sel_{1\pi}^{j,\mu}$, $\sel_{2\pi}^{j,q,\mu}$ and $\sel_{3\pi}^{j,q,\mu}$ respectively. Then, using the procedure outlined in Theorem \ref{thm:blockEncodeDivConq} we block encode $H_{23\pi}^{j,\mu}$ and finally sum these block encodings to block encode $H_{\pi}$. 


First let us define the following ancillae preparation routines for the derivative / vector potential terms that appear in the expansion of the kinetic term.  We need to choose a specific discretization of the derivative operator in order to give the decomposition.  Specifically, we choose the derivative expansions to be of the form
\begin{equation}
    \nabla_{\mu}^2\psi(x)=\frac{1}{\Delta^2}\sum_{k=-a}^ad_{2a+1,k}\psi(x+k\Delta\hat{e}_{\mu})+\mathcal{R}_{2a+1}
\end{equation}
where $\mathcal{R}_{2a+1}$ is the remainder term in the second derivative expansion above.  If we take this form and define $d''$ to be the coefficients for the first derivative expansion then the $\prep$ operations for each of the terms take the form
\begin{eqnarray}
    \prep_{1\pi}^{j,\mu} &=& \left(\sum_{k=-a}^a\sqrt{\frac{|d_{2a+1,k}|}{\sum_{k}|d_{2a+1,k}|}}\ket{k+a}\right);      \nonumber \\
    \prep_{2\pi}^{j,q,\mu} &=& \left(\sum_{k_1=-a}^a\sqrt{\frac{|d_{2a+1,k_1}''|}{\sum_{k_1}|d_{2a+1,k_1}''|}}\ket{k_1+a}\right)\otimes\left(\sum_{k_2=1}^{\log_2d}\sqrt{\frac{w_{k_2}}{\sum_{k_2}w_{k_2}}}\ket{k_2}\right);   \nonumber \\
    \prep_{3\pi}^{j,q,\mu} &=& \left(\sum_{k_3=1}^{\frac{\log^2d+\log d}{2}}\sqrt{\frac{w_k'}{\sum_{k_3} w_k' }}\ket{k_3 }\right)
\end{eqnarray}
The terms in the coefficients, i.e. $d_{2a+1,k}$, $d_{2a+1,k}''$, $w_{k_2}$ and $w_k'$ are obtained from the LCU decomposition of the operators $\nabla^2$, $\nabla$, $A$ and $A^2$, respectively, as described in Lemmas \ref{lem:lcuNabla2}, \ref{lem:lcuNabla} and \ref{cor:A_lcu} of Appendix \ref{app:lcu}. The states of these ancillae are used to select and implement the unitaries in $H_{1\pi}^{j,\mu}$, $H_{2\pi}^{j,q,\mu}$ and $H_{3\pi}^{j,q,\mu}$, respectively. 
\begin{eqnarray}
    &&\sel_{1\pi}^{j,\mu}:\ket{k}\ket{\phi} \mapsto\ket{k}
    U_{j,\mu}^{a-k}\ket{\phi}    \nonumber \\
    &&\sel_{2\pi}^{j,q,\mu}:\ket{k_1}\ket{k_2}\ket{\phi}\mapsto\ket{k_1}\ket{k_2}
     U_{j,\mu}^{a-k_1}\left(A_{k_2}\right)_{q,\mu}\ket{\phi} \nonumber \\
    &&\sel_{3\pi}^{j,q,\mu}:\ket{k_3}\ket{\phi}\mapsto\ket{k_3}\left(A_{k_3}^2\right)_{q,\mu}\ket{\phi} \nonumber 
\end{eqnarray}

From Eq.~\eqref{eqn:HpiFragment}, using triangle inequality and the bound on norms of operators from Appendix \ref{app:lcu}, we observe the following, and define the following scalar values as 
\begin{align}
    \|H_{2\pi}^{j,q,\mu}\| &\leq \frac{\zeta_j}{m_j}\|\nabla_{j,\mu}\|\|A_{q_j,\mu}\| \leq \frac{2 \pi \zeta_j \ln(2a^2)}{m_j h \Delta} \leq \frac{2 \pi \ln(2a^2)}{ h \Delta} \\ 
    \|H_{3\pi}^{j,q,\mu}\| &\leq \frac{\zeta_j^2}{m_j}\|A_{q,\mu}^2\|\leq \frac{4 \pi \zeta_j^2}{m_j h^2} \leq \frac{4 \pi }{h^2}  \\
    \|H_{1\pi}^{j,\mu}\| &\leq \frac{1}{m_j}\|\nabla_{j,\mu}^2\| \leq \frac{4\pi^2}{3 m_j\Delta^2} \leq \frac{4\pi^2}{3 \Delta^2}, 
\end{align}
since  $|\zeta_j|/m_j\le 1$ (which is appropriate in atomic units for the cases of electrons and nuclei). We then define constants
\begin{eqnarray}
    \lambda_1 &=& \frac{2 \pi \ln(2a^2)}{ h \Delta}\\
    \lambda_1' &=& \frac{4 \pi }{h^2}\\
    \lambda_2 &=& \frac{4\pi^2}{3 \Delta^2}
\end{eqnarray}

Next, applying \Cref{lem:LCU} we then see that 
\begin{align}
    (\bra{0}\otimes \id)(\prep_{1\pi}^{j,\mu})^{\dagger}\cdot\sel_{1\pi}^{j,\mu}\cdot\prep_{1\pi}^{j,\mu}(\ket{0}\otimes \id)  &= \frac{H_{1\pi}^{j,\mu}}{ \lambda_2 }   \\
    \braket{ 0|(\prep_{2\pi}^{j,q,\mu})^{\dagger}\cdot\sel_{2\pi}^{j,q,\mu}\cdot\prep_{2\pi}^{j,q,\mu} |0 } &= \frac{H_{2\pi}^{j,q,\mu}}{ \lambda_1 }   \\
    \braket{ 0|(\prep_{3\pi}^{j,q,\mu})^{\dagger}\cdot\sel_{3\pi}^{j,q,\mu}\cdot\prep_{3\pi}^{j,q,\mu} |0 } &= \frac{H_{3\pi}^{j,q,\mu}}{ \lambda_1' }  
\end{align}
implying we get $(\lambda_2,.,0)$, $(\lambda_1,.,0)$ and $(\lambda_1',.,0)$ block-encodings of $H_{1\pi}^{j,\mu}$, $H_{2\pi}^{j,q,\mu}$ and $H_{3\pi}^{j,q,\mu}$, respectively.

Next, using the procedure outlined in Theorem \ref{thm:blockEncodeDivConq} we block encode $H_{23}^{j,\mu}$ using the above block encodings. Let $\nconst_1=2\lambda_1+\lambda_1'$. The ancillae preparation routine is,
\begin{eqnarray}
    \prep_{23\pi}^{j,\mu} &=& \left(\frac{1}{\sqrt{M}}\sum_{q=0}^{M-1}\ket{q}\right)\otimes\left(\frac{1}{\sqrt{2}}\sum_{r=0}^{1}\ket{r}\right) \otimes\left(\sqrt{\frac{2\lambda_1}{\mathcal{A}_1}}\ket{0}+\sqrt{\frac{\lambda_1'}{\mathcal{A}_1}}\ket{1}\right)  \nonumber \\
    &&\otimes \prep_{2\pi}^{j,q,\mu}\otimes\prep_{3\pi}^{j,q,\mu};  
\end{eqnarray}
and the unitary selection routine is,
\begin{eqnarray}
    &&\sel_{23\pi}^{j,\mu}:\ket{q,r,0}\ket{k_1,k_2,k_3}\ket{\phi}\mapsto \ket{q,r,0}\ket{k_3}R_{j,x\in D_{q}}^r \cdot\sel_{2\pi}^{j,q,\mu}\left(\ket{k_1,k_2}\ket{\phi}\right) \nonumber \\
    &&\sel_{23\pi}^{j,\mu}:\ket{q,r,1}\ket{k_1,k_2,k_3}\ket{\phi}\mapsto \ket{q,r,1}\ket{k_1,k_2}R_{j,x\in D_{q}}^r \cdot\sel_{3\pi}^{j,q,\mu}\left(\ket{k_3}\ket{\phi}\right). \nonumber
\end{eqnarray}
where $R_{j,x\in D_{q}}^r$ is the reflection operator with a control on the state of qubit $\ket{r}$. Now, we can show that we obtain a $(\nconst_1',.,0)$ block-encoding of $H_{23\pi}^{j,\mu}$, where we define $\nconst_1'$ as
\begin{equation}
    \|H_{23\pi}^{j,\mu}\|\leq M\left(2\|H_{2\pi}^{j,q,\mu}\|+\|H_{3\pi}^{j,q,\mu}\| \right)\leq\nconst_1'.
\end{equation}
Therefore, we have that
\begin{equation}
    \nconst_1' = \frac{4\pi M \ln(2a^2)}{h\Delta}+\frac{4\pi M}{ h^2}
\end{equation} and 
\begin{eqnarray}
    \braket{ 0|(\prep_{23\pi}^{j,\mu})^{\dagger}\cdot\sel_{23\pi}^{j,\mu}\cdot\prep_{23\pi}^{j,\mu} |0 } = \frac{H_{23\pi}^{j,\mu}}{ \nconst_1' }.   \nonumber
\end{eqnarray}

Now, again we use Theorem \ref{thm:blockEncodeDivConq} in order to block encode $H_{\pi}$ using the block encodings of $\frac{H_{23}^{j,\mu}}{\nconst_1'}$ and $\frac{H_{1\pi}^{j,\mu}}{\lambda_2}$. Let $\nconst_2=\lambda_2+\nconst_1'$. Thus the entire ancilla preparation routine can be described as follows.

\begin{eqnarray}
    \prep_{\pi}    &=& \left(\frac{1}{\sqrt{\eta}}\sum_{j=0}^{\eta-1}\ket{j}\right)\otimes\left(\frac{1}{\sqrt{3}}\sum_{\mu=0}^2\ket{\mu}\right)\otimes\left(\sqrt{\frac{\lambda_2}{\mathcal{A}_2}}\ket{0}+\sqrt{\frac{\nconst_1'}{\mathcal{A}_2}}\ket{1}\right)\otimes\left(\frac{1}{\sqrt{M}}\sum_{q=0}^{M-1}\ket{q}\right)\otimes\left(\frac{1}{\sqrt{2}}\sum_{r=0}^{1}\ket{r}\right)   \nonumber \\
    && \otimes\left(\sqrt{\frac{2\lambda_1}{\mathcal{A}_1}}\ket{0}+\sqrt{\frac{\lambda_1'}{\mathcal{A}_1}}\ket{1}\right)\otimes\prep_{1\pi}^{j,\mu}\otimes\prep_{2\pi}^{j,q,\mu}\otimes\prep_{3\pi}^{j,q,\mu}    \label{eqn:prepPi}
\end{eqnarray}
The entire unitary selection procedure is described as follows.
\begin{eqnarray}
    &&\sel_{\pi}:\ket{j,\mu,0}\ket{q,r,b_1}\ket{k,k_1,k_2,k_3}\ket{\phi}\mapsto\ket{j,\mu,0}\ket{q,r,b_1}\ket{k_1,k_2,k_3}\sel_{1\pi}^{j,\mu}\left(\ket{k}\ket{\phi}\right)    \nonumber \\
\text{i.e.}&& \sel_{\pi}:\ket{j,\mu,0}\ket{q,r,b_1}\ket{k,k_1,k_2,k_3}\ket{\phi}\mapsto\ket{j,\mu,0}\ket{q,r,b_1}\ket{k,k_1,k_2,k_3}\left(\nabla_{k}^2\right)_{j,\mu}\ket{\phi}   \label{eqn:selPi1} \\
    &&\sel_{\pi}:\ket{j,\mu,1}\ket{q,r,b_1}\ket{k,k_1,k_2,k_3}\ket{\phi}\mapsto\ket{j,\mu,1}\ket{k}\sel_{23\pi}^{j,\mu}\left(\ket{q,r,b_1}\ket{k_1,k_2,k_3}\ket{\phi}\right)    \nonumber \\
  \text{i.e.}  &&\sel_{\pi}:\ket{j,\mu,1}\ket{q,r,0}\ket{k,k_1,k_2,k_3}\ket{\phi}\mapsto\ket{j,\mu,1}\ket{q,r,0}\ket{k,k_1,k_2,k_3}(\nabla_{k_1})_{j,\mu}(A_{k_2})_{x,\mu}\ket{\phi}    \label{eqn:selPi2} \\
  \text{and }&&\sel_{\pi}:\ket{j,\mu,1}\ket{q,r,1}\ket{k,k_1,k_2,k_3}\ket{\phi}\mapsto\ket{j,\mu,1}\ket{q,r,1}\ket{k,k_1,k_2,k_3}(A_{k_3}^2)_{r,\mu}\ket{\phi}   
  \label{eqn:selPi3}
\end{eqnarray}
Again, using Theorem \ref{thm:blockEncodeDivConq} the normalization constant of a sum of block encodings is the sum of the normalization constants thus we define the $(\nconst_2',.,0)$ block-encoding of $H_{\pi}$, where 
\begin{equation}
    \|H_{\pi}\|\leq \frac{3\eta}{2}\left(\|H_{1\pi}^{j,\mu}\|+ \|H_{23\pi}^{j,\mu}\|\right)\leq\nconst_2'
\end{equation}
and
\begin{equation}
\nconst_2'\le \frac{2\pi^2 \eta}{\Delta^2} + \frac{6 \pi \eta M \ln(2a^2)}{h\Delta} + \frac{6 \pi \eta M }{h^2}
\end{equation}
Specifically, the block-encoding of the kinetic part of the Hamiltonian is
\begin{eqnarray}
    (\bra{0}\otimes \id)|\prep_{\pi}^{\dagger}\cdot\sel_{\pi}\cdot\prep_{\pi} (\ket{0}\otimes \id) = \frac{H_{\pi}}{ \nconst_2' }.   
\end{eqnarray}

\textbf{Cost of ancilla preparation routine $\prep_{\pi}$ : } Now we describe each of the registers in $\prep_{\pi}$ (Eq.~\eqref{eqn:prepPi}) and the cost of preparing the required state of the ancillae in terms of rotation gates.  Note in the following for brevity we take $d=O(\Lambda)$ to be the maximum value of the integer encoding of used within the field registers.

First, we list the qubit counts for each ancilla in order of left to right in Eq.~\eqref{eqn:prepPi}
\begin{enumerate}
    \item $\lceil\log_2(\eta)\rceil$ qubits 
    \item $\lceil\log_2(3)\rceil$ qubits 
    \item 1 qubit
    \item $\lceil\log_2(M)\rceil$ qubits 
    \item 1 qubit
    \item 1 qubit
    \item $\lceil\log_2(2a+1)\rceil$ qubits ($\prep_{1\pi}^{j,\mu}$)
    \item $\lceil\log_2(2a+1)\rceil$ qubits ($\prep_{2\pi}^{j,q,\mu}$ first register)
    \item $\lceil\log_2(\log_2(d))\rceil$ qubits ($\prep_{2\pi}^{j,q,\mu}$ second register)
    \item $\lceil\log_2 \left( \frac{\log_2^2(d) + \log_2(d)}{2} \right)\rceil$ qubits ($\prep_{3\pi}^{j,q,\mu}$)
\end{enumerate}

The following break down of gate costs is computed using the synthesis approach in Ref.~\cite{2016_NDW}. Both Clifford and T-gates are included for completeness but only T-gates are included in the final asymptotic cost.

The first register of $\prep_{\pi}$ has $\log_2\eta$ ancillae and we prepare an equal superposition of $\eta$ states, corresponding to the indices of the particles. We assume the cost is then
\begin{equation}
    \mathcal{C}_{\text{reg1}} = \log_2(\eta) \,\,\,\, \text{H gates}.
\end{equation}
Similarly, in the second and fourth registers we prepare an equal superposition of 3 states indexing directions and $M$ states indexing electric field links with cost
\begin{equation}
    \mathcal{C}_{\text{reg2}} + \mathcal{C}_{\text{reg4}} = \log_2(M+2) \,\,\,\, \text{H gates}.
\end{equation}
The fifth register is used to control the identity and reflection that removes any interaction that is not in the selected field cell denoted by the fourth register.

The third and fifth registers are 1-qubit rotation gates, required in order to combine different block encodings, as explained in Theorem \ref{thm:blockEncodeDivConq} for a cost of
\begin{equation}
    \mathcal{C}_{\text{reg3}} + \mathcal{C}_{\text{reg5}} = 2 \,\,\,\, \text{Rotation gates}.
\end{equation}
$\prep_{1\pi}^{j,\mu}$ acts on the $\approx\log_2(2a)$-qubit seventh register where we store the indices of the adders in the decomposition of $\nabla^2$ (Lemma \ref{lem:lcuNabla2}) with appropriate weights as well as controlling the reflection $R_{x\in D_q}$ which negates any part of the Hamiltonian that creates an interaction with a particle with a field from outside its local field cell $D_q$. The cost of preparing the state on the sixth register is simply 
\begin{equation}
 \mathcal{C}_{\text{reg6}} = 1 \,\,\,\, \text{H gate}.
\end{equation}
The remaining preparation operations on the seventh register can be done using 
\begin{equation}
 \mathcal{C}_{\text{reg7}} = \log_2(2a) \,\,\,\, \text{H gates}, \,\,\,\, 4a+3\log_2(2a)-7 \,\,\,\, \text{CNOT gates}, \,\,\,\, 4a-2 \,\,\,\, \text{Rotation gates}.
\end{equation}
$\prep_{2\pi}^{j,q,\mu}$ acts on the eighth and ninth registers. The eighth one has $\log_2(2a)$ qubits and stores the indices of the adders in the LCU decomposition of $\nabla$ (Lemma \ref{lem:lcuNabla}). We observe that we work with $i\nabla$ because it is Hermitian and this factor is adjusted in the weights. The ninth register has $\log_2\log_2d$ qubits and stores the indices of the Z gates occurring in the LCU decomposition of $A$ (Lemma \ref{cor:A_lcu}). To prepare these superpositions we require 
\begin{align}
 \mathcal{C}_{\text{reg8}} + \mathcal{C}_{\text{reg9}} &= \log_2 (2a)+\log_2\log_2d=\log_2(2a\log_2d) \,\,\,\, \text{H gates}, \\
 & (4a+3\log_2(2a)-7)+(2\log_2d+3\log_2\log_2d-7)=4a+2\log_2d+3\log_2(2a\log_2d)-14 \,\,\,\, \text{CNOT gates}, \\
 & (4a-2)+(2\log_2d-2)=4a+2\log_2d-4 \,\,\,\, \text{Rotation gates}.
\end{align}
Finally, $\prep_{3\pi}^{j,q,\mu,q_j}$ acts on the tenth register. Since $A^2$ is a sum of $\frac{\log_2^2d+\log_2d}{2}$ unitaries, so we prepare a $\log_2\left(\frac{\log_2^2d+\log_2d}{2}\right)$-qubit register in a superposition weighted according to the LCU decomposition of $A^2$ (Lemma \ref{cor:A_lcu}).  To do this, we require 
\begin{align}
 \mathcal{C}_{\text{reg10}} &= \log_2\left(\frac{\log_2^2d+\log_2d}{2}\right) \,\,\,\, \text{H gates}, \\
 & \log_2^2d+\log_2d-3\log_2\left(\frac{\log_2^2d+\log_2d}{2}\right)-7 \,\,\,\, \text{CNOT gates}, \\
 & \log_2^2d+\log_2d-2 \,\,\,\, \text{Rotation gates}.
\end{align}

Thus for the complete ancilla preparation procedure of $\prep_{\pi}$ we require $O\left( \log_2(a\eta M\log^2 d)  \right)$ ancillae, an equivalent number of H gates, and $O(a+\log^2d)$ rotation and CNOT gates. The rotation gates can be implemented using 
\begin{equation}
     \mathcal{C}(\prep_{\pi}) = O((a+\log^2d)\log 1/\epsilon_r) \,\,\,\, \text{T-gates}.
\end{equation}
Here $\epsilon_r$ is the precision error chosen in the approximate synthesis of each rotation gate.

\textbf{Cost of unitary selection routine $\sel_{\pi}$: } Now we describe the cost of the procedure $\sel_{\pi}$, as described in Eqs.~\eqref{eqn:selPi1}-\eqref{eqn:selPi3}. For simplicity, we report only the asymptotic gate cost in Clifford+T gate count, and report them only as T-gate counts according to our assumption on the cost model that $T$-gate implementation is the computational bottleneck. Each of the $3\eta$ particle subspace and $3M$ electric link subspace has an ancilla that sets to $\ket{1}$ if the subspace is selected for the implementation of a unitary. We use $\eta$ and 3 (compute-uncompute pairs) of $C^{\log_2\eta}X$ and $C^{2}X$ gates, respectively, in order to select a particle space. That is, conditioned on the state $\ket{j}$ and $\ket{\mu}$ of the first two registers we select the subspace corresponding to particle $j$ and direction $\mu$. This can be done using $\eta$ and $3$ compute-uncompute pairs of $C^{\log\eta}X$ and $C^2X$, respectively. As mentioned earlier, using the constructions in \cite{2020_dMGM, 2023_RBMetal, 2023_MSW, 2024_Mqram} we require
\begin{equation}
     \mathcal{C}_{\text{subspace}} = O(\eta) \,\,\,\, \text{T-gates}.
\end{equation}
in order to implement this set of multi-controlled-X gates. 

If the state of the bit in the third register is $\ket{0}$ then we implement $\sel_{1\pi}^{j,\mu}$ (Eq. \eqref{eqn:selPi1}), that is we apply the operator $\nabla^2$ on this subspace. For this we select a unitary in the LCU decomposition of $\nabla^2$ (Lemma \ref{lem:lcuNabla}) depending upon the state of the qubits in the seventh register. We require $2a$ pairs of $C^{\log 2a}X$ gates in order to select a unitary indexed by an integer $\ket{k'}$ and this can be implemented with $O(a)$ gates. We apply this (controlled)-unitary on the subspace corresponding to particle $j$ and direction $\mu$. Each controlled unitary is a controlled $\log_2d$-qubit adder. Each controlled adder can be implemented with $O(\log_2d)$ gates.  We need to implement $3\eta\cdot 2a$ number of controlled adders and so we require $O(a\eta\log_2d)$ gates to implement these adders. Thus in order to implement Eq.~\eqref{eqn:selPi1} we require
\begin{equation}
     \mathcal{C}(\sel_{1\pi}^{j,\mu}) = O(a\eta\log_2(d)) \,\,\,\, \text{T-gates}.
\end{equation}

If the state of the bit in the third register is $\ket{1}$ then we implement $\sel_{23\pi}^{j,\mu}$ (Eq.~\eqref{eqn:selPi1}). First using the state in the fourth register we select an electric link subspace using $M$ pairs of $C^{\log_2M}X$ gates. The multi-controlled-NOT gates can be decomposed into
\begin{equation}
     \mathcal{C}(\text{Subspace}_{1\pi, 23\pi}) = O(M) \,\,\,\, \text{T-gates}.
\end{equation}

Then if the sixth bit is $\ket{1}$ and the state in the second and fourth registers are $\ket{\mu}$ and $\ket{q}$, respectively, then we apply the operator $A^2$ on the subspace corresponding to electric link index $q$ and direction $\mu$. Now, $A^2$ can be expressed as a sum of $ \frac{\log_2^2d+\log_2d}{2}\in O(\log^2d)$ $Z$ operators. We select a $Z$ operator depending on the state of the last register. For this we require $ \frac{\log_2^2d+\log_2d}{2}\in O(\log^2d)$ compute-uncompute pairs of $C^{\log\frac{\log^2d+\log d}{2} }X$ gates and these in turn can be implemented with 
$O(\log^2d)$ gates. So, we require 
\begin{equation}
     \mathcal{C}(A^2) = O\left(M\log^2d\right) \,\,\,\, \text{T-gates}
\end{equation}
in order to implement the operators in all electric link spaces in Eq.~\eqref{eqn:selPi3}.

Now suppose the sixth bit is 0 and the subspace corresponding to the $q^{th}$ electric link and $\mu^{th}$ direction has been selected. That, is the state of the qubits in the second and fourth registers are $\ket{\mu}$ and $\ket{q}$, respectively. Then we apply the operator $\nabla$ on the subspace corresponding to particle $j$ and direction $\mu$, depending on the state of the qubits in the first two registers. $\nabla$ can be expressed as a sum of $2a$ number of $\log_2d$-qubit adders (Lemma \ref{lem:lcuNabla2}). We select an adder depending on the state of the qubits in the eighth register. We require $2a$ pairs of $C^{\log 2a}X$ gates in order to select an adder and these can be decomposed into $O(2a)$ Clifford+T gates. Each adder can be implemented with $O(\log^2d)$ gates. Thus we require
\begin{equation}
     \mathcal{C}(\nabla) = O(a\log^2d) \,\,\,\, \text{T-gates}
\end{equation}
in order to implement all the controlled adders.

The unitary $R_{x\in D_q}$ used to zero out interactions that do not correspond to a particular field requires comparisons to evaluate whether there exists a particle in a given field cell.  We can perform this check in a similar manner to Step $7$ of the algorithm of Lemma~\ref{lem:VintImp},  We perform a sort algorithm to identify each of the positions of the corresponding locations where there is charge in precisely the same manner except we do not compute the electric flux through the cell.  As established, this process requires:
\begin{equation}
    \mathcal{C}_{sort,R} = O((M+\eta)\log(MN)\log(MN\eta))
\end{equation}
non-Clifford operations.  Then for each $q$ that we wish to compare we simply need to apply controlled-$Z$ operations to flip the phase of the given interaction if the register is non-empty.  This requires $O(\log(\eta))$ $T$-gates to test if the register is zero and then a controlled  $Z$ gate on the result controlled by the register containing $r$.  Thus the number of $T$-gates needed in total is
\begin{equation}
    \mathcal{C}(R_{x\in D_q}) = \mathcal{C}_{sort,R} + \mathcal{C}_{C-Z,R} =  O((M+\eta)\log(MN)\log(MN\eta) + M \log(\eta))
\end{equation}

Hence, the $T$ gate complexity of the procedure $\sel_{\pi}$ is in
\begin{align}
\mathcal{C}(\sel_{\pi})&=
    O\left(\eta+ a\eta\log_2d+M+M\log_2^2d+Ma\log_2d+M\log_2d +\mathcal{C}(R_{x\in D_q})\right)\nonumber\\
    &\subseteq O\left(a\eta\log_2d+M\log_2d(a+M'+\log_2d) + (M+\eta)\log(MN)\log(MN\eta)\right).
    \label{eqn:appSelPiGate}
\end{align}

Thus the overall T-count for block encoding $H_{\pi}$ is 
\begin{eqnarray}
    \g_1^t &\in& O\left((a+\log^2d)\log 1/\epsilon_r+a\eta\log_2d+M\log_2d(a+\log_2d) + (M+\eta)\log(MN)\log(MN\eta) \right)    \nonumber \\
    &\subseteq& O\left((a+\log^2(d))\log 1/\epsilon_r + (M+\eta)\left( a\log(d) + \log(MN)\log(MN\eta)\right)  \right).
    \label{eqn:g1t}
\end{eqnarray}
\end{proof}

\section{Proof of Lemma~\ref{lem:VintImp}}\label{app:vintImp}
Here we provide our proof for one of our key results which is a technical lemma that shows the complexity of computing the interaction picture Hamiltonian for the Gauss' law constraint.  The proof of the lemma is complex owing to the algorithm's liberal use of merge sort as a subroutine to avoid the multiplicative costs of $\eta$ and $M$ that plague most na\"ive implementations of the interaction term.  These swaps will allow us to make strong assumption about the locations of the charges and divergences found within the space allowing us to avoid searching to see if each particle is located within a particular Gaussian surface and in turn reduce the cost from $\tilde{O}(\eta M)$ to $\tilde{O}(\eta +M)$.  Proof follows.
\begin{proof}[Proof of Lemma~\ref{lem:VintImp}]
    Here ${H}_c$ is a diagonal operator, which means that simulation circuits can be directly constructed for it for any simulation time $t$.  We follow the standard method for simulating eigenvalues: we compute the eigenvalues into a bit register and then use a controlled rotation to apply the correct eigenphase to the eigenvalues.  
    As $[H_c,\sum_{q,\mu} I \otimes E_{q_\mu}^2 /2)k (T_{\max}/2^{n_t})]=0$ both operators share a simultaneous eigenbasis and so
    \begin{eqnarray}
        \sum_{k=0}^{2^{n_t-1}} \ketbra{k}{k}\otimes e^{-i (\lambda H_c + \sum_{q,\mu} I \otimes E_{q_\mu}^2 /2)k (T_{\max}/2^{n_t})} &=& \sum_{k=0}^{2^{n_t-1}} \ketbra{k}{k}\otimes  \left(e^{-i (\lambda H_c)k (T_{\max}/2^{n_t})}e^{-i (\sum_{q,\mu} I \otimes E_{q_\mu}^2 /2)k (T_{\max}/2^{n_t})}\right)\nonumber\\
        &=&\sum_{k=0}^{2^{n_t-1}} \ketbra{k}{k}\otimes  \left(e^{-i (\sum_{q,\mu} I \otimes E_{q_\mu}^2 /2)k (T_{\max}/2^{n_t})}\right)\nonumber\\
        &&\times \sum_{k=0}^{2^{n_t} -1} \ketbra{k}{k}\otimes e^{-i (\lambda H_c)k (T_{\max}/2^{n_t})}
    \end{eqnarray}
    Both unitaries in this decomposition are diagonal and can be simulated using standard techniques.  Specifically, we compute the eigenvalues of both unitaries and perform controlled rotations to apply the eigenphase to each of the input eigenstates.  

    For simplicity, let us assume that $\zeta_i = \zeta_j=1$ for all particles.  In the event that we have that the charge is not constant over all particles then we can perform a transformation of the form
    \begin{equation}
        \bigotimes_{\mu}\ket{x}_{i,\mu}\ket{0}^{|\zeta_i| -1} \mapsto \bigotimes_{\mu}\ket{x}_{i,\mu}\bigotimes_{\mu}\ket{x}_{i,\mu}^{|\zeta_i|-1}
    \end{equation}
    This approach is known as the replica trick and involves us simply replacing the charge with multiple copies of the same electric charge.  For the purposes of Gauss' law calculations this is the same.  This replication requires a number of ancillary qubits on the order of $O(\log(N)\sum_{i}|\zeta_i-1|)$ and can copy the bit strings using only Clifford operations.  Thus we can without loss of generality in the following we can use this trick to replace each positively charged and negatively charged particle with a larger number of particles all with $|\zeta_i|=1$.  For notational simplicity we will assume in the following that $|\zeta_i|=1$ and that $\eta$ denotes the total number of particles irrespective of the replica trick.

    We will first construct the unitary for the constraint part of the Hamiltonian.  Let $\ket{\phi}$ be an eigenstate of $H_c$.  As $H_c$ is a projection operator we have that if we define $\Sigma$ to be the set of all eigenvectors that satisfy Gauss' law then
    \begin{equation}
        e^{-i H_c\lambda k T_{\max} /2^{n_t}}\ket{\phi} = e^{-i \delta_{\phi\in \Sigma} \lambda k T_{\max} /2^{n_t} }\ket{\phi}
    \end{equation}
    Thus if we define a unitary $U_{feas}$ such that
    \begin{equation}
        U_{feas} \ket{\phi} \ket{0} = \ket{\phi} \ket{\delta_{\phi \not\in \Sigma}}
    \end{equation}
    then
    \begin{equation}
        \sum_{k} \ketbra{k}{k}\otimes e^{-i H_c \lambda k T_{\max}/2^{n_t}} \ket{\phi}\ket{0}= \begin{quantikz}
        \lstick{$\ket{k}$}&\qw &\qw& \gate{~}&\qw&\qw\\
        \lstick{$\ket{\phi}$} & \qw &\gate[2]{U_{feas}} & \qw & \gate[2]{U_{feas}^{\dagger}}&\qw \\
        \lstick{$\ket{0}$} &\qw &\qw&\gate{R_z(-\lambda k T_{\max}/2^{n_t})}\vqw{-2}&\qw&\qw
        \end{quantikz}
    \end{equation}

    The unitary $U_{feas}$ involves computing the function
    \begin{equation}
        \prod_{q=0}^{M-1}\!{\rm rect}\left({\frac{\Delta^2}{2}\sum_{\mu=0}^2 I\otimes(E_{q+e_\mu,\mu} - E_{q-e_\mu, \mu} ) - 4\pi \sum_{i=0}^{\eta-1} {\rm sign}(\zeta_i)\Pi_{q,i}} \right)
    \end{equation}
    This can be computed via a series of $(M-1)$ Toffoli operations gates combined with based on the individual evaluations of the rect function.  Next, let us define for brevity an operation $G_q$ which computes whether there is a violation of Gauss' law acting on a computational basis state for the combined field position system that is of the form
    \begin{equation}
        G_q\ket{\vec{E}}\ket{\vec{x}}\ket{0} = \ket{\vec{E}}\ket{\vec{x}}\left|{{\rm rect}\left({\frac{\Delta^2}{2}\sum_{\mu=0}^2 I\otimes(E_{q+e_\mu,\mu} - E_{q-e_\mu, \mu} ) - 4\pi \sum_{i=0}^{\eta-1} {\rm sign}(\zeta_i)\Pi_{q,i}} \right)}\right\rangle
    \end{equation}
    Then from the above discussion the total number of Clifford and $T$ gates needed to implement the operator is
    \begin{equation}
        \mathcal{C}(U_{feas}) = 8(M-1) + 2M \mathcal{C}(G_q)
    \end{equation}

    Next let us find the cost of computing $G_q$.  This requires us to build reversible circuits to compute the charge enclosed in a field cube $q$ and also the divergence of the field within that cube.  We will first consider the computation of the total charge and then the divergence.  The computation of $G_q$ is surprisingly subtle.  The central challenge is that a direct implementation of a reversible circuit to compute this function will take time $\widetilde{O}(M\eta)$.  This is because in principle every particle could be in every cell.  Such a search for a violation of Gauss' law then necessitates checking every possible pair and, unlike in  classical computing, we cannot simply measure the state to decide whether all of the particles have been correctly paired with all of the divergences in the field.  This necessitates a different approach.

    Our approach makes liberal use of sorting. 
 This is essential here to reduce the cost because sorting allows us to make guarantees about the structure of our data without measurement.  Quantum implementations of Merge Sort exist and can be used to sort a list of $L$ items using a number of comparators that is in $O(L\log(L))$~\cite{cheng2006quantum}, which saturates query lower bounds~\cite{hoyer2002quantum}. Using contemporary results on comparators~\cite{2018_G}, if each such register contains $\mathcal{B}$ bits then the number of $T$ gates needed to perform the mergesort, $\mathcal{C}_{MS}$, is in 
 \begin{equation}
     \mathcal{C}_{MS} = O(L\mathcal{B} \log(L))\label{eq:costMerge}
 \end{equation}.

The following algorithm to decide whether there is a violation in Gauss' law extensively uses the above Merge Sort algorithm.
    \begin{enumerate}
        \item For each particle in position $\ket{x_i}$ compute $\ket{x_i}\ket{B_i}$ where $\ket{B_i}$ is an $O(\log(M))$ qubit string corresponding to the cell number of the field that the particle is in.
        \item Using quantum merge sort, sort all particles by their value of $\ket{B_i}$ and append a $\lceil \log(\eta)\rceil$ bit register, denoted $Q_i$ in the following, to each of these particles initialized to the state $\ket{1}$.  This yields for any computational basis state of particle positions a state  of the form $$\ket{x_{f_1}} \ket{B_{f_1}}\ket{1}_{Q_1}\ket{x_{f_2}} \ket{B_{f_2}}\ket{1}_{Q_2}\cdots \ket{x_{f_\eta}} \ket{B_{f_\eta}}\ket{1}_{Q_\eta}\otimes \ket{junk}$$
        for a sequence $f_j$ and an ancillary state $\ket{junk}$ which stores the information needed to unsort the list.
        \item For $i$ decreasing from $\eta$ to $2$ compare $B_{f_i}$ to $B_{f_{i-1}}$ and if they are the same add $\zeta_{f_{i}}$ to register $Q_{i-1}$ and swap register $Q_i$ with an unentangled ancilla containing $\ket{0}$.
        \item Sort all registers using quantum merge sort on based on $B$. Specifically, sort each register $\ket{x_{f_i}} \ket{B_{f_i}}\ket{Q_i}$ by the value of $B_{f_i}(1-\delta_{Q_i ,0}) + \delta_{Q_i,0}$ which will order the states such that any particle that has been counted with all the other ones will be sorted out of the list.
        \item For each particle $i=1\ldots \eta$ search over its $(2b+1)^3-1$ neighbors in the list and if $\ket{B_j}$ is in the cube create a fictitious charge of $Q_i$ at field location $B_j$ of the form $\ket{B_j}\ket{Q_i}$.
        \item Sort each of the $((2b+1)^3-1)\eta$  field / charge site $\ket{B_i}\ket{Q_i}$ by their cell number $B_i$ using quantum merge sort.
        \item For $i$ in decreasing order, if $\ket{B_i} = \ket{B_{i-1}}$ then add $M$ to $B_i$ and then merge sort again the resulting array by $B_i$ to remove the duplicates of the fictitious charges.
        \item For each of the $((2b+1)^3-1)\nu$ fictitious charges search merge sort the arrays based on coordinate $\mu$ of $B_i$ and compute sum of the charges along each of the closest $(2b+1)^3\eta$ that are within an $L_1$ distance $b$ of $B_i$. 
        \item Sum the results of the above steps of each of the $((2b+1)^3-1)\nu$ locations and add the result together in a new register $\ket{Q_{tot,i}}$
        \item For each field cell, $q$, compute the divergence to prepare a state of the form
        $$\ket{\vec{E}} \ket{0} \ket{\mathcal{D}_0}\cdots \ket{M-1}\ket{\mathcal{D}_{M-1}},$$
        where each $\mathcal{D}_q:={\rm round}\left(\frac{h^2}{8\pi}\sum_{\mu=0}^2 \sum_{u,v \in \mathcal{S}_{q,\mu;b}}\beta_{u,v} I\otimes(E_{q+be_\mu,\mu} - E_{q-be_\mu, \mu} )\right)$ is the electric flux computed for the cube of length $(2b+1)h$ centered at cell $q$ and $\beta_{\mu,\nu}$ are the coefficients of the quadrature formula used to integrate the flux.
        \item Use quantum merge sort to sort the $\ket{q}\ket{\mathcal{D}_q}$ vectors by $q(1-\delta_{\mathcal{D}_q ,0}) - \delta_{\mathcal{D}_q,0}$ which corresponds to $q$ if $\mathcal{D}_q\ne 0$ and $-1$ otherwise.  This yields for any computational basis state of position and field (after an irrelevant permutation) a state of the form
        $$
\ket{\vec{E}} \ket{g_1} \ket{\mathcal{D}_{g_1}} \ket{x_1} \ket{B_{f_1}}\cdots \ket{g_\nu}\ket{\mathcal{D}_{g_\nu}} \ket{x_\eta} \ket{B_{f_\eta}} \left(\ket{g_{\eta+1}}\ket{\mathcal{D}_{g_{\eta+1}}}\ket{g_{\eta+2}}\ket{\mathcal{D}_{g_{\eta+2}}}\cdots\right)\ket{junk_2}
        ,$$
        for an ancillary state $\ket{junk_2}$ that stores the extra qubits needed to make the sorting processes reversible.
        \item For each $\ket{B_{f_i}}$ and compare the integer to $\ket{g_i}$ and if equal add the register $Q_{f_i}$ the register containing the corresponding divergence $\ket{Q_{f_i}}\ket{\mathcal{D}_{g_i}}\rightarrow \ket{Q_{f_i}}\ket{\mathcal{D}_{g_i}-Q_{f_i}}$.
        \item Implement the isometry $$\ket{D_{g_1}}\cdots \ket{D_{g_\eta}} \rightarrow \begin{cases}
            \ket{D_{g_1}}\cdots \ket{D_{g_M}}\ket{0} & \text{if $\mathcal{D}_{g_i} = 0~\forall~i$}\\
            \ket{D_{g_1}}\cdots \ket{D_{g_M}}\ket{1} & \text{otherwise}
        \end{cases}
        $$
        This final qubit computes ${{\rm rect}\left(\sum_q{\frac{\Delta^2}{2}\sum_{\mu=0}^2 I\otimes(E_{q+e_\mu,\mu} - E_{q-e_\mu, \mu} ) - 4\pi \sum_{i=0}^{\eta-1} {\rm sign}(\zeta_i)\Pi_{q,i}} \right)}$.
    \end{enumerate}

    Let us consider step 1.  As the particles are in a grid, we simply have to divide their positions be $\Delta$ to identify their grid locations and store them as an $O(\log(M))$-bit value.  Division requires $O(\log^2(N))$ operations using gradeschool multiplication and this needs to be repeated $\eta$ times and so the cost of this is 
    \begin{equation}
        \mathcal{C}_1 = O(\eta \log^2(N))\label{eq:Gfirst}
    \end{equation} under the assumption that $M\le N$. 

    Step 2 requires the use of quantum Merge Sort.  The computation of the sorting function here is sub-dominant to the cost of the merge sort and can be neglected as it requires simply an $O(\log(M))$ bit comparison and a controlled swap or incrementer for each entry.  There are a total of $L=\eta$ registers of length $\beta=O(\log(MN\eta))$ which incurs a number of $T$ gates from~\eqref{eq:costMerge} that are in
    \begin{eqnarray}
        {\mathcal{C}}_{2} = O\left(\eta\log(\eta)\log(\eta NM) \right)
    \end{eqnarray}

    Step 3 requires that we compare the bin numbers of each of the $\eta-1$ registers and that adjacent to it.  The bin numbers are $O(\log(M))$ bit numbers.  From~\cite{2018_G} we have that the cost of each comparison requires $O(\log(M))$ $T$ gates.  Conditioned on this we need to perform an increment and decrement on the register that encodes the flux, which is $O(\log(\Lambda))$ bits.  Thus the cost is
    \begin{eqnarray}
        \mathcal{C}_3 = O(\eta \log(M\Lambda)).
    \end{eqnarray}

Step 4 requires another instance of quantum merge sort.  In this case there are $M$ registers consisting of $O(\log(N\eta))$ bits.  As before, the computation of the sorting function 
$B_{f_i}(1-\delta_{Q_i ,0}) - \delta_{\mathcal{D}_q,0}$ requires a conditional swap based on an $O(\log(\eta))$ comparison and is sub-dominant to the overall cost of the merge sort.
The cost of the merge sort step is then
\begin{equation}
    \mathcal{C}_4 = O(\eta \log(\eta)\log(N\eta)).
\end{equation}

Step 5 involves looping over $O(b^3\eta)$ and copying the value of the charge into each of these sites.  This requires $O((b+1)^3 \log(\eta))$ operations as each charge register requires $O(\log(\eta))$ qubits to represent the value stored within.  Thus the cost for this step is
\begin{equation}
    \mathcal{C}_5 = O((b+1)^3\eta \log(\eta))
\end{equation}

Step 6 is an application of merge sort.  The combined numbers that we are sorting consist of sorting $O((b+1)^3\eta$ bit strings of length $O(\log(M\eta(b+1)))$ qubits (note here and in the following we do not need the particle positions and so we do not need to sort based on  the $O(\log(N))$ qubit register that stores the positions $\ket{x_i}$).  Thus the overall cost is
\begin{equation}
    \mathcal{C}_6 = O((b+1)^3\eta\log((b+1)\eta)\log(M\eta(b+1))
\end{equation}

Step 7 involves going through each of $O((b+1)^3 \eta)$ elements and comparing a $\log(M)$ qubit number from each and then conditioned on the result add $M$ to the previous entry which requires $O((b+1)^3\eta \log(M))$ operations.
\begin{equation}
    \mathcal{C}_7 = O((b+1)^3\eta\log(M(b+1)\eta)\log(M\eta(b+1))
\end{equation}

Step $8$ Requires $3$ applications of merge sort on bit string of length $O(\log(NM(b+1)\eta))$ qubits.  Then after this we need to go through the list and compute the $L_1$ distance between each element within $O((b+1)^3)$ of the site.  The computation of the $L_1$ distance requires adding $3$ numbers of length $O(\log(M))$ qubits.  This requires $O(\log(M))$ operations per site.  Then we need to perform a controlled adder onto an $O(\log(\eta(b+1)))$ qubit register.  This requires $O(\log(\eta(b+1)))$ per site.  Summing these two costs yields
\begin{equation}
    \mathcal{C}_9+\mathcal{C}_8 = O((b+1)^3\eta\log(M(b+1)\eta)\log(NM\eta(b+1))
\end{equation}

Step 10 requires us to compute the divergence through a cubic cell consisting of $O(b^2)$ cells that are on the boundary for each of the cells in question.  The arithmetic needed to compute this involves multiplication of $O(1)$ numbers which costs $O(\log^2(\Lambda))$ gates and in total requires $O(Mb^2 \log^2(\Lambda))$ operations.  The cost of summing these values is sub-dominiant to this.  Thus the overall complexity is
\begin{equation}
    \mathcal{C}_{10}\in O(Mb^2 \log^2(\Lambda)).
\end{equation}

Step 11 Requires another application of Merge Sort on a register of length $O(\log(\eta M N \Lambda))$.  The cost of computing the values that we are sorting on is $O(\eta\log^2(\Lambda))$.  Thus the overall cost of the the step is 
\begin{equation}
    \mathcal{C}_{11} \in O(\eta\log(\eta)\log(MN\eta\Lambda) + \eta\log^2(\Lambda)).
\end{equation}

Step $12$ requires $O(M)$ comparisons and subtractions on qubit strings of size $O(\log(\Lambda))$.  Similarly step $13$ requires $O(M)$ comparisons and then a multiply controlled NOT gate between them.  These two costs are asymptotically identical and so \begin{equation}
    \mathcal{C}_{12} + \mathcal{C}_{13} \in O(M\log(\Lambda))).\label{eq:Glast}
\end{equation}

Summing the gate counts given in~\eqref{eq:Gfirst} to \eqref{eq:Glast} we find that
\begin{align}
    \sum_i \mathcal{C}_i \in \widetilde{O}\left(b^3(\eta+M)\log^2(N\Lambda) \right)
\end{align}
Note that for the Toffoli circuits considered in this process, that the number of Clifford gates needed in the implementation of the Toffoli is proportional to the number of $T$ gates and so this scaling is also valid for the number of Clifford operations.

Finally we need to apply a rotation that returns the eigenvalue $e^{-i \lambda t}$ conditioned on the eigenvalue of the rect function being zero.  This corresponds to being in the $+1$ eigenspace of the projector $H_c$.  This can be achieved using a single qubit rotation on the qubit controls the output conditioned on the bitstring that encodes the specific time duration desired.  
    
The $T$-count of implementing the controlled operation $c-R_z(k \lambda T_{\max}/2^{n_t})$ within error $\epsilon/n_t$ is $O(\log(n_t/\epsilon))$ from~\cite{2015_KMM} irrespective of the value of $\lambda$ and $t$. 
 From Box 4.1 of Ref.~\cite{2010_NC} we know that errors grow sub-additively.  Thus as $n_t$ rotations are needed we then have that the over all error is $\epsilon$ as required.  This process needs to be repeated $n_t$ times and hence the overall cost of the controlled rotations using the synthesis method of Ref.~\cite{2015_KMM} is 
 \begin{equation}
     \mathcal{C}_{rot} \in \widetilde{O}(n_t \log(1/\epsilon))
 \end{equation} 
    We then need to uncompute the previous steps, which can be done using Clifford operations.
    Thus the over all $T$-count is
    \begin{equation}
        \mathcal{C}(e^{-i \lambda H_c t}) \in O\left(\sum_i \mathcal{C}_i +\mathcal{C}_{rot}\right) \subseteq \widetilde{O}(b^3(\eta+M)\log^2(N\Lambda) + n_t\log(1/\epsilon)).
    \end{equation}

    Next we need to consider the cost of implementing the evolution $\sum_k \ketbra{k}{k}\otimes e^{-i (\sum_{q,\mu} \id \otimes E_{q\mu}^2/2) \frac{k T_{\max}}{L}}$.  As this is again a diagonal operator we can simulate it by computing the eigenvalues and using a series of $O(n_t)$ controlled single qubit rotations to apply each eigenphase to each eigenvector. We need to sum over each of the $M$ cells and compute the square of the field is $O(n^2)$.  Dividing the number by $2$ requires no $T$ gates as the bit shift can be implemented by swapping the bits right, which requires only $O(n)$ Clifford operations. 
 Thus the total number of gate operations needed to compute the each eigenvalue is $O(M\log^2(\Lambda))$.   We then apply the eigenphase to each of the eigenvectors using $n_t$ controlled single qubit rotations.  Following the same argument we find that the cost is sub-dominant to the cost of computing the constraint.  Thus we have that after applying loose upper bounds on the logarithmic factors that the number of Clifford and non-Clifford operations scales at most as
    \begin{eqnarray}
        \mathcal{C}(V_{\rm int}) = \widetilde{O}( b^3(\eta+M)\log^2(N\Lambda) + n_t\log(1/\epsilon))
    \end{eqnarray}

Finally, as we used the replica trick here in this version of the algorithm the total number of fictitious particles considered above is not equal to the total number of particles in the original system.  This corresponds to $\eta \mapsto \sum_{i} |\zeta_i|$, which gives us
\begin{eqnarray}
        \mathcal{C}(V_{\rm int}) = \widetilde{O}( b^3(M+\sum_i |\zeta_i|)\log^2(N\Lambda) + n_t\log(1/\epsilon))
    \end{eqnarray}
\end{proof}
Note that above the use of the replica trick leads to scaling with the total charge in the system that is likely to be suboptimal.  Specifically, improved arguments could be used to show that by sorting the particles independently by the values of the charge could be used to achieve the same scaling at cost that is independent of the number of charges; however, the above argument is much simpler and in typical cases the difference between the two will be at most a constant factor as nuclear charge in chemical applications is bounded above for all practical purposes by  for that reason we eschew the more complica
Next let us go forward and discuss the block encoding of the magnetic field operator.

\begin{proof}[Proof of Lemma~\ref{lem:HpiBE}]
    First we have that
    \begin{equation}
        H_{f2}=\frac{c^2 h}{8 \pi}\sum_{q=0}^{M-1}  \openone\otimes\left(\sum_{i,j,k} \sum_{p=-a}^a d'_{2a+1} A_{q + p e_j, k} \varepsilon_{ijk}\right)^2
    \end{equation}
    From  Lemma~\ref{cor:A_lcu} we obtain that an LCU decomposition of each $A_{q,\mu}$ is
    \begin{eqnarray}
  A_{q,\mu}=\frac{2\pi}{\Lambda h}\left(\frac{\Lambda-1}{2}\openone-\mathcal{F}_{q,\mu}\left(\sum_{i=0}^{\log(2\Lambda)-1}2^{i-1}\Z_{(i+1)}\right)\mathcal{F}^{\dagger}_{q,\mu}\right) \nonumber
 \end{eqnarray}
 and the block-encoding one-norm of this is $\|A_{q,\mu}\|_{\ell_1} \le \frac{2\pi}{h}$.
 First let $U_{A_{q,\mu}}$ be an $O(2\pi/h,\cdot,0)$ block encoding of $A_{q,\mu}$.  It then follows that we can construct a block-encoding for $H_{f2}$ out of the sums of these products of block encodings of $A$. 
 This products of block encodings can be implemented using Lemma 53 of~\cite{2019_GSLW} which states that if each $U_{A_{q,\mu}}$ is an $(\alpha,b,0)$  block encoding of $A_{q,\mu}$ then we can construct a block encoding of the form 
 \begin{equation}\label{eq:UAA}U_{A_{q,\mu}A_{q,\mu'}}:=(\openone_{\mu'} \otimes U_{A_{q,\mu}})(\openone_{\mu} \otimes U_{A_{q',\mu'}})\end{equation} that is an $(\alpha^2,2b,0)$ block-encoding of the product, where $\openone_{\mu},\openone_{\mu'}$ correspond to identities on the ancillae spaces of the other block encodings.  This spaces much be disjoint in order to prevent cross terms from the multiplication to ruin the block encoding.  

Given this notation, we can express
 \begin{equation}
     H_{f2}:=\sum_{q,p,p',i,j,k,i',j',k'} \alpha_{p,p',i,j,k,i',j',k'}(U_{A_{q+pe_j,k}A_{q+p'e_{j'},k'}}).
 \end{equation} As the block-encodings of $A_{q,\mu}$ each have normalization $\alpha=(2\pi/h)$, the normalization for each product is then $4\pi^2/h^2$. We can see that the weights in the linear combinations of unitaries needed to block-encode $H_{f2}$ can be written as

 \begin{equation}
     \alpha_{p,p'
    ,i,j,k,i',j',k'}=\left(\frac{4\pi^2}{h^2} \right)\left(\frac{c^2hd'_{2a+1,p} d'_{2a'+1,p'}\epsilon_{ijk}\epsilon_{i'j'k'}}{8\pi}\right).
 \end{equation}
 for cell numbers $q$ with field directions $ijk,i'j'k'$ and displacements $p,p'$.
 Further the normalization constant is then from Lemma~\ref{lem:d'}
 \begin{equation}
     \alpha_{f2}:=\sum_{q,p,p'\ldots} \left|\left(\frac{4\pi^2}{h^2} \right)\left(\frac{c^2hd'_{2a+1,p} d'_{2a+1,p'}\epsilon_{ijk}\epsilon_{i'j'k'}}{8\pi}\right) \right|\in O\left(\frac{Mc^2\log^2(a)}{h} \right)
 \end{equation}
 This validates that we can achieve a $(Mc^2\log^2(a)/h,\cdot,\epsilon)$ block encoding of $H_{f2}$ provided that we can $\epsilon$-approximate the coefficients and the operators in the result.

 The $\prep$ circuit for this creates a superposition over all of the non-trivially varying coefficients.  If we use brute force synthesis for this the cost depends on the number of non-constant coefficients in the expansion.  There are $O(a^2)$ such constants because $ijk$ all take values in $\{0,1,2\}$.  Thus the cost of performing the state
 \begin{equation}
 \prep\ket{0}=\sum_{p,p',ijk,i'j'k'}\sqrt{\frac{\left|\left(\frac{4\pi^2}{h^2} \right)\left(\frac{c^2hd'_{2a+1,p} d'_{2a+1,p'}}{8\pi}\right) \right|}{2(2a+1)^2 3^6\alpha_{f2}}} \ket{p,p',ijk,i'j'k'}   H^{\otimes \log(M)}\ket{0}_{q}H\ket{0}_{|\epsilon|=0} 
 \end{equation}
 Here the normalization follows because there are at most $3$ possible values for each $ijk$, and $|\epsilon_{ijk}|\le 1$.  The last qubit is introduced as a flag that we will apply a phase to for any $\epsilon_{ijk}\epsilon_{i'j'k'}$ term that evaluates to zero.  This flag can be removed if we alter the state preparation so that the prepare state has no support on any of the zero-valued combinations of directions of the directions $ijk,i'j'k'$ for the two respective cross products in the magnetic field term.  Note that here we have assumed that $M$ is a power of $2$ so that the Hadamard transform suffices to prepare a uniform mixture over all configurations.

 The cost for the $\prep$ circuit here is simply the cost of performing an arbitrary state preparation over the terms above plus the cost of performing the superposition over the ancillary registers.  The Hadamard transformations that enact the latter have no non-Clifford gates; whereas the preparation of the remaining state requires $O(a^2)$ unique amplitudes which in turn necessitates a cost of $\widetilde{O}(a^2\log(1/\epsilon))$ to prepare within Euclidean error $\epsilon$.  Thus the total cost of $\prep$, up to sub-dominant logarithmic factors, is
 \begin{equation}
     \mathcal{C}(\prep) \in \widetilde{O}\left(a^2\log(1/\epsilon) \right)
 \end{equation}

  Now turning our attention to $\sel$, we will first examine implementing $U_{A_{q,\mu} A_{q',\mu'}}$.  This cost is simply the sum of the costs of implementing the two block-encodings.  The non-Clifford count for implementing $U_{A_{q,\mu}}$ is, if we use the approximate quantum Fourier transform of~\cite{cleve2000fast} is $\widetilde{O}(\log(\Lambda)\log(1/\epsilon))$.  Thus the cost for implementing the block encoding of the product $U_{A_{q,\mu}A_{q',\mu'}}$ is also $\widetilde{O}(\log(\Lambda)\log(1/\epsilon))$ from~\eqref{eq:UAA}.  As there are $M$ distinct sites that this operator needs to be applied at, the cost for implementing these operations controllably at every possible pair of states that can emerge in $\prep\ket{0}$ is $\widetilde{O}(a^2\log(\Lambda)\log(1/\epsilon))$
  because there are $O(a^2)$ terms that we need to apply this for and the additional error tolerance needed for the synthesis only leads to a sub-dominant logarithmic cost that vanishes in $\tilde{O}(\cdot)$.  The control logic needed to trigger each of the $O(Ma^2)$ unitaries requires $O(\log(Ma^2))$ gates, thus the cost of the controls in the prepare circuit for the unitaries is 
  \begin{equation}
      \widetilde{O}( Ma^2 \log(Ma^2) + Ma^2 \log(\Lambda)\log(1/\epsilon)) = \widetilde{O}(Ma^2 (\log(\Lambda)\log(1/\epsilon))).
  \end{equation}
  Next we need to worry about whether the term is zero valued because our prepare circuit does not explicitly rule out $\epsilon_{ijk}=0$ for example.  The Levi-Cevita symbol is zero if and only if $ijk$ is not a cyclic permutation of $012$  As there are a finite number of combinations, the cost of computing this with a reversible circuit is in $O(1)$.  Similarly, as $\varepsilon_{ijk}\varepsilon_{i'j'k'}$ is zero if and only if the bitwise-or of the two functions is zero, we can compute this boolean function using $O(1)$ gate operations as well.  Thus we can apply an operator of the form
  \begin{equation}
      U_{\varepsilon=0}:=\Pi_{\varepsilon_{ijk}\varepsilon_{i'j'k'}=0} \otimes Z_{|\varepsilon|=0}+ \openone\otimes \openone,
  \end{equation}
  using $O(1)$ operations.

  Next let us consider a unitary $U_{{\rm sign}(d\varepsilon)}$ that which computes the sign of an individual term in the expansion.  This is needed with the standard form linear combination of unitary encoding as the amplitudes square to eliminate the phase here.  The sign is the parity of the sign values of $d'_{2a+1,p}$, $d'_{2a+1,p'}$ and those of the Levi-Cevita symbols.  The former are easy, we simply need to apply $(-1)^{p+p'}$ to compute the sign of the product of the two using Lemma~\ref{lem:lcuNabla}.  This can be implemented using the tensor product of Pauli-Z operators acting on the least significant bits of both registers: $Z_{p,lsb}\otimes Z_{p,lsb}$.  The pattern for the Levi-Cevita symbol is a little harder in this case; however, a unitary operator exists that can compute the value using $O(1)$ gates and can be exactly synthesized usign $H$ and Toffoli.  Thus this contribution to the sign can be computed using $O(1)$ gates.  Therefore the product of the two, $U_{{\rm sign}(d\varepsilon)}$ can be implemented using $O(1)$ gate operations.  Thus for each of the inputs, the $\sel$ operation needs to apply an operator of the form $U_{A_{q,\mu} A_{q',\mu'}} U_{{\rm sign}(\varepsilon)} U_{_{\varepsilon=0}}$ and the cost of implementing this operator is
  \begin{equation}
      \mathcal{C}(U_{A_{q,\mu} A_{q',\mu'}} U_{{\rm sign}(\varepsilon)} U_{_{\varepsilon=0}})= \mathcal{C}(U_{A_{q,\mu} A_{q',\mu'}}) + O(1).
  \end{equation}
  The cost of implementing $U_{A_{q,\mu} A_{q',\mu'}}$ is the cost of implementing two block-encodings of $A$ from~\cite{2019_GSLW}, which is $\widetilde{O}(\log(\Lambda)\log(1/\epsilon))$.  Thus we have that
  \begin{equation}
      \mathcal{C}(U_{A_{q,\mu} A_{q',\mu'}} U_{{\rm sign}(\varepsilon)} U_{_{\varepsilon=0}})= \widetilde{O}(\log(\Lambda)\log(1/\epsilon)).
  \end{equation}
  There are $O(Ma^2)$ such unitaries, as argued above, and thus the $\sel$ operation needs to control each of these operations and also perform the Toffoli logic (which requires $O(\log(Ma^2))$ operations) to trigger it.  The Toffoli gates are sub-dominant to the main scaling and thus the overall complexity of $\sel$ is
  \begin{equation}
      \mathcal{C}(\sel) = \widetilde{O}(Ma^2\log(\Lambda)\log(1/\epsilon)).
  \end{equation}  
  Thus by adding the cost of $\prep$ to that of $\sel$ we find that the cost of block-encoding $H_{f2}$ within error $\epsilon$ is 
  \begin{equation}
      \mathcal{C}_{f2}= \widetilde{O}(Ma^2\log(\Lambda)\log(1/\epsilon)).
  \end{equation}
\end{proof}


\section{Proof of Lemma~\ref{lem:2DCotes}} \label{app:newtonCotesProof}
\begin{proof}[Proof of Lemma~\ref{lem:2DCotes}]
The surace integral in question for the computation of Gauss' law is over the surface of a cube.  This integral breaks down into $6$ integrals over each of the faces of the cube.  For this reason, it suffices for us to consider the integral over one of the faces of the cube and recognize that the error in the discrete approximation to the sum is going to be $6$ times the error from one of the faces.  Without loss of generality then, let us consider the error that we find when approximating the surface integral at $z= bh$.  In this case the surface integral becomes
\begin{equation}
    \int_{(-b-1/2)h}^{(b+1/2)h} \int_{(-b-1/2)h}^{(b+1/2)h}E_z([x,y,bh]) dx dy \approx \int_{(-b-1/2)h}^{(b+1/2)h} \sum_{q=-b}^b \gamma_q E_z([qh,y,bh])  dy  
\end{equation}
where $\gamma_q$ are the Coefficients for the $2b+1$ point Newton-Cotes formula.  This further allows us to approximate the double integral by a double sum:
\begin{equation}
    \int_{(-b-1/2)h}^{(b+1/2)h} \int_{(-b-1/2)h}^{(b+1/2)h}E_z([x,y,bh]) dx dy \approx \sum_{q,q'} \gamma_q \gamma_{q'} E[qh,q'h,bh].
\end{equation}
The error in such formulas is typically expressed for one-dimensional integrals.  This means that while the first layer of approximations have easily discoverable bounds the second does not.  However, an application of the triangle inequality allows us to bound the error in terms of the error s in each of the successive approximations.  Specifically
\begin{align}
    &\left| \int_{(-b-1/2)h}^{(b+1/2)h} \int_{(-b-1/2)h}^{(b+1/2)h}E_z([x,y,bh]) dx dy -\sum_{q,q'} \gamma_q \gamma_{q'} E[qh,q'h,bh]. \right| \nonumber\\
    & \qquad \le \left| \int_{(-b-1/2)h}^{(b+1/2)h} \int_{(-b-1/2)h}^{(b+1/2)h}E_z([x,y,bh]) dx dy - \int_{(-b-1/2)h}^{(b+1/2)h} \sum_{q=-b}^b \gamma_q E_z([qh,y,bh])  dy \right| \nonumber\\
    &\qquad\quad+ \left| \int_{(-b-1/2)h}^{(b+1/2)h} \sum_{q=-b}^b \gamma_q E_z([qh,y,bh])  dy - \sum_{q,q'} \gamma_q \gamma_{q'} E[qh,q'h,bh]\right|
\end{align}
From this, we can then further use the triangle inneqaulity to show that
\begin{align}
    &\left| \int_{(-b-1/2)h}^{(b+1/2)h} \sum_{q=-b}^b \gamma_q E_z([qh,y,bh])  dy - \sum_{q,q'} \gamma_q \gamma_{q'} E[qh,q'h,bh]\right| \nonumber\\
    &\quad\le \sum_q |\gamma_q| \max_q\left| \int_{(-b-1/2)h}^{(b+1/2)h} E_z([qh,y,bh])  dy - \sum_{q'} \gamma_{q'} E[qh,q'h,bh]\right|
\end{align}
and similarly
\begin{align}
     &\left| \int_{(-b-1/2)h}^{(b+1/2)h} \int_{(-b-1/2)h}^{(b+1/2)h}E_z([x,y,bh]) dx dy - \int_{(-b-1/2)h}^{(b+1/2)h} \sum_{q=-b}^b \gamma_q E_z([qh,y,bh])  dy \right| \nonumber\\
     &\qquad\le(2b+1)h \max_y \left|  \int_{(-b-1/2)h}^{(b+1/2)h}E_z([x,y,bh]) dx  -  \sum_{q=-b}^b \gamma_q E_z([qh,y,bh])  dy \right|
\end{align}
We have from~\cite{simon2024amplified} that the error in the integrals is bounded above by $\epsilon_0h$ for a value of $b$ chosen via
\begin{equation}
    b\in \widetilde{O}\left(\log( \max_{z\in\mathcal{C}} \|\nabla E(z)\|/\epsilon_0h) \right)
\end{equation}
where $\max_{z\in \mathcal{C}} \|\nabla E(z)\|$ 
is the maximum be the spatial derivative on the complex contour $\mathcal{C}=\{z : z=1+3e^{i\phi}+3-e^{-i\phi}\}$ for 
$\phi\in[0,2\pi)$ assuming that $\partial_z E(z)$ is analytic on 
  $[-7/3,13/3]$.
Further, we have that
\begin{equation}
    \sum_q \gamma_q \in O(bh).
\end{equation}
\begin{align}
    \left| \int_{(-b-1/2)h}^{(b+1/2)h} \int_{(-b-1/2)h}^{(b+1/2)h}E_z([x,y,bh]) dx dy -\sum_{q,q'} \gamma_q \gamma_{q'} E[qh,q'h,bh]. \right| \in O(bh \epsilon_0).
\end{align}
Thus it suffices to take $\epsilon_0 \in \Theta (\epsilon/(b h))$.  As a result we have from solving for $b$ in the resulting expression that
\begin{equation}
    b\in \tilde{O}\left(\log\left( h\max_{z\in\mathcal{C}} \|\nabla E(z)\|/\epsilon \right) \right) 
\end{equation}
\end{proof}

\section{Proof of Prop~\ref{prop:Maxwell}}\label{app:Maxwell}

\begin{proof}[Proof of Prop~\ref{prop:Maxwell}]
  We will justify the contraction by performing a series expansion for $A$ which is based on the fact that $U_{q,\mu} = e^{-i A_{q,\mu}E_{\max}/\Lambda}$:
\begin{equation}\label{eq:Aqmu}
    A_{q,\mu} = \frac{ -i\Lambda}{E_{\max}}(1- U_{q,\mu}) +\mathcal{E},
\end{equation}
where for all valid states $\bra{\psi} \mathcal{E} \ket{\psi} \in O(L/\Lambda)$ for some value of $L$ which we take to be in $o(\Lambda)$.
This is appropriate when the quantum state is a smoothly varying function of $E$ at each $(q,\mu)$ in the field grid.  
As before we will consider the energy cost of introducing a contractable loop and show that this cost diverges as $c\rightarrow \infty$ under these assumptions.

Let us consider the magnetic field operator at an arbitrary  in the Heisenberg picture \begin{equation}
    \partial B_{p,x}(t)/\partial t = -i[H,B_{p,x}(t)] = -i[H, (\nabla \times A)_{p,x}(t)]
\end{equation}
Using the central difference formula for the derivative we have that
\begin{equation}
(\nabla \times A)_{p,x} = \frac{\partial}{\partial y}A_{pz} -\frac{\partial}{\partial z} A_{py}\rightarrow \frac{1}{2h}\left(A_{p+e_y,z} - A_{p-e_y,z} + A_{p+e_z,y} -A_{p-e_z,y}\right).
\end{equation}
We can see that the points above can be viewed as an oriented path  in the $y-z$ plane.  This expression can be modified for our purpose in the following manner.  We will replace each point in the central difference formula by the average of its preceding point in a cycle around the point $p$.  This gives us the following discretization to the curl, which is accurate to at least $O(h)$ under appropriate continuity assumptions.

\begin{align}
    (\nabla \times A)_{p,x}\rightarrow \frac{1}{4h}&\left(A_{p+e_y,z}+A_{p+e_y-e_z,z} - A_{p-e_y,z} -A_{p-e_y+e_z,z}\right. \nonumber\\
    &\left.+A_{p+e_z,y}+A_{p+e_x-e_y,y}-A_{p-e_z,y}-A_{p-e_z+e_y,y}\right)
\end{align}

We then have 
\begin{align}
    [H,(\nabla\times A)_{p,x}(t)] = e^{iHt}[\sum_{q }E^2_q, (\nabla \times A)_{p,x}]e^{-iHt}=e^{iHt}\left(\sum_{q}(E_q [E_q,(\nabla\times A)_{p_x}]+[E_q,(\nabla\times A)_{p_x}]E_q)\right)e^{-iHt}
\end{align}
using \eqref{eq:Aqmu} we can see that
\begin{align}
    [E_{q,\nu},A_{q,\nu}] = \frac{i\Lambda}{E_{\max}} [E_{q,\nu},U_{q,\nu}]
\end{align}
Thus if we assume that our states have no support over the cutoff then we do not have to worry about wrap around effects when computing the commutator and

\begin{equation}[E_{q,x},U_{q,x}]=\frac{E_{\max}}{\Lambda} \sum_{j} \ketbra{j}{j-1}_{q,x}= \frac{E_{\max}}{\Lambda} U_{q,x}
\end{equation}
This shows that for states that have no support on the cutoff,
\begin{equation}
    \frac{\partial B_{p,x}(t)}{\partial t} =e^{i Ht}\left(\Box_{p,x}+\Box_{p,x}'\right) e^{-iHt}
\end{equation}
where $\Box_{p,x}$ is an electric field loop term combined with field incrementer operations of the form
\begin{align}
    \Box_{p,x} = \frac{1}{4h}&\left(E_{p+e_y,z}U_{p+e_y,z}+E_{p+e_y-e_z,z}U_{p+e_y-e_z,z} - E_{p-e_y,z}U_{p-e_y,z} -E_{p-e_y+e_z,z}U_{p-e_y+e_z,z}\right. \nonumber\\
&\left.+E_{p+e_z,y}U_{p+e_z,y}+E_{p+e_x-e_y,y}U_{p+e_x-e_y,y}-E_{p-e_z,y}U_{p-e_z,y}-E_{p-e_z+e_y,y}U_{p-e_z+e_y,y}\right)
\end{align}
and $\Box'_{p,x}$ is a similar operation with the order of $E$ and $U$ reversed for each term
\begin{align}
    \Box_{p,x}' = \frac{1}{4h}&\left(U_{p+e_y,z}E_{p+e_y,z}+U_{p+e_y-e_z,z}E_{p+e_y-e_z,z} - U_{p-e_y,z}E_{p-e_y,z} -U_{p-e_y+e_z,z}E_{p-e_y+e_z,z}\right. \nonumber\\
&\left.+U_{p+e_z,y}E_{p+e_z,y}+U_{p+e_x-e_y,y}E_{p+e_x-e_y,y}-U_{p-e_z,y}E_{p-e_z,y}-U_{p-e_z+e_y,y}E_{p-e_z+e_y,y}\right)
\end{align}
If we assume that our states of interest obey $\|\ket{\psi(t)} - U_{p,x}\ket{\psi(t)}\| \in O(L/\Lambda)$
then we have that up to error that vanishes as $\Lambda\rightarrow \infty$, both $\Box_{p,x}$ and $\Box_{p,x}'$ act as a contractable electric field loop operator.

Thus for such states we have that 
\begin{equation}
    \lim_{\Lambda \rightarrow \infty} \bra{\psi}\frac{\partial B_{p,x}(t)}{\partial t}  \ket{\psi}=0,~\forall~t
\end{equation}
if $\ket{\psi}$ does not have a contractable field loop in it meaning that it is not in the kernel of $\Box_{p,x}$.  Similarly, if we have that the derivative is zero for all $t$ and all $\ket{\psi}$ inside the valid space of states then we must have that $\Box_{p,x}$ acts as the zero operator on these states and hence these states must be in its kernel.

Repeating this argument for the $x,y$ and $x,z$ plane leads us to the conclusion that at any $p$
\begin{equation}
    \lim_{\Lambda\rightarrow \infty}\bra{\psi} \frac{\partial B_{p,\nu}(t)}{\partial t} \ket{\psi} = \lim_{\Lambda\rightarrow \infty} \bra{\psi} e^{iHt}\Box_{p,\nu}e^{-iHt} \ket{\psi}, 
\end{equation}
which gives the Farraday-Maxwell relation between field and magnetic flux for an elementary contractable loop. 

Any contractable loop in the plane can be written as a linear combination of these elementary loops.  As the expected derivative of the magnetic field is zero for all contractable such loops it is also zero for their sum.  Thus the property holds for all contractable loops in the $z,y$ plane which then gives us the integral version by summing over all such decompositions as claimed.  Using the slightly abused notation of treating the discrete sums as line and surface integrals
yields
    \begin{equation}
    \lim_{\Lambda\rightarrow \infty} \bra{\psi} \oiint_C \frac{\partial B(t)}{\partial t}\cdot dA \ket{\psi} = -\lim_{\Lambda \rightarrow \infty} \bra{\psi(t)} \oint_{C}E(q)\cdot dq \ket{\psi(t)}
    \end{equation}
\end{proof}
where the negative sign is present because the definition of $E(q)$ yields a negatively oriented loop rather than the conventional positively oriented loop.

\end{document}